\documentclass[twoside,11pt]{article}

\input{jair_plus_hyperref}

\usepackage{amssymb,amsthm,amsmath}
\usepackage[noend]{algpseudocode}
\usepackage{algorithm}
\usepackage{microtype,xcolor,graphicx,booktabs,subcaption}
\usepackage[all]{xy}
\usepackage{tikz}
\usetikzlibrary{positioning,graphs,graphs.standard}
\usepackage{wrapfig}

\renewcommand*{\le}{\leqslant}
\renewcommand*{\ge}{\geqslant}
\renewcommand*{\epsilon}{\varepsilon}
\renewcommand*{\emptyset}{\varnothing}
\newcommand{\forced}{\setlength{\fboxsep}{0pt}\mbox{${\hspace{5pt}^\to}\hspace{-14pt}\colorbox{white}{$C$\hspace{0.5pt}}\:\:$}}
\newcommand{\free}{\setlength{\fboxsep}{0.15pt}\mbox{${\hspace{5pt}^\rightrightarrows}\hspace{-14pt}\colorbox{white}{$C$\hspace{0.5pt}}\:\:$}}

\newcommand*{\citet}[1]{\mshortciteA{#1}}
\renewcommand{\cite}[1]{\mshortcite{#1}}
\newcommand*{\citealp}[1]{\mshortciteR{#1}}

\newcommand{\vecs}{\mathbf{s}}
\newcommand{\pos}{\mathrm{pos}}
\newcommand{\score}{\mathrm{score}}
\newcommand{\cost}{\mathrm{cost}}

\newtheorem{theorem}{Theorem}[section]
\newtheorem{corollary}[theorem]{Corollary}
\newtheorem{lemma}[theorem]{Lemma}
\newtheorem{proposition}[theorem]{Proposition}
\newtheorem{definition}[theorem]{Definition}
\newtheorem{example}[theorem]{Example}

\jairheading{73}{2022}{231-276}{08/2020}{01/2022}
\ShortHeadings{Preferences Single-Peaked on a Tree}
{Peters, Yu, Chan, \& Elkind}
\firstpageno{231}

\begin{document}

\title{Preferences Single-Peaked on a Tree: \\ Multiwinner Elections and Structural Results
}

\author{\name Dominik Peters \email dominik@cs.toronto.edu \\
		\addr Department of Computer Science \\
		University of Toronto, Toronto, ON, Canada
		\AND
		\name Lan Yu \email jen.lan.yu@gmail.com \\
		\addr Google Inc. \\
		Mountain View, CA, USA
		\AND
		\name Hau Chan \email hchan3@unl.edu \\
		\addr Department of Computer Science and Engineering\\
		University of Nebraska--Lincoln, NE, USA
		\AND
		\name Edith Elkind \email elkind@cs.ox.ac.uk \\
		\addr Department of Computer Science \\
		University of Oxford, UK
}

\maketitle

\begin{abstract}
	A preference profile is single-peaked on a tree if the candidate set 
        can be equipped with a tree structure so that the preferences of each 
        voter are decreasing from their top candidate along all paths in the tree.
	This notion was introduced by \citet{demange1982single}, 
        and subsequently \citet{trick1989recognizing} described
        an efficient algorithm for deciding if a given profile is single-peaked on a tree.
	We study the complexity of multiwinner elections under several variants 
	of the Chamberlin--Courant rule for preferences single-peaked on trees. 
	We show that in this setting the egalitarian version of this rule admits a polynomial-time 
	winner determination algorithm. 
	For the utilitarian version,
	we prove that winner determination remains NP-hard for the Borda scoring function; indeed, 
	this hardness results extends to a large family of scoring functions.
	However, a winning committee can be found in polynomial time
	if either the number of leaves or the number of internal vertices 
	of the underlying tree is bounded by a constant. 
	To benefit from these positive results, we need a procedure that can determine
	whether a given profile is single-peaked on a tree that has additional desirable
	properties (such as, e.g., a small number of leaves).
        To address this challenge, we develop a structural approach that enables
	us to compactly represent all trees with respect to which a given profile is single-peaked.
	We show how to use this representation to efficiently 
	find the best tree for a given profile for use 
	with our winner determination algorithms:
	Given a profile, we can efficiently find a tree with the minimum number of leaves,
	or a tree with the minimum number of internal vertices among trees on which
	the profile is single-peaked.
	We then explore the power and limitations of this framework: we develop
	polynomial-time algorithms to find trees with the smallest maximum degree, diameter,
	or pathwidth, but show that it is NP-hard to check whether a given profile
	is single-peaked on a tree that is isomorphic to a given tree, or on a regular tree. 
\end{abstract}

\section{Introduction}\label{sec:introduction}
Computational social choice deals
with algorithmic aspects of collective decision-making. One of the 
fundamental questions studied in this area is the complexity of
determining the election winner(s) for voting rules:
indeed, for a rule to be practically applicable, it has to be the case
that we can find the winner of an election in a reasonable amount of time.

Most common rules that are designed to output a single winner admit 
polynomial-time winner determination algorithms;
examples include such diverse rules as Plurality, Borda, Maximin, Copeland, and Bucklin
(for definitions, see, e.g., the handbook 
by~\shortciteauthor{arrow-handbook}, \citeyearR{arrow-handbook}). 
However, there are also some intuitively appealing single-winner rules for 
which winner determination is known to be computationally hard: this is the case, 
for instance, for Dodgson's rule~\cite{dodgson-hard1,dodgson-hard2}, 
Young's rule~\cite{young-hard}, and Kemeny's rule~\cite{dodgson-hard1,kemeny-hard}.
More recently, there has been much interest in the computational complexity of voting rules 
whose purpose is to elect a representative \emph{committee} of candidates rather 
than select a single winner \cite{mw-survey}. One can adapt common single-winner
rules to this setting, for example by appointing the candidates with the top
$k$ scores, where $k$ is the target committee size. Such rules will pick
candidates of high ``quality'', and are useful for shortlisting purposes.
However, if we aim for a representative committee, it is preferable to use
a voting rule that is specifically designed for this purpose.
We note that \citet{mw-survey} provide a detailed discussion of different goals in
multi-winner elections, and which types of rules are suitable for each
goal.

An important representation-focused rule was proposed by \citet{chamberlin}.
Given a committee~$A$ of $k$ candidates, the rule assumes that each
voter $i$ is \emph{represented} by her most-preferred candidate in $A$, 
that is, the member of $A$ ranked highest in her preferences.
Voter $i$ is assumed to obtain utility from this representation. 
This utility is non-decreasing in the rank of her representative in her preference ranking. 
For example, her utility could be obtained as the Borda score 
she assigns to her representative (i.e., the number of candidates 
she ranks below that representative), 
but other scoring functions can be used as well.
There are no constraints on the number of voters that can be represented by a single candidate; 
the assumption is that the committee will make its decisions by weighted voting, where the weight
of each candidate is proportional to the fraction of the electorate that she represents
(or, alternatively, that the purpose of the committee is deliberation 
rather than decision-making, so the goal is to select a diverse 
committee that represents many voters).
Given a target committee size $k$, Chamberlin and Courant's scheme outputs a size-$k$ committee that 
maximizes the sum of voters' utilities according to the chosen scoring function 
(see Section~\ref{sec:prelim} for a formal definition).%
\footnote{\citet{monroe} has subsequently proposed a variant
of this scheme where the committee is assumed to use non-weighted voting, and, consequently, 
each member of the committee is required to represent approximately 
the same number of voters (up to a rounding error).}
Subsequently, \citet{betzler2013computation}
suggested an egalitarian, or maximin, variant, where the quality of a committee
is measured by the utility of the worst-off voter rather than the sum of individual utilities.

Unfortunately, the problem of identifying an optimal committee under  
the Chamberlin--Courant rule is known to be computationally hard, 
even for fairly simple scoring functions. 
In particular, \citet{mw-hard1} show that this is the case 
under $r$-approval scoring functions, where a voter obtains utility~1 
if her representative is one of her $r$ highest-ranked candidates, and utility~0 otherwise.
\citet{mw-hard2} give an NP-hardness proof for the Chamberlin--Courant rule
under the Borda scoring function. 
\citet{betzler2013computation} extend these hardness results
to the egalitarian variant.

Clearly, this is bad news if we want to use the Chamberlin--Courant rule 
in practice: elections may involve millions of voters and 
hundreds of candidates, and the election outcome needs to be announced 
soon after the votes have been cast.
On the other hand, simply abandoning these voting rules in favor
of easy-to-compute adaptations of single-winner rules is not acceptable if the goal
is to select a truly representative committee. 
Thus, it is natural to try to circumvent
the hardness results, either by designing efficient algorithms that compute an
\emph{approximately optimal} committee or by identifying reasonable assumptions
on the structure of the election that ensure computational tractability. 
The former approach was initiated by \citet{mw-hard2},
and subsequently \citet{mw-approx} and \citet{munagala2021optimal} have developed algorithms
with strong approximation guarantees; see the survey by \citet{mw-survey}.  
The latter approach was proposed by \citet{betzler2013computation},
who provide an extensive analysis of the fixed-parameter tractability 
of the winner determination problem under both utilitarian and egalitarian variants 
of the Chamberlin--Courant rule.
They also investigate the complexity of this problem for \emph{single-peaked electorates}.

A profile is said to be \emph{single-peaked} \cite{black-book} if 
the set of candidates can be placed on a one-dimensional axis,
so that a voter prefers candidates that are close to her top choice on the axis.
We can expect a profile to be single-peaked when every voter evaluates
the candidates according to their position on a numerical issue, 
such as the income tax rate or minimum wage level, 
or by their position on the left-right ideological axis.
Many voting-related problems that are known to be computationally hard for 
general preferences become easy when voters' preferences are assumed to be single-peaked
\cite{structure-survey}.
For instance, this is the case for the winner determination problem under Dodgson's, 
Young's and Kemeny's rules \cite{brandt2015bypassing}.
\citet{betzler2013computation}
show that this is also the case for winner determination 
of both the utilitarian and the egalitarian version of the Chamberlin--Courant rule.

\paragraph{Our Contribution}
The goal of this paper is to investigate whether 
the easiness results of \citet{betzler2013computation} for single-peaked electorates can be extended
to a more general class of profiles.
To this end, we explore a generalization of single-peaked preferences introduced by 
\citet{demange1982single}, which captures a much broader class of voters' preferences, while 
still implying the existence of a Condorcet winner.
This is the class of preference profiles that are single-peaked \emph{on a tree}. 
Informally, an election belongs to this class if we can construct a tree 
whose vertices are candidates in the election, and each voter
ranks all candidates according to their
perceived distance along this tree from her most-preferred candidate,
with closer candidates preferred to those who are further away. 
A profile is single-peaked if and only if it is single-peaked on a path.
Preferences that are single-peaked on a tree capture, e.g., the setting
where voters' preferences are single-peaked over non-extreme candidates, 
but a small number of extreme candidates prove to be difficult to order;
the resulting tree may then have each of the extreme candidates as a leaf. 
They also arise in the context
of choosing a location for a facility (such as a hospital or a convenience store)
given an acyclic road network. Further examples are provided by \citet{demange1982single}.
Moreover, this preference domain is a natural and well-studied 
extension of the single-peaked domain, and checking if the algorithms 
of \citet{betzler2013computation} extend to preferences that are single-peaked
on trees helps us understand whether the linear structure 
is essential for tractability. As we will see, it turns out that the answer to this question 
is `no'.


We focus on the Chamberlin--Courant rule.
We first show that, for the egalitarian variant of this rule,
winner determination is easy for an arbitrary scoring function when 
voters' preferences are single-peaked on a tree. 
Our proof proceeds by reducing our problem to an easy variant of the \textsc{Hitting Set} problem. 
For the utilitarian setting, we show that winner determination for the 
Chamberlin--Courant rule remains NP-complete if preferences are single-peaked on a tree, 
for many scoring functions, including the Borda scoring function. 
Hardness holds even if preferences are single-peaked on a tree of 
bounded diameter and bounded pathwidth.
However, we present an efficient winner determination algorithm 
for preferences that are single-peaked on a tree with a \emph{small number of leaves}:
the running time of our algorithm is polynomial in the size of the profile, 
but exponential in the number of leaves. 
Formally, the problem is in XP with respect to the number of leaves.
Further, we give an algorithm that works for trees with a small number of 
\emph{internal vertices} (i.e., with a \emph{large} number of leaves) 
when using the Borda scoring function. This algorithm places the problem 
in XP with respect to the number of internal vertices and
in FPT with respect to the combined parameter
`committee size+the number of internal vertices'.

Now, these parameterized algorithms assume that the tree with respect to which 
the preferences are single-peaked is given as an input. However, in practice
we cannot expect this to be the case: typically, we are only given the voters'
preferences and have to construct such a tree (if it exists) ourselves. 
While the algorithm of Betzler \emph{et al.}~faces the same issue (i.e., it needs to know
the societal axis), there exist efficient algorithms for determining
the societal axis given the voters' preferences 
\cite{bar-tri:j:sp,doignon1994polynomial,escoffier2008single}.
In contrast, for trees the situation is more complicated.  
\citeauthor{trick1989recognizing}~\citeyear{trick1989recognizing} 
describes a polynomial-time algorithm that
decides whether there exists a tree such that a given election 
is single-peaked with respect to it, 
and constructs \emph{some} such tree if this is indeed the case.
However, Trick's algorithm leaves us a lot of freedom when constructing the tree.
As a result, if the election is single-peaked with respect to several different trees, the output
of Trick's algorithm will be dependent on the implementation details.
In particular, there is no guarantee that an arbitrary implementation will find
a tree that caters to the demands of the winner determination algorithms that we present:
for example, the algorithm may return a tree with many leaves, while we wish to find one
that has as few leaves as possible. 
Indeed, Trick's algorithm may output a complex tree even when 
the input profile is single-peaked on a line.

To address this issue, we propose a general framework for finding trees with desired properties, 
and use it to obtain polynomial-time algorithms for identifying `nice' trees when they exist,
for several appealing notions of `niceness'. Specifically, we define a digraph
that encodes, in a compact fashion, all trees with respect to which a given profile
is single-peaked. This digraph enables us to count and/or enumerate all such trees.
Moreover, we show that it has many useful structural properties. These properties can be 
exploited to efficiently find trees with the minimum number of leaves, or the number 
of internal vertices, or the degree, or diameter, or pathwidth, among all trees 
with respect to which a given profile is single-peaked.
These recognition algorithms complement our parameterized algorithms for winner determination.
However, not all interesting questions about finding special trees are easy to solve.
In particular, we show that it is NP-hard to decide whether a given profile is single-peaked
on a regular tree, i.e., a tree all of whose internal vertices have the same degree.
It is also NP-complete to decide whether a profile is single-peaked on a tree which 
is isomorphic to a given tree.

\paragraph{Related Work}
The recent literature on the use of structured preferences in computational social choice is 
surveyed by \citet{elkind2016preferencerestrictions,structure-survey}.

\citet{demange1982single} introduced the domain of preferences single-peaked on a tree and 
showed that every profile in this domain admits a Condorcet winner. Thus, there exists a 
strategyproof voting rule on this domain. \citet{danilov1994structure} characterized the set 
of all strategyproof voting rules on this domain, generalizing the classic characterization 
for preferences single-peaked on a line by \citet{moulin1980strategy}. 
\citet{schummer2002strategy} consider strategyproofness for the case where the tree is 
embedded in $\mathbb R^2$ and preferences are Euclidean. \citet{peters2019unanimous} 
characterize strategyproof probabilistic voting rules for preferences single-peaked on trees 
and other graphs.
%
The domain of single-\emph{crossing} preferences \cite{mir:j:sc,rob:j:tax}
has also been extended to trees and other 
graphs \cite{kung2015sorting,clearwater2015generalizing,puppe2019condorcet}.

A polynomial-time algorithm for recognizing whether a profile is single-peaked 
on a line was given by 
\citet{bar-tri:j:sp}. Subsequently, faster algorithms were developed by
\citet{doignon1994polynomial} and \citet{escoffier2008single}. 
\citeauthor{fitzsimmons2020incomplete}~\citeyear{fitzsimmons2020incomplete} 
put forward an efficient algorithm for preferences 
that may contain ties. For 
single-peakedness on trees, \citet{trick1989recognizing} gives an algorithm that we describe in 
detail later. Trick's algorithm only works when voters' preferences are strict.
For preferences that may contain ties, more complicated algorithms have been 
proposed \cite{trick1988induced,conitzer2004combinatorial,tarjan1984simple,sheng2012review}. 
\citet{spoc} give a polynomial-time algorithm for recognizing preferences single-peaked on a 
circle; very recently, these results have been extended to pseudotrees 
\cite{escoffier2020}.
On the other hand, a result of \citet{gottlob2013decomposing} implies that recognizing 
whether preferences are single-peaked on a graph of treewidth 3 is NP-hard.

The complexity of winner determination under the Chamberlin--Courant rule
has been investigated by a number of authors, starting with the work of 
\citet{mw-hard1} and \citet{mw-hard2}; see the survey of \citet{mw-survey}. 
The first paper to consider this problem for a structured preference domain
was the work of \citet{betzler2013computation}, who 
gave a dynamic programming algorithm for single-peaked preferences.
This result was extended by \citet{cor-gal-spa:c:spwidth} to profiles with bounded 
\emph{single-peaked width}.
\citet{skowron2015complexity} show that 
a winning committee under the Chamberlin--Courant rule can be computed in polynomial time 
for preferences that are single-crossing, or, more generally, have bounded single-crossing width.
\citet{andrei} build on the work of \citet{clearwater2015generalizing} to
extend this result to preferences that are single-crossing on a tree.
\citet{spoc} develop a polynomial-time winner determination algorithm for profiles that are 
single-peaked on a circle, via an integer linear program that is totally unimodular if 
preferences are single-peaked on a circle, and hence optimally solved by its linear 
programming relaxation. 
In contrast, for 2D-Euclidean preferences,
\citet{godziszewski2021analysis} obtain an NP-hardness   
result for a variant of the Chamberlin--Courant rule
that uses approval ballots. 
The computational complexity of the winner determination problem 
under structured preferences has also been studied
for other voting rules: for example, \citet{brandt2015bypassing}
consider the complexity of the Dodgson rule and the Kemeny rule under 
single-peaked preferences. 

\citet{conitzer2009eliciting} showed that single-peakedness offers advantages if we wish to 
elicit voters' preferences using comparison queries: Without any information about 
the structure of the preferences, one requires $\Omega(nm\log m)$ 
comparison queries to discover the preferences of 
$n$ voters over $m$ alternatives. If preferences are known to be single-peaked on a line, then 
$O(nm)$ queries suffice. \citet{dey2016elicitation} show that single-peakedness on trees can 
also be used to speed up elicitation, as long as we know the underlying tree and this tree is 
well-structured (the relevant notions of structure are
similar to the ones considered in Section~\ref{sec:recogn}). 
\citet{sliwinski2019preferences} consider the problem of sampling preferences that are 
single-peaked on a given tree uniformly at random, and explain how to identify 
a tree that is most likely to generate a given profile, assuming that preferences
are sampled uniformly at random.


\section{Preliminaries}\label{sec:prelim}
Let $C$ be a finite set of \emph{candidates}, and let $N = \{ 1, \dots, n \}$ be a set of \emph{voters}.
A \emph{preference profile} $P$ assigns to each voter a strict total order over $C$.
For each $i \in N$ we write $\succ_i$ 
for the preference order of $i$. If $a \succ_i b$, 
then we say that voter $i$ 
prefers $a$ to $b$.

Given a profile $P$, we denote by $\pos(i, a)$ the position of 
candidate $a\in C$ in the preference order of voter $i\in N$:
\[\pos(i, a)=|\{b\in C : b\succ_i a\}|+1.\]
We write $\text{top}(i)$ for voter $i$'s most-preferred candidate, 
i.e., the candidate in position 1, we write $\text{second}(i)$ 
for the candidate in position 2, and $\text{bottom}(i)$ 
for $i$'s least-preferred candidate, i.e., the candidate in position $m$.
Given a subset of candidates $W \subseteq C$,
we extend this notation and let $\text{top}(i, W)$, $\text{second}(i, W)$, 
and $\text{bottom}(i, W)$ denote voter $i$'s most-, second-most- and least-preferred 
candidate in $W$, respectively, provided that $|W| \ge 3$.

Given a subset $W \subseteq C$, we write $P|_{W}$ for the profile obtained from $P$ by 
restricting the candidate set to $W$.

\paragraph{Multi-Winner Elections}
A \emph{scoring function} for a given $N$ and $C$ is a mapping 
$\mu:N\times C\to \mathbb Z$ such that $\pos(i, a)<\pos(i, b)$
implies $\mu(i, a) \ge \mu(i, b)$. Intuitively, $\mu(i, a)$ indicates
how well candidate $a$ represents voter $i$. 
A scoring function is said to be \emph{positional} if there exists a vector
$\vecs=(s_1, \dots, s_m)\in \mathbb Z^m$ with $s_1\ge s_2\ge\dots\ge s_m$ such that 
$\mu(i, a) = s_{\pos(i, a)}$; when this is the case, 
we will say that the scoring function is {\em induced} by the vector $\vecs$. 
It will be convenient to work with vectors $\vecs$ such that $s_1 = 0$ and 
$s_2, \dots, s_m \le 0$, where negative values correspond to 
`misrepresentation'. This choice is without loss of generality, 
as applying a positive affine transformation to $\vecs$ 
does not change the output of the voting rules we introduce below.
We will refer to the positional scoring function that corresponds 
to the vector $(0, -1, \dots, -m+1)$ as the \emph{Borda scoring function}.

Given a preference profile $P$, a committee of candidates $W\subseteq C$,
and a scoring function $\mu:N\times C\to\mathbb Z$, 
we take voter $i$'s utility from the committee $W$ to be $\mu(i, \textup{top}(i, W))$,
that is, the score she gives to her favorite candidate in $W$.
We also write
\[
\score_\mu^+(P, W)=\sum_{i\in N}\mu(i, \textup{top}(i, W))
\]
for the sum of the utilities of all voters (the \emph{utilitarian Chamberlin--Courant score}), and
\[
\score_\mu^{\min}(P, W) = \min_{i\in N}\mu(i, \textup{top}(i, W))
\]
for the utility obtained by the worst-off voter (the \emph{egalitarian Chamberlin--Courant score}).
Given a committee size $k$ with $1 \le k \le |C|$, the \emph{utilitarian Chamberlin--Courant rule} elects all committees $W\subseteq C$ with $|W| = k$ such that $\score_\mu^+(P, W)$ is maximum. 
The \emph{egalitarian} Chamberlin--Courant rule elects all committees $W\subseteq C$ with $|W| = k$ such that $\score_\mu^{\min}(P, W)$ is maximum. 
When referring to the Chamberlin--Courant rule, we will mean the utilitarian version by default. 
Sometimes it is useful to think of this rule as minimizing costs rather than maximizing scores, 
so we will write 
$\cost_\mu^+(P, W) = -\score_\mu^+(P, W)$ and 
$\cost_\mu^{\min}(P, W) = -\score_\mu^{\min}(P, W)$. 
The Chamberlin--Courant rule minimizes these costs.

To study the computation of winning committees under these rules, we now formally define the 
decision problems associated with their optimization versions.

\begin{definition}
	An instance of the \textsc{Utilitarian CC} (respectively, \textsc{Egalitarian CC})
	problem is given by a preference profile $P$,
	a committee size $k$, $1 \le k \le |C|$, 
	a scoring function $\mu:N\times C\to \mathbb Z$, 
	and a bound $B \in \mathbb Z$. 
	It is a `yes'-instance if there is a subset of candidates $W\subseteq C$
	with $|W|=k$ such that $\score_\mu^+(P, W)\ge B$ (respectively, $\score_\mu^{\min}(P, W)\ge B$)
	and a `no'-instance otherwise.%
	\footnote{
		Under our definition
		it may happen that some candidate
		in the committee does not represent any voter, i.e., there exists an $a\in W$ 
		such that $a\neq\textup{top}(i, W)$ for all $i\in N$; 
		equivalently, we allow for committees of size $k'<k$. 
		It is assumed that the voting weight of such candidate in the resulting committee will be $0$. 
		This definition is also used by, e.g., \citet{cor-gal-spa:c:spwidth}
                and \citet{mw-approx}. In contrast, 
		\citet{betzler2013computation} define the Chamberlin--Courant
		rule by explicitly specifying an assignment of voters to candidates, so that each candidate
		in $W$ has at least one voter who is assigned to her. The resulting voting rule
		is somewhat harder to analyze algorithmically. 
		Note that when $|\{\textup{top}(i, C) : i\in N\}|\ge k$, the two variants
		of the Chamberlin--Courant rule coincide.
	}
\end{definition}

We will sometimes consider the complexity of these problems for specific families
of scoring functions.
Note that a scoring function is defined for fixed $C$ and $N$, so the question
of asymptotic complexity makes sense for families of scoring functions
(parameterized by $C$ and $N$), but not for individual scoring functions.
For instance, the Borda scoring function can be viewed as a family 
of scoring functions, as it is well-defined for any $C$ and $N$.

\paragraph{Graphs and Digraphs}
A \emph{digraph} $D = (C, A)$ is given by a set $C$ of vertices and a set 
$A \subseteq C \times C$ of pairs, which we call \emph{arcs}. 
If $(a,b) \in A$, we say that $(a,b)$ is an \emph{outgoing arc} of~$a$.
An \textit{acyclic} digraph (a \textit{dag}) is a digraph with no directed cycles. 
For a vertex $a \in C$, its \textit{out-degree} $d^+(a) = |\{ b \in C : (a,b) \in A \}|$ 
is the number of outgoing arcs of $a$.
A \textit{sink} is a vertex $a$ with $d^+(a) = 0$, i.e., a vertex without outgoing arcs.
It is easy to see that every dag has at least one sink.
Given a digraph $D = (C, A)$, we write $\mathcal G(D)$
for the undirected graph $(C, E)$ where for all $a,b \in C$ we have $\{a,b\} \in E$ 
if and only if $(a,b)\in A$ or $(b, a)\in A$.
Thus, $\mathcal G(D)$ is the graph obtained from $D$ when 
we forget about the orientations of the arcs of $D$.

Given a digraph $D = (C, A)$ and a set $W \subseteq C$, we write $D|_{W}$ for the induced 
subdigraph. Similarly, for a graph $G = (C, E)$, we write $G|_{W}$ for the induced subgraph. 
We say that a set $W \subseteq C$ is \emph{connected} in a graph $G$ if $G|_{W}$ is connected.

\paragraph{Classes of Trees}

\begin{figure}
	\begin{subfigure}[b]{0.3\linewidth}
		\[
		\begin{tikzpicture}
		\graph[empty nodes,nodes={circle, fill=black!80, draw=none, inner sep=1.8pt}] {
			subgraph I_n[n=7, clockwise] -- {z/$z$}
		};
		\end{tikzpicture}
		\]
		\caption{Star}
	\end{subfigure}%
	\begin{subfigure}[b]{0.3\linewidth}
		\[
		\begin{tikzpicture}
		[every node/.style={circle, fill=black!80, draw=none, inner sep=1.8pt}]
		\node (b1) at (0.75,0.8) {};
		\node (b2) at (1.25,0.8) {};
		\node (b3) at (1,-0.8) {};
		\node (c1) at (1.75,-0.8) {};
		\node (c2) at (2.25,-0.8) {};
		\node (d1) at (3,0.8) {};
		\graph[empty nodes] {
			a -- b -- c -- d;
			b -- {(b1), (b2), (b3)};
			c -- {(c1),(c2)};
			d -- (d1);
		};
		\end{tikzpicture}
		\]
		\caption{Caterpillar}
	\end{subfigure}%
	\begin{subfigure}[b]{0.3\linewidth}
		\[
		\begin{tikzpicture}
		[every node/.style={circle, fill=black!80, draw=none, inner sep=1.8pt}]
		\node (b1) at (-0.25,0.25) {};
		\node (b2) at (-0.50,0.50) {};
		\node (b3) at (-0.75,0.75) {};
		\node (c1) at (0.25,0.25) {};
		\node (c2) at (0.50,0.50) {};
		\node (c3) at (0.75,0.75) {};
		\node (c4) at (1,1) {};
		\node (d1) at (0,-0.33) {};
		\node (d2) at (0,-0.66) {};
		\node (d3) at (0,-1.00) {};
		\graph[empty nodes] {
			c -- (c1) -- (c2) -- (c3) -- (c4);
			c -- (b1) -- (b2) -- (b3);
			c -- (d1) -- (d2) -- (d3);
		};
		\end{tikzpicture}
		\]
		\caption{Subdivision of a star}
	\end{subfigure}
	\caption{Examples of different classes of trees}
	\label{fig:example-trees}
\end{figure}
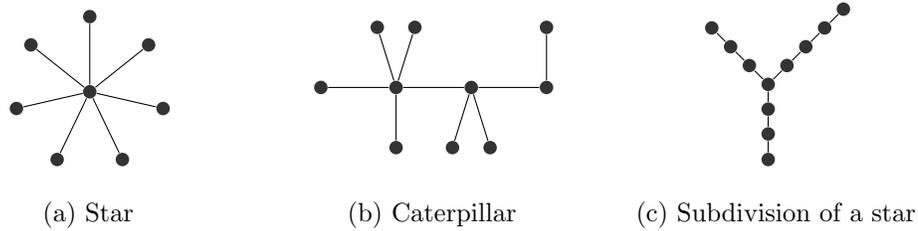

Recall that a \emph{tree} is a connected graph that has no cycles.
A \emph{leaf} of a tree is a vertex of degree~$1$. 
Vertices that are not leaves are \textit{internal} vertices. 
A \emph{path} is a tree that has exactly two leaves. 
A \emph{star} $K_{1,n}$ is a tree that has exactly one internal vertex and $n$ leaves. The internal vertex is called the \emph{center} of the star. 
The \emph{diameter} of a tree $T$ is the maximum distance between two vertices of $T$;
e.g., the diameter of a star is~$2$.
A \textit{$k$-regular tree} is a tree in which every internal vertex has degree~$k$.
Note that paths are $2$-regular, and the star $K_{1,n}$ is $n$-regular.
A \textit{caterpillar} is a tree in which every vertex is within distance $1$ of a central path.
A \textit{subdivision of a star} is a tree obtained from a star by means of replacing 
each edge by a path. Figure~\ref{fig:example-trees} illustrates some of these concepts.

\paragraph{Pathwidth}
The \emph{pathwidth} of a tree $T$ is a measure of how close~$T$ is to being a path. A \emph{path decomposition} of $T = (C,E)$ is given by a sequence $S_1, \dots, S_r$ of subsets of $C$ (called \emph{bags}) such that 
\begin{itemize}
	\item for each edge $\{a, b\}\in E$, there is a bag $S_i$ with $a,b \in S_i$, and
	\item for each $a \in C$, the bags containing $a$ form an interval of the sequence, 
              so that if $a \in S_i$ and $a \in S_j$ for $i < j$, 
              then $a$ also belongs to each of $S_{i+1}, S_{i+2}, \dots, S_{j-1}$.
\end{itemize}
The \emph{width} of the path decomposition is $\max_{i\in [r]} |S_i| -1$. 
The \emph{pathwidth} of $T$ is the minimum width of a path decomposition of $T$. 
For more on pathwidth and the related concept of treewidth, see, e.g., 
the survey by \citet{bodlaender1994tourist}.

\paragraph{Preferences that are Single-Peaked on a Tree}

Consider a tree $T$ with vertex set~$C$. 
A preference profile $P$ is said to be \emph{single-peaked on $T$}~\cite{demange1982single} if 
for every voter $i\in N$ and every pair of distinct candidates $a, b \in C$ 
such that $b$ lies on the unique path from $\text{top}(i)$ to $a$ in $T$ it holds that 
$\text{top}(i) \succ_i b \succ_i a$.
The profile $P$ is said to be \emph{single-peaked on a tree} if there exists a tree $T$
with vertex set $C$ such that $P$ is single-peaked on $T$.
The profile $P$ is said to be \emph{single-peaked} if 
$P$ is single-peaked on some tree $T$ that is a path.

The following proposition considers alternative ways of defining
preferences single-peaked on a tree $T$. The (known) proof is straightforward 
from the definitions.

\begin{proposition}
	\label{prop:spt-equivalences}
	Let $P$ be a preference profile and let $T$ be a tree on vertex set $C$. 
        The following properties are equivalent:
	\begin{itemize}
		\item $P$ is single-peaked on $T$.
		\item For every $W \subseteq C$ that is connected in $T$, $P|_{W}$ is single-peaked on $T|_{W}$.
		\item For every $i \in N$ and every $a \in C$, the \emph{top-initial segment} 
                      $\{ b \in C : b \succ_i a \}$ is connected in $T$.
	\end{itemize}
\end{proposition}

\noindent Given a profile $P$, we denote the set of all trees $T$ such that $P$ is 
single-peaked on $T$ by $\mathcal{T}(P)$.

\section{Egalitarian Chamberlin--Courant on Arbitrary Trees}

We start by presenting a simple algorithm for finding a winning committee under the egalitarian Chamberlin--Courant rule that works for preferences single-peaked on arbitrary trees.  
Our algorithm proceeds by finding a committee of minimum size that satisfies a given 
worst-case utility bound.

First, we show that the winner determination problem in the egalitarian case
can be reduced to the following 
variant of the \textsc{Hitting Set} problem, where the ground set is the vertex set of a tree $T$, 
and we need to hit a collection of connected subsets of $T$.  

\begin{definition}
	An instance of the \textsc{Tree Hitting Set} problem is given by a tree $T$ on a vertex set $C$, 
	a family $\mathcal C=\{C_1,\ldots,C_n\}$ of subsets of $C$ such that each $C_i$ is connected in $T$, and a target cover size $k\in {\mathbb Z}_+$.  
	It is a `yes'-instance if there is a subset of vertices $W\subseteq C$ with $|W|\le k$ 
	such that $W\cap C_i\neq \emptyset$ for $i=1,\ldots,n$, and a `no'-instance otherwise.
\end{definition}

\citet{guo2006exact} show that the \textsc{Tree Hitting Set} problem can be solved in polynomial time. Since they consider a dual formulation (in terms of set cover), we present an adaptation of the short argument here.
\begin{theorem}[\citet{guo2006exact}]
	\textsc{Tree Hitting Set} can be solved in polynomial time.
\end{theorem}
\begin{proof}
	Consider a vertex $a \in C$ that is a leaf of $T$, and let $b \in C$ be the (unique) 
	vertex that $a$ is adjacent to. Suppose that $a \in C_i$ for some $i$. Then, because $C_i$ is 
	a connected subset of $T$, we either have $C_i = \{a\}$ or $b \in C_i$.
	
	With this observation, we can now give a simple algorithm that constructs 
        a minimum hitting set: Consider a leaf vertex $a \in C$ adjacent to $b \in C$. 
        If there exists some $C_i \in \mathcal C$ with $C_i = \{a\}$, then any hitting set 
        must include $a$, so add $a$ to the hitting set under construction, remove $a$ 
        from $T$ and remove all copies of $\{a\}$ from $\mathcal C$. Otherwise, every set $C_i$ 
        that would be hit by $a$ is also hit by $b$, so any hitting set including $a$ 
        remains a hitting set when $a$ is replaced by $b$. Hence it is safe to delete $a$ 
        from $T$ and from each $C_i \in \mathcal C$. Now repeat the process on the smaller 
        instance constructed. Once all vertices have been deleted, return the constructed 
        hitting set, which is minimum by our argument.
\end{proof}

Now we show how to reduce our winner determination problem to the hitting set problem. 
Suppose we are given an instance of the \textsc{Egalitarian CC} problem 
consisting of a profile $P$, a tree $T$ on which $P$ is single-peaked, a target committee size $k$, and the bound $B$.
We construct a \textsc{Tree Hitting Set} instance as follows.

The ground set is the candidate set $C$, the tree $T$ is the tree with 
respect to which voters' preferences are single-peaked, and the target cover size 
equals the committee size $k$. 
For each $i\in N = \{1, \dots, n\}$, construct the set 
$$
C_i = \{a\in C : \mu(i,a)\ge B\}.
$$  
Since $\mu$ is monotone, the set $C_i$ is a top-initial segment of $i$'s preference order, 
i.e., is of the form $\{a\in C : a \succ_i b \}$ for some $b \in C$.
By Proposition~\ref{prop:spt-equivalences}, since $P$ is single-peaked on $T$, 
each set $C_i$ is connected in $T$, so we have constructed an instance 
of \textsc{Tree Hitting Set}.
Now note that for every set $W \in C$ we have
\[
\score_\mu^{\max}(P, W)\ge B \text{ if and only if }
W \cap C_i = W\cap \{a\in C : \mu(i,a)\ge B\}\neq \emptyset \text{ for all~$i$}.
\]
It follows that our reduction is correct.

Using this reduction and the algorithm for \textsc{Tree Hitting Set}, we can solve \textsc{Egalitarian CC} in polynomial time.

\begin{theorem}
	For profiles that are single-peaked on a tree, we can find a winning committee under the egalitarian Chamberlin--Courant rule in polynomial time.
\end{theorem}

\section{Hardness of Utilitarian Chamberlin--Courant on Arbitrary Trees}\label{sec:hard}

For preferences single-peaked on a \emph{path}, the utilitarian version of the Chamberlin--Courant 
rule is computationally easy: 
a winning committee can be computed using a dynamic programming
algorithm \cite{betzler2013computation}. 
While we are able to generalize this algorithm to work for some other trees 
(see Section~\ref{sec:few-leaves}), it is not clear how to extend it to arbitrary trees. 
Indeed, we will now show that for the utilitarian Chamberlin--Courant rule 
the winner determination problem remains NP-complete for 
preferences single-peaked on a tree. This hardness result holds for 
the Borda scoring function, and applies
even to trees that have diameter 4 and pathwidth 2.

We have defined the Borda scoring function as the vector $(0, -1, \dots, -(m - 1))$. 
Recall that we defined $\cost_\mu^+(P, W) = -\score_\mu^+(P, W)$, and we will use costs in the following proof to avoid negative numbers.

\begin{theorem}\label{thm:hard-borda}
	Given a profile $P$ that is single-peaked on a tree, 
        a target committee size $k$, and a target score $B$, it is \textup{NP}-complete 
        to decide whether there exists a committee of size $k$ with score at least $B$ 
        under the utilitarian Chamberlin--Courant rule with the Borda scoring function.
        Hardness holds even restricted to profiles single-peaked on a tree 
        with diameter 4 and pathwidth 2.
\end{theorem}
\begin{proof}
	We will reduce the classic \textsc{Vertex Cover} problem
	to \textsc{Utilitarian CC}. 
	An instance of \textsc{Vertex Cover} is given by 
	an undirected graph $G=(V, E)$ and a positive integer $t$.
	It is a `yes'-instance if it admits a vertex cover of size $t$, 
	i.e., a subset of vertices $S\subseteq V$ such that for each $\{u, v\}\in E$
	we have that $u\in S$ or $v\in S$. This problem is known to be NP-hard \cite{karp1972reducibility}.
	
	Given an instance $(G, t)$ of \textsc{Vertex Cover} such that $G=(V, E)$,
	$V=\{u_1, \dots, u_p\}$ and $E=\{e_1, \dots, e_q\}$,
        we construct an instance 
	of \textsc{Utilitarian CC} as follows. 

	Let $M=5pq$; intuitively, $M$ is a large number.
	We introduce a candidate $a$, two candidates $y_i$ and $z_i$ 
        for each vertex $u_i \in V$, and $M$ dummy candidates. 
        Formally, we set 
 	$$
	Y=\{y_1, \dots, y_p\}, \quad Z=\{z_1, \dots, z_p\}, \quad
        D = \{ d_1,\dots,d_M \},
	$$ 
	and define the candidate set to be
	$$C=\{a\} \cup Y \cup Z\cup D.$$
	We set the target committee size to be $k = p + t$.
	
	We now introduce the voters, who will come in three types $N = N_1 \cup N_2 \cup N_3$.
	
	\[\def\arraystretch{1.3}
	\begin{array}{cccccccccccccccccccccc}
	\toprule
	\multicolumn{3}{c}{N_1} & & \multicolumn{3}{c}{N_2} & & \multicolumn{3}{c}{N_3} \\
	\cmidrule{1-3} \cmidrule{5-7} \cmidrule{9-11} 
	5q & \cdots & 5q      && 1 & \cdots & 1 && M & \cdots & M \\
	\cmidrule{1-3} \cmidrule{5-7} \cmidrule{9-11} 
	y_1 &              & y_p          && a              && a                  &&  z_1   &&   z_p  \\
	z_1 &              & z_p          && y_{j_{1,1}}    && y_{j_{q,1}}        &&  y_1   &&   y_p   \\
	a   &              & a            && y_{j_{1,2}}    && y_{j_{q,2}}        &&  a     && a \\
	\smash\vdots &     & \smash\vdots && d_1            && d_1                &&  \smash\vdots && \smash\vdots \\
	&         &                       && \smash\vdots   && \smash\vdots       &&  \\[-2pt]
	&         &                       && d_M            && d_M                &&  \\[-2pt]
	&         &                       && \smash\vdots   && \smash\vdots       &&  \\
	\bottomrule
	\end{array}
	\]
	
	\begin{itemize}
		\item The set $N_1$ consists of $5q$ identical voters for each $u_i\in V$: 
                      they rank $y_i$ first, $z_i$ second, and $a$ third, followed by 
                      all other candidates as specified below. Intuitively,
		      the purpose of these voters 
                      is to ensure that good committees contain 
                      representatives $y_i$ of vertices in $V$.
		\item The set $N_2$ consists of a single voter for each edge $e_j\in E$: 
                      this voter ranks $a$ first, followed by the two candidates 
                      from $Y$ that correspond to the endpoints of $e_j$ (in an arbitrary order),
                      followed by the dummy candidates $d_1,\dots,d_M$, 
                      followed by all other candidates as specified below. 
                      The purpose of these voters is to ensure that every edge is covered 
                      by one of the vertices that correspond to a committee member, 
                      and to incur a heavy penalty of $M$ if the edge is uncovered.
		\item The set $N_3$ is a set of $M$ identical voters for each $u_i\in V$ 
                      who all rank $z_i$ first, $y_i$ second, and $a$ third, followed by all other 
                      candidates as specified below. 
                      The purpose of these voters is to force every good committee to include 
                      \emph{all} the $z_i$ candidates. 
	\end{itemize}
	
	We complete the voters' preferences so that the resulting profile is 
        single-peaked on the following tree:
	
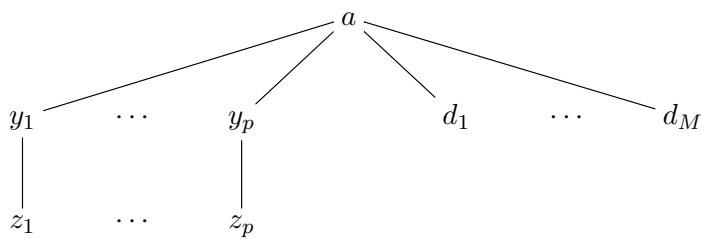
\[
	\xymatrix{
		&&& a  \ar@{-}[dlll] \ar@{-}[dl] \ar@{-}[dr] \ar@{-}[drrr]
		\\
		y_1 \ar@{-}[d] & \cdots & y_p \ar@{-}[d] & \quad & d_1 & \cdots & d_M \\
		z_1 & \cdots & z_p
	}
	\]
	This tree is obtained by taking a star with center $a$ and leaves $Y\cup D$, 
        and then attaching $z_i$ as a leaf onto $y_i$ for every $i = 1,\dots,p$. 
	It is easy to see that it has diameter~4 and pathwidth~2 (with bags $\{a,y_1,z_1\}, \dots, \{a,y_p,z_p\}, \{a,d_1\}, \dots, \{a,d_M\}$). 
        We will now specify how to complete each vote in our profile
	to ensure that the resulting profile is single-peaked on this tree. 
        Inspecting the tree, we see that it suffices to ensure that 
        for each $i=1, \dots, p$ it holds that in all votes where the positions 
        of $y_i$ and $z_i$ are not given explicitly, candidate $y_i$ is ranked above~$z_i$.
	
	This completes the construction of the profile $P$ with voter set $N$ and candidate set $C$.
	We use the Borda scoring function $\vecs = (0, -1, -2, \dots)$,
	and set the cost bound to be $B=(5q)(p-t)+2q$ and ask whether there exists a committee with  $\cost_\mu ^+(P, W) \le B$. Note that, by construction, $M > B$. 
	This completes the description of our instance of the \textsc{Utilitarian CC} 
	problem.
	Intuitively, the `correct committee' we have in mind consists of all $z_i$ candidates 
	(of which there are $p$) and a selection of $y_i$ candidates that corresponds to a vertex cover 
	(of which there should be $t$), should a vertex cover of size $t$ exist. 
	Now let us prove that the reduction is correct.
	
	Suppose we have started with a `yes'-instance of \textsc{Vertex Cover}, 
        and let $S$ be a collection of $t$ vertices that form a vertex cover of $G$. 
	Consider the committee $W=Z\cup\{y_i : u_i\in S\}$; note that $|W| = p + t = k$.
	The voters in $N_3$ and $5qt$ voters in $N_1$ have their most-preferred 
        candidate in $W$, so they contribute $0$ to the cost of $W$.
	The remaining $(5q)(p-t)$ voters in $N_1$ contribute $1$ to the cost of $W$, 
        since $z_i\in W$ for all $i$. 
	Further, each voter in $N_2$ contributes at most $2$ to the cost of $W$.
	Indeed, the candidates that correspond to the endpoints of the 
        respective edge are ranked in positions 2 and 3
        in this voter's ranking, and since $S$ is a vertex cover for $G$,
	one of these candidates is in $W$. We conclude that $\cost^+_\mu(P, W)\le (5q)(p-t)+2q = B$.
	
	Conversely, suppose there exists a committee $W$ of size $k = p+t$
	with $\cost^+_\mu(P, W)\le B$.
	Note first that $W$ has to contain all candidates in $Z$: otherwise, 
        there are $M$ voters in $N_3$
	who contribute at least $1$ to the cost of $W$, 
        and then the utilitarian Chamberlin--Courant cost of $W$
	is at least $M > B$. Thus $Z\subseteq W$. 
	We will now argue that $W\setminus Z$ is a subset of $Y$, 
	and that $S' = \{u_i : y_i\in W\setminus Z\}$ is a vertex cover of $G$.
	Suppose that $W\setminus Z$ contains too few candidates from $Y$, 
        i.e., at most $t-1$ candidates from $Y$. Then $N_1$ contains at least $(5q)(p-(t-1))$
	voters who contribute at least $1$ to the cost of $W$, so 
	$\cost^+_\mu(P, W)\ge (5q)(p-t+1) > (5q)(p-t) + 2q = B$, a contradiction.
	Thus, we have $W\setminus Z\subseteq Y$. 
	Now, suppose that $S'$ is not a vertex cover for~$G$.
	Let $e_j\in E$ be an edge that is not covered by $S'$, 
	and consider the voter in $N_2$ corresponding to $e_j$. Clearly, 
	none of the candidates ranked in positions $1, \dots, M + 3$ by this voter
	appear in $W$. Thus, this voter contributes 
	more than $M$ to the cost of $W$, so the total cost of $W$ is more than $M > B$, 
        a contradiction.
	Thus, a `yes'-instance of  \textsc{Utilitarian CC} corresponds
	to a `yes'-instance of \textsc{Vertex Cover}.
\end{proof}

In the appendix, we modify this reduction to establish 
that {\sc Utilitarian CC} remains hard even on trees with maximum 
degree 3 (Theorem~\ref{thm:hard-deg3}); intuitively, it suffices to `clone' candidate $a$. 

\section{Utilitarian Chamberlin--Courant on Trees with Few Leaves}
\label{sec:few-leaves}

The hardness result in Section~\ref{sec:hard} shows that single-peakedness on trees 
is not a strong enough assumption to make {\sc Utilitarian CC} tractable. 
However, we will now show that it is possible to achieve tractability by
placing further constraints on the shape of the underlying tree. 

Specifically, 
in this section, we present an algorithm for the utilitarian Chamberlin--Courant rule 
whose running time is polynomial
on any profile that is single-peaked on a tree with a constant number of leaves. 
The algorithm proceeds by dynamic programming. It can be viewed as a generalization 
of the algorithm due to \citet{betzler2013computation} for preferences single-peaked 
on a path (i.e., a tree with two leaves).

\begin{theorem}\label{thm:fewleaves}
	Given a profile $P$ with $|C|=m$ and $|N|=n$,
	a tree $T$ with~$\lambda$ leaves such that $P$
	is single-peaked on $T$, and a target committee size $k$,
        we can find a winning committee
	under the utilitarian Chamberlin--Courant rule
	in time 
	$\mathrm{poly}(n, m^\lambda, k^\lambda)$.
\end{theorem}

	\noindent
	\emph{Proof.}\,
	We use dynamic programming to find a committee
	of size $k$ that maximizes the utilitarian Chamberlin--Courant score.
	
	\begin{wrapfigure}[8]{r}{0.37\linewidth}
		\centering
		\vspace{-15pt}
		\begin{tikzpicture}
		[n/.style={circle, fill=black!80, draw=none, inner sep=1.8pt},
		a/.style={circle, draw, inner sep=2.8pt}]
		\node (rlabel) at (0.1, 0.35) {$r^*$};
		\node[n] (r) at (0,0) {};
		
		\node[n] (a) at (-2,-1) {};
		\node[n] (b) at (0,-1) {};
		\node[n] (c) at (2,-1) {};
		
		\node[n] (d) at (-2.5,-2) {};
		\node[n] (e) at (-2,-2) {};
		\node[n] (f) at (-1.5,-2) {};
		
		\node[n] (g) at (-0.25,-2) {};
		\node[n] (h) at (0.25,-2) {};
		
		\node[n] (i) at (1.5,-2) {};
		\node[n] (j) at (2,-2) {};
		\node[n] (k) at (2.5,-2) {};
		
		\node[a] (d-) at (-2.5,-2) {};
		\node[a] (e-) at (-2,-2) {};
		\node[a] (b-) at (0,-1) {};
		\node[a] (k-) at (2.5,-2) {};
		
		\graph{ 
			(r) -- {(a), (b), (c)}; 
			(a) -- {(d), (e), (f)}; 
			(b) -- {(g), (h)};
			(c) -- {(i), (j), (k)}; 
		};
		\end{tikzpicture}
		\caption{An anti-chain}
		\label{fig:anti-chain}
	\end{wrapfigure}
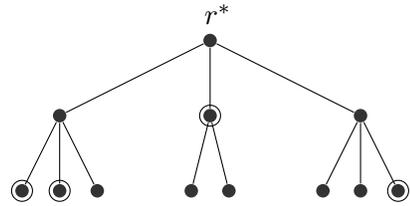
	We pick an arbitrary vertex $r^*$ to be the root of~$T$.
	This choice induces a partial order $\succ$ on $C$: we set 
	$a\succ b$ if $a$ lies on the (unique) path from $r^*$ to~$b$ in~$T$.
	Thus, $r^* \succ a$ for every $a \in C \setminus \{r^*\}$.
	A set $A\subseteq C$ is said to be an \emph{anti-chain}
	if no two elements of $A$ are comparable with respect to~$\succ$. 
	Figure~\ref{fig:anti-chain} on the right provides an example; 
        if we added the left child of $r^*$ to the set, it would no longer be an anti-chain.
	Observe that for 
	every subset of $C$ its set of maximal elements 
	with respect to $\succ$ forms an anti-chain.  
	Note also that every two ancestors of a leaf are comparable with respect to $\succ$.
        Hence, if $a$ and $b$ belong to an anti-chain $A\subseteq C$
	and $c$ is a leaf of $T$, then it cannot be the case that both 
	$a$ and $b$ are ancestors of $c$. Therefore, $|A|\le \lambda$.
	
	Given a vertex $r$, let $T_r$ be the subtree of $T$ rooted at $r$.
	The vertex set of $T_r$ is $C_r = \{r\}\cup\{a : r\succ a\}$. 
	Let $N_r = \{ i \in N : \textup{top}(i) \in C_r\}$ be the set of all voters whose most-preferred candidate
	belongs to $C_r$.
	Let $P_r$ be the profile obtained from $P$ by restricting the candidate
	set to $C_r$ and the voter set to $N_r$. 
	For each $r\in C$ and each $\ell=1, \dots, k$ let 
	\[ 
        M(r, \ell) = \max \big\{ \score_\mu^+(P_r, W) : 
        W \subseteq C_r \text{ with } |W| = \ell \text{ and } r \in W \big\} 
       \]
	be the highest Chamberlin--Courant score obtainable in $P_r$ by a committee from $C_r$ of size at most $\ell$, subject to $r$ being selected.
	
	Suppose that we have computed these quantities for all descendants
	of $r$. We will now explain how to compute them for $r$.
	Let $W \subseteq C_r$ be an optimal committee for $P_r$ that has size $\ell$ 
        and includes $r$, so that $\score_\mu^+(P_r, W) = M(r, \ell)$.
	Let $A=\{r_1, \dots, r_s\}$ be the set of maximal elements of $W\setminus \{r\}$ with respect to $\succ$
	and let $\ell_j=|W\cap C_{r_j}|$ for $j=1, \dots, s$; we have $\ell_1+\dots+\ell_s=\ell-1$.
	Since $P_r$ is single-peaked on $T$, for each $j=1, \dots, s$ it holds that each voter in $N_{r_j}$ 
        is better represented by $r_j$ than by any candidate not in $C_{r_j}$. 
        Thus, the contribution of voters in $N_{r_j}$ to the total score $M(r, \ell)$ of $W$ 
        is given by $\score_\mu^+(P_{r_j}, W \cap C_{r_j})$. In fact, this quantity must equal 
        $M(r_j, \ell_j)$, since otherwise we could replace the candidates in $W \cap C_{r_j}$ 
        by an optimizer of $M(r_j, \ell_j)$, thereby increasing the score of $W$, 
        which would be a contradiction.
	On the other hand, 
	consider a voter $i$ in $N_r\setminus(N_{r_1}\cup\dots\cup N_{r_s})$. 
	Since $P_r$ is single-peaked on $T$,
	for each $j=1, \dots, s$ it holds that candidate
	$r_j$ is a better representative for $i$ than any other candidate in $C_{r_j}$. 
	Thus, voter $i$'s most-preferred
	candidate in $W$ must be one of $r,r_1, \dots, r_s$.
	
	This suggests the following procedure for computing $M(r, \ell)$.
	The case $\ell=1$ is straightforward, as the unique optimal committee in this 
        case is $\{r\}$. For $\ell>1$,
	let $\mathcal{T}_r$ be the set of all anti-chains in $T_r$.
	A \emph{$t$-division scheme} for an anti-chain $A=\{r_1,\dots,r_s\}\in\mathcal{T}_r$
	is a list $L = (\ell_1, \dots, \ell_s)$ such that $\ell_j\ge 1$ for all $j=1, \dots, s$
	and $\ell_1+\dots+\ell_s=t$. We denote by $\mathcal{L}^A_t$ the set of all $t$-division
	schemes for $A$.
	Now, for every anti-chain $A=\{r_1, \dots, r_s\}\in \mathcal{T}_r\setminus\{\{r\}\}$ 
        and every $(\ell-1)$-division scheme
	$L=(\ell_1,\dots, \ell_s)\in\mathcal{L}^A_{\ell-1}$, we set 
	$N'_r = N_r\setminus(N_{r_1}\cup\ldots\cup N_{r_s})$
	and 
	\[
	M(A, L) = \sum_{j=1}^s M(r_j, \ell_j)+\sum_{i\in N'_r} \mu(i, \textup{top}(i, A\cup\{r\})).
	\]
	We then have $M(r, \ell)=\max_{A\in \mathcal{T}_r\setminus\{\{r\}\},L\in \mathcal{L}^A_{\ell-1}}M(A, L)$,
	where we maximize
	over all anti-chains in $T_r$ except $\{r\}$ and over all ways of dividing the $\ell-1$ slots
	among the elements of the anti-chain.
	The base case for this recurrence corresponds to the case where 
	$r$ is a leaf, 
	and is easy to deal with. 
	
	Now, the score of an optimum Chamberlin--Courant committee containing $r^*$
	is $M(r^*, k)$. We can repeat the algorithm for all possible choices of $r^*$. Then, the optimum Chamberlin--Courant score is the highest value of $M(r^*, k)$ that we have encountered.
	
	
	We have argued that the size of each anti-chain is at most~$\lambda$.
	Therefore, to calculate each $M(r, \ell)$, we enumerate at most $m^\lambda$ anti-chains
	and at most $k^\lambda$ divisions. This calculation needs to be performed for every vertex 
	$r$ (proceeding from the leaves towards the root) and for each $\ell\in [k]$.
	This establishes our bound on the running time.
	\qed
\smallskip

Notice that the time bound in Theorem~\ref{thm:fewleaves} 
implies that our problem is in XP with respect to 
the number $\lambda$ of leaves in the underlying tree. Whether there is an FPT algorithm for 
this parameter, or even for the combined parameter $(k, \lambda)$, remains 
an open problem.

\section{Utilitarian Chamberlin--Courant on Trees with Few Internal Vertices}\label{sec:few-internal}

Consider the star with center candidate $z$ and leaf candidates $c_1,\dots,c_7$. 
Which preference orders are single-peaked on this tree?
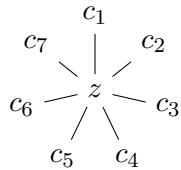
\[
\begin{tikzpicture}
\graph {
	subgraph I_n[n=7, V={$c_1$,$c_2$,$c_3$,$c_4$,$c_5$,$c_6$,$c_7$}, clockwise] -- {z/$z$}
};
\end{tikzpicture}
\]
A ranking could begin with $z$. After $z$, 
we can rank the other candidates in an arbitrary order without violating single-peakedness. But 
suppose we begin the ranking with a leaf candidate such as $c_1$. Then $z$ must be the 
second candidate, because the set consisting of the top two candidates must be connected in the 
tree. After ranking $c_1$ and $z$, we can then order the remaining candidates arbitrarily 
without violating single-peakedness. Thus, the rankings that 
are single-peaked on the star are precisely the rankings in which the center vertex is 
ranked first or second.

\begin{proposition}
	\label{prop:spt-star}
	A preference profile is single-peaked on a star if and only if there exists a candidate that every voter ranks in first or second position.
\end{proposition}

This observation implies that, in some sense, the restriction of being single-peaked on a tree 
does not give us much information. For example, consider the problem of computing an optimal 
Kemeny ranking, i.e., a ranking that minimizes the sum of Kendall tau distances to the rankings
in the input profile. This problem is NP-hard in general \cite{dodgson-hard1}, and we can easily 
see that it remains hard for preferences single-peaked on a star. Indeed, we can transform
a general instance of this problem into one that is single-peaked on a star by adding 
a new candidate that is ranked in the first position by every voter; the resulting problem 
is clearly as hard as the original one.

For some other problems, though, the restriction to stars makes the problem easy. In 
particular, this is the case for the utilitarian Chamberlin--Courant rule with the Borda 
scoring function. To see this, note that it will often be a good idea to include the candidate 
who is the center vertex of the star in the committee. Once we have done so, every voter is 
already quite well represented: the Borda score of each voter's representative is either $0$ 
or $-1$. Thus, it remains to identify $k-1$ candidates whose inclusion in the committee 
would bring the score of as many voters as possible up to $0$, which amounts to simply 
selecting $k-1$ candidates with highest Plurality scores.%
\footnote{Recall that a candidate's Plurality score is the number of voters 
who rank this candidate first.} Finally, we need to consider the 
case where an optimal committee does not include the center vertex. 
One can check that in this case the committee necessarily consists 
of $k$ candidates with highest Plurality scores 
(see the proof of Theorem~\ref{thm:borda} below). By selecting the better of these two 
committees, we find a winning committee. This procedure works for many scoring 
functions other than Borda (see the end of this section). However, as we show in the appendix, 
this argument does not extend to \emph{all} scoring functions: 
For some positional scoring functions, winner 
determination for utilitarian Chamberlin--Courant remains hard 
even for preferences single-peaked on a star (Theorem~\ref{thm:hard-star}).

The algorithm we have sketched for the Borda scoring function on stars can be generalized to 
trees that have a small number of internal vertices (and thus a large number of 
leaves). While for stars it suffices to guess whether the center vertex would be part of the winning 
committee, we now have to make a similar guess for \emph{each} internal vertex. 

\begin{theorem}\label{thm:borda}
	Given a profile $P$ with $|C| = m$ and $|N| = n$,
	a tree $T \in \mathcal{T}(P)$ with $\eta$ internal vertices 
        such that $P$ is single-peaked on $T$,
	and a target committee size $k\ge 1$, we can find a winning committee of size $k$
	for $P$ under the Chamberlin--Courant rule with the Borda scoring
	function in time $\mathrm{poly}(n,m, (k+1)^\eta)$. 
\end{theorem}
\begin{proof}
	\newcommand{\internal}{C^\circ}
	\newcommand{\ch}{{\mathrm{lvs}}}
	\newcommand{\plu}{{\mathrm{plu}}}
	Given a candidate $c\in C$, let $\plu(c) = |\{ i \in N : \text{top}(i) = c \}|$ be the number of voters in $P$ that rank $c$ first.
	Let $\internal$ be the set of internal vertices of $T$.
	For each candidate $c\in \internal$, let $\ch(c)$ denote the set of leaf candidates in $C\setminus \internal$
	that are adjacent to $c$ in $T$.  
	
	Our algorithm proceeds as follows. 
	For each candidate $c\in \internal$ it guesses a pair $(x(c), \ell(c))$,
	where $x(c)\in\{0,1\}$ and $0\le \ell(c) \le |\ch(c)|$.
	The component $x(c)$ indicates whether $c$
	itself is in the committee, and $\ell(c)$ indicates 
	how many candidates in $\ch(c)$ are in the committee. 
	We require $\sum_{c\in \internal}(x(c)+\ell(c))=k$. 
	Next, the algorithm sets $W=\{c\in \internal: x(c)=1\}$, and then for each $c\in \internal$
	it orders the candidates in $\ch(c)$ in non-increasing
	order of $\plu(c)$ (breaking ties according to a fixed ordering $\rhd$ over $C$), 
	and adds the first $\ell(c)$ candidates in this order to $W$. 
	
	Each guess corresponds to a committee of size $k$. Guessing can be implemented
	deterministically: consider all options for the collection of pairs 
	$\{(x(c), \ell(c))\}_{c\in \internal}$ satisfying $\sum_{c\in \internal}(x(c)+\ell(c))=k$ (there are at most $2^\eta\cdot (k+1)^\eta$ possibilities),
	compute the Chamberlin--Courant score of the resulting committee for each option, and output the best one.
	
	It remains to argue that this algorithm finds a committee with the maximum Chamberlin--Courant score.
	To see this, 
	let $\cal S$ be the set of all size-$k$ committees with the maximum Chamberlin--Courant score,
	and let $S^*$ be the maximal committee in $\arg\max_{W\in \cal S}|W\cap \internal|$ with respect
	to the fixed tie-breaking ordering $\rhd$ (where, given two size-$k$ committees $S, S'\in\cal S$, 
	we write $S'\rhd S$ if $a'\rhd a$, where $a'$ is the maximal element of $S'\setminus S$
	with respect to $\rhd$ and $a$ is the maximal element of $S\setminus S'$ with respect to $\rhd$).
	
	For each $c\in \internal$, let $x^*(c)=1$ if $c\in S^*$
	and $x^*(c)=0$ otherwise, and let $\ell^*(c)=|\ch(c)\cap S^*|$. 
	Our algorithm will consider the collection of pairs $\{(x^*(c), \ell^*(c))\}_{c\in \internal}$ 
	at some point, and construct a committee $S$ based on this collection. 
        We will now argue that $S=S^*$. 
	This would show correctness of our algorithm, 
        since it returns a committee with a total score at least as high as that of $S$.
	
	Clearly, we have $\internal\cap S=\internal\cap S^*$, so it remains to argue that
	$\ch(c)\cap S^* = \ch(c)\cap S$ for each $c\in \internal$.
	Suppose for the sake of contradiction that this is not the case,
	i.e., there exists a $c\in \internal$ and a pair of candidates $a, b\in\ch(c)$
	with $a\in S\setminus S^*$ and $b\in S^*\setminus S$. 
	We distinguish two cases: $c \in S^*$ or $c \not \in S^*$.
	
	If $c\in S^*$, consider the committee $S'=(S^*\setminus\{b\})\cup\{a\}$.
	We claim that $S'$ has the same Chamberlin--Courant score as $S^*$. 
        Note that when moving from $S^*$ to $S'$,
	\begin{itemize}
		\item the contribution of the $\plu(b)$ voters who rank $b$ first changes from $0$ to $-1$,
		\item the contribution of the $\plu(a)$ voters who rank $a$ first changes from $-1$ to $0$, 
		\item the contribution of all other voters is unaffected by the change, 
                      since they prefer $c \in S^* \cap S'$ to both $a$ and $b$.
	\end{itemize}
	We also have $\plu(a)\ge \plu(b)$ by construction of $S$, 
        and so the score of $S'$ is at least the score of $S^*$, and hence $\plu(a)=\plu(b)$.
	But then by construction of $S$ we have $a\rhd b$, and this contradicts
	our choice of $S^*$ from $\arg\max_{W\in \cal S}|W\cap \internal|$.
	
	Now, suppose that $c\not\in S^*$.
	Consider the committee $S'=(S^*\setminus\{b\})\cup\{c\}$.
	Again, we claim that $S'$ has the same Chamberlin--Courant score as $S^*$. 
        Note that when moving from $S^*$ to $S'$,
	\begin{itemize}
		\item the contribution of each of the $\plu(b)$ voters who rank $b$ first decreases by $1$
		      (as all of them rank $c$ second),
		\item the contribution of each of the $\plu(a)$ voters who rank $a$ first increases by at least $1$ 
		(as all of them rank $c$ second),
		\item the contribution of any other voter does not decrease
                      (as all of them prefer $c$ to $b$).
	\end{itemize}
	Again, we have $\plu(a)\ge \plu(b)$ by construction of $S$, 
        and so the score of $S'$ is at least the score of $S^*$, and hence $\plu(a)=\plu(b)$.
	Thus, the Chamberlin--Courant score of $S'$ is optimal, and so $S' \in \cal S$. 
        But $|S'\cap \internal|>|S^*\cap \internal|$, which contradicts our choice of $S^*$ 
	from $\arg\max_{W\in \cal S}|W\cap \internal|$.
\end{proof}
The reader may wonder if Theorem~\ref{thm:borda} can be strengthened to an FPT
result with respect to $\eta$, e.g., by guessing the subset of internal nodes 
to be selected and then picking the leaves in a globally greedy fashion in order
of their Plurality scores. However, it can be shown that this approach does
not necessarily produce an optimal committee: this is because, when considering 
a leaf that is adjacent to an internal node that is not selected, we need
to take into account its contribution to the utility of voters 
who do not rank it first, and this contribution may depend on what other 
leaves have been selected.

Further, it is clear from our proof that Theorem~\ref{thm:borda} holds for every positional
scoring function whose score vector satisfies $s_1=0$, $s_2=-1$, $s_3\le -2$.
The proof does not extend to arbitrary positional scoring functions, and
Theorem~\ref{thm:hard-star} in the appendix shows that 
for some positional scoring functions {\sc Utilitarian CC} is NP-hard
even if preferences are single-peaked on a star. Note that, in contrast, 
Theorem~\ref{thm:fewleaves} holds for any positional scoring function.
On the other hand, the algorithm described in the proof of Theorem~\ref{thm:borda}
is in FPT with respect to the combined parameter $(k, \eta)$.
In contrast, for general preferences computing the Chamberlin--Courant winners
is $\mathrm{W}[2]$-hard with respect to $k$ under the Borda scoring function
\cite{betzler2013computation}.

\section{Structure of the Set of Trees a Profile is Single-Peaked on: The Attachment Digraph}
\label{sec:attachment-digraph}
We now move on from our study of multiwinner elections and turn towards the problem of 
recognizing when a given preference profile is single-peaked on a tree. In particular, for 
each profile $P$, we will study the collection $\mathcal T(P)$ of \emph{all} trees on which 
$P$ is single-peaked. It turns out that the set $\mathcal T(P)$ has interesting structural 
properties, and admits a concise representation. In many cases, this will allow us to pick a 
`nice' tree from $\mathcal T(P)$ i.e., a tree that satisfies certain additional requirements. 
For example, 
to use the algorithm from Section~\ref{sec:few-leaves}, we would want to pick a tree from 
$\mathcal T(P)$ with the smallest number of leaves, and to use the algorithm from 
Section~\ref{sec:few-internal}, we would want to use a tree with the smallest number of internal 
vertices.

\citet{trick1989recognizing} presents an algorithm that decides whether $\mathcal T(P)$ 
is non-empty. If so, the algorithm produces some tree $T$ with $T \in \mathcal T(P)$. While 
building the tree, the algorithm makes various arbitrary choices. In our approach, we will 
store all the choices that the algorithm could take. To this end, 
we introduce a data structure, which we 
call the \textit{attachment digraph} of profile $P$.

We will start by giving a description of Trick's algorithm and its proof of correctness.
We follow the presentation of Trick's paper closely, but give somewhat more detailed proofs.

Trick's algorithm can be seen as taking inspiration from algorithms for recognizing preferences that 
are single-peaked on a line. Those typically start out by noticing that an alternative that is ranked 
bottom-most by some voter must be placed at one of the ends of the axis. Trick's algorithm uses the 
same idea; the analogue for trees is as follows.
\begin{proposition}
	[\citealp{trick1989recognizing}]
	\label{prop:bottom-is-leaf}
	Suppose $P$ is single-peaked on $T$, and suppose $a$ occurs as a bottom-most alternative, 
	that is, $\operatorname{bottom}(i) = a$ for some $i\in N$. Then $a$ is a leaf of $T$.
\end{proposition}
\begin{proof}
	The set $A\setminus\{a\}$ is a top-initial segment of the $i$-th vote, and 
        hence must be connected in $T$. This can only be the case if $a$ is a leaf of $T$.
\end{proof}

Suppose we have identified a bottom-ranked alternative $a$. We deduce that if our profile is 
single-peaked on any tree $T$, then $a$ is a leaf of $T$. Now, being a leaf, $a$ must have 
exactly one neighboring vertex $b$. Which vertex could this be? The following simple 
observation gives some necessary conditions.

\begin{proposition}
	[\citealp{trick1989recognizing}]
	\label{prop:spt-leaf-attachment}
	Suppose $P$ is single-peaked on $T$, and suppose $a\in C$ is a leaf of $T$, 
        adjacent to $b\in C$. Let $i \in N$ be a voter. Then either 
	\begin{enumerate}
		\item[(i)] $b \succ_i a$, or
		\item[(ii)] $a = \textup{top}(i)$ and $b = \textup{second}(i)$.
	\end{enumerate}
\end{proposition}
\begin{proof}
	(i) Suppose first that $a$ is not $i$'s top-ranked alternative, 
            and rather $\text{top}(i) = c$. Take the unique path in $T$ from $c$ to $a$, 
            which passes through $b$ since $b$ is the only neighbor of $a$. 
	    Since $i$'s vote is single-peaked on $T$, it is single-peaked on this path, 
	    and hence $i$'s preference decreases along it from $c$ to $a$. Since $b$ is visited before $a$,
	    it follows that $b\succ_i a$.
	
	(ii) Suppose, otherwise, that $a$ is $i$'s top-ranked alternative. Then $\{a,\textup{second}(i)\}$
	    is a top-initial segment of $i$'s vote, which by Proposition~\ref{prop:spt-equivalences} 
            is a connected set in $T$, and hence forms an edge. Thus, $a$ is adjacent to 
            $\textup{second}(i)$, so $\textup{second}(i) = b$ as required. 
\end{proof}

Thus, in our search for a neighbor of the leaf $a$, we can restrict our attention 
to those alternatives $b$ that satisfy either (i) or (ii) in the proposition above, 
for every voter $i\in N$. Let us write this down more formally: 
For each $a\in C$ and $i\in N$, define
\[
B(i,a) = \begin{cases}
\{ c \in C : c \succ_i a \} & \text{if } \text{top}(i) \neq a, \\
\{ \operatorname{second}(i) \} & \text{if }  \text{top}(i) = a.
\end{cases}
\]
Applying Proposition~\ref{prop:spt-leaf-attachment} to all voters $i\in N$, 
we see that $a$ needs to be adjacent to an element in
\[
B(a) = \bigcap_{i\in N} B(i,a).
\]
Thus, we have the following corollary.
\begin{corollary}
	[\citealp{trick1989recognizing}]
	\label{cor:must-attach-to-B}
	Suppose a profile is single-peaked on $T$, and $a\in C$ is a leaf of $T$. 
	Then $a$ must be adjacent to an element of $B(a)$.
\end{corollary}

We have established that it is necessary for leaf $a$ to be adjacent to some alternative in
$B(a)$. It turns out that if the profile is single-peaked on a tree, then for \emph{every} 
alternative $b\in B(a)$ there is some tree $T \in \mathcal T(P)$ in which $a$ is adjacent to 
$b$.

\begin{proposition}
	[\citealp{trick1989recognizing}]
	\label{prop:can-attach}
	Let $P$ be a profile in which $a$ occurs bottom-ranked. Suppose that 
	$P|_{C \setminus \{a\}}$ is single-peaked on some tree $T_{-a}$
	with vertex set $C\setminus\{a\}$, and let $T$ 
        be a tree obtained from $T_{-a}$ by attaching $a$ as a leaf adjacent 
	to some element $b\in B(a)$.
        Then $P$ is single-peaked on $T$.
\end{proposition}
\begin{proof}
	Let $T$ be a tree obtained as described. We show that $P$ is single-peaked on $T$.
	Let $S \subseteq C$ be a top-initial segment of the ranking of some voter $i$ in $P$. 
        We need to show that $S$ is connected in $T$. 
	\begin{itemize}
		\item If $a\not\in S$, then $S$ is connected in $T_{-a}$ 
                      because $P|_{C \setminus \{a\}}$ is single-peaked on $T_{-a}$. 
                      Hence $S$ is also connected in $T$.
		\item If $S = \{a\}$, then $S$ is trivially connected in $T$.
		\item If $a \in S$ and $S \neq \{a\}$, then $S \setminus \{a\}$ 
                      is connected in $T_{-a}$ because $P|_{C \setminus \{a\}}$ 
                      is single-peaked on $T_{-a}$. Therefore, to show that $S$ 
                      is connected in $T$, it suffices to show that $a$'s neighbor $b$ 
                      is also an element of $S$. Since $b \in B(a) = \bigcap_{i\in N} B(i,a)$, 
                      we have that $b \in B(i,a)$. If $\text{top}(i) = a$, 
                      then $B(i,a) = \{\text{second}(i)\}$, so $b = \text{second}(i)$. 
                      As $S$ is a top-initial segment of $i$ with $|S| \ge 2$, 
                      we have $b \in S$, as desired. Otherwise $\text{top}(i) \neq a$, 
                      and so $B(i,a) = \{ c : c \succ_i a \}$, hence $b \succ_i a$. 
                      As $S$ is a top-initial segment of $i$ including $a$, 
                      we must have $b\in S$, as desired.
		\qedhere
	\end{itemize}
\end{proof}

\begin{algorithm}[th]
	\caption{Trick's algorithm to decide whether a profile is single-peaked on a tree}
	\begin{algorithmic}
		\State $T \gets (C, \emptyset)$, the empty graph on $C$
		\State $C_1 \gets C$, $r \gets 1$
		\While{$|C_r| \ge 3$}
		\State $L_r \gets \{ \text{bottom}(i, C_r) : i \in N \} $
		\For{each candidate $a\in L_r$}
		\State $B(a) \gets \bigcap_{i\in N} B(i, C_r, a)$
		\If {$B(a) = \emptyset$}
		\State \Return \textit{fail : $P$ is not single-peaked on any tree}
		\Else
		\State select $b \in B(a)$ arbitrarily
		\State add an edge between $a$ and $b$ in $T$
		\EndIf
		\EndFor
		\State $C_{i+1} \gets C_r \setminus L_r$
		\State $r \gets r+1$
		\EndWhile
		\If{$|C_r| = 2$}
		\State add an edge between the two candidates in $C_r$ to $T$
		\EndIf
		\State \Return $P$ is single-peaked on $T$
	\end{algorithmic}
	\label{alg:trick}
\end{algorithm}

With these results in place, we can now see how a recognition algorithm could work. 
Select an alternative $a$ that is ranked bottom-most by some voter, select an arbitrary 
candidate $b \in B(a)$, add an edge $\{b,a\}$ to the tree under construction, remove $a$ from 
the profile, and recurse on the remaining candidates. If at any point we find that 
$B(a) = \emptyset$, then we can conclude from Corollary~\ref{cor:must-attach-to-B} 
that the profile is 
not single-peaked on any tree. Algorithm~\ref{alg:trick} formalizes this procedure. To avoid 
recursion, the algorithm uses the following notation: for every subset $S \subset C$, for each 
$a \in S$, and each $i \in N$, define
\[
B(i,S,a) = \begin{cases}
\{ c \in S : c \succ_i a \} & \text{if } \text{top}(i, S) \neq a, \\
\{ \operatorname{second}(i, S) \} & \text{if }  \text{top}(i, S) = a.
\end{cases}
\]

\begin{theorem}
	[\citealp{trick1989recognizing}]
	\label{thm:trick-is-correct}
	Algorithm~\ref{alg:trick} correctly decides whether a profile is single-peaked on a tree.
\end{theorem}
\begin{proof}

	First, note that if Algorithm~\ref{alg:trick} succeeds and returns a graph $T$, then 
        $T$ is a tree. Indeed, it is easy to see that $T$ has $|C| - 1$ edges. Moreover, 
        $T$ is connected, because every vertex has a path to a vertex in the set $C_r$ 
	at the end of the algorithm, and $C_r$ is either a singleton or a connected set of size 2.
	
	We show that the algorithm is correct by induction on $|C|$. 
	If $|C| = 1$ or $|C| = 2$, every profile is single-peaked on the unique tree on $C$, 
        and Algorithm~\ref{alg:trick} correctly determines this. If $|C| \ge 3$, 
        then the \emph{while}-loop is executed at least once. If in the first iteration 
        the algorithm claims that the profile is not single-peaked on a tree because 
        $B(a) = \emptyset$ for some $a \in L_1$, then this statement is correct by 
        Corollary~\ref{cor:must-attach-to-B}. Otherwise, after the first iteration
	the algorithm behaves as if it was run on $P|_{C_2}$ 
        (recall that $C_2 = C \setminus L_1$). 
	
	Now, if the algorithm fails in later iterations, by the inductive hypothesis, 
        $P|_{C_2}$ is not single-peaked on a tree. But then $P$ is not single-peaked 
        on a tree either: Suppose it was single-peaked on $T$. 
        Then, by Proposition~\ref{prop:bottom-is-leaf}, all candidates in $L_1$ are leaves 
        of $T$, and therefore $T|_{C_2}$ is still a tree, and so $P|_{C_2}$ is single-peaked 
        on $T|_{C_2}$ (by Proposition~\ref{prop:spt-equivalences}), a contradiction. 
        Thus, in this case, the algorithm executed on $P$ correctly determines that $P$ 
        is not single-peaked on a tree.
	
	On the other hand, if the algorithm's run on $P$ terminates and returns a tree $T$, 
        then its run on $P|_{C_2}$ would have terminated and returned the tree 
        $T|_{C_2}$. By the inductive hypothesis, $P|_{C_2}$ is single-peaked on $T|_{C_2}$. 
        Hence, by Proposition~\ref{prop:can-attach}, $P$ is single-peaked on $T$, 
        and so the algorithm is correct.
\end{proof}
This concludes our presentation of Trick's approach. \\

Trick's algorithm makes some arbitrary choices when selecting alternatives $b \in B(a)$. Our 
aim is to understand the set $\mathcal T(P)$ of \emph{all} trees that the input profile $P$ is 
single-peaked on, so a natural approach is to record all possible choices that Trick's 
algorithm could make at each step, as this will encode all possible outputs of the algorithm. 
We do this by running Algorithm~\ref{alg:attachment-digraph}, which has the same structure as 
Algorithm~\ref{alg:trick}. Given a profile that is single-peaked on \emph{some} tree, it 
constructs and returns a digraph $D$ with vertex set $C$ which contains all possible choices 
that Trick's algorithm can make. We call $D$ the \emph{attachment digraph} of the input 
profile.

\begin{algorithm}[th]
	\caption{Build attachment digraph $D=(C, A)$ of $P$}
	\begin{algorithmic}
		\State $D \gets (C, A)$, $A \gets \emptyset$, so $D$ is the empty digraph on $C$
		\State $C_1 \gets C$, $r \gets 1$
		\While{$|C_r| \ge 3$}
		\State $L_r \gets \{ \text{bottom}(i, C_r) : i \in N \} $
		\For{each candidate $a\in L_r$}
		\State $B(a) \gets \bigcap_{i\in N} B(i, C_r, a)$
		\If {$B(a) = \emptyset$}
		\State \Return \textit{fail : $P$ is not single-peaked on any tree}
		\Else
		\State for each $b \in B(a)$, add an arc $(a,b)$ to $A$
		\EndIf
		\EndFor
		\State $C_{i+1} \gets C_r \setminus L_r$
		\State $r \gets r+1$
		\EndWhile
		\If{$|C_r| = 2$}
		\State add an arc between the two candidates in $C_r$ to $A$, arbitrarily directed
		\EndIf
		\State \Return $D$
	\end{algorithmic}
	\label{alg:attachment-digraph}
\end{algorithm}

\begin{figure}[t]
	\centering
	\begin{subfigure}[b]{0.4\linewidth}
		\centering
		\begin{tikzpicture}
		[trans/.style={->,gray!60,thick,bend left=20},
		forced/.style={->,thick},
		node distance=4mm and 10mm, inner sep=1.2pt]
		
		\node (c) {$a$};
		\node (d) [right=of c] {$b$};
		\node (e) [right=of d] {$c$};
		\node (f) [right=of e] {$d$};
		\node (g) [right=of f] {$e$};
		
		\path 
		(c) edge [forced] (d)
		(d) edge [forced] (e)
		(g) edge [forced] (f)
		(f) edge [forced] (e)
		;
		\end{tikzpicture}
		\caption{}
	\end{subfigure}%
	\begin{subfigure}[b]{0.3\linewidth}
		\centering
		\begin{tikzpicture}
		[trans/.style={->,gray!60,thick},
		forced/.style={->,thick},
		node distance=8mm and 13mm, inner sep=1pt]
		
		\node (a) {$a$};
		\node (b) [right=of a] {$b$};
		\node (c) [right=of b] {$c$};
		\node (d) [above=of c] {$d$};
		\node (e) [above=of b] {$e$};
		
		\path 
		(a) edge [forced] (b)
		(e) edge [forced] (b)
		(d) edge [trans,bend left=20] (c)
		(d) edge [trans,bend right=20] (b)
		(c) edge [forced] (b)
		;
		\end{tikzpicture}
		\caption{}
	\end{subfigure}

	\begin{subfigure}[b]{0.6\linewidth}
		\centering
		\begin{tikzpicture}
		[trans/.style={->,gray!60,thick,bend left=20},
		forced/.style={->,thick},
		node distance=4mm and 15mm, inner sep=1pt]
		
		\node (c) {$c$};
		\node (d) [right=of c] {$d$};
		\node (b) [above=of d] {$b$};
		\node (a) [above=of b] {$a$};
		\node (e) [right=of d] {$e$};
		\node (f) [right=of e] {$f$};
		\node (g) [right=of f] {$g$};
		\node (h) [above=of g] {$h$};
		\node (i)  [above left=8mm and 9mm of h] {$i$};
		\node (j) [above right=of a] {$j$};
		\node (k) [below=2mm of g] {$k$};
		
		\path 
		(a) edge [trans] (b)
		edge [trans,bend right] (c)
		edge [trans,bend right] (d)
		edge [trans] (e)
		(b) edge [trans,bend right] (c)
		edge [trans] (d)
		(c) edge [forced] (d)
		(d) edge [forced] (e)
		(g) edge [forced] (f)
		(k) edge [forced] (f)
		(e) edge [forced] (f)
		(j) edge [trans,bend right] (c)
		edge [trans,bend right] (d)
		edge [trans,bend right] (e)
		edge [trans] (f)
		edge [trans] (g)
		edge [trans] (h)
		edge [trans] (i)
		(i) edge [trans,bend right=20] (f)
		edge [trans] (g)
		edge [trans] (h)
		(h) edge [trans,bend right=20] (f)
		edge [trans] (g)
		;
		\end{tikzpicture}
		\caption{}
	\end{subfigure}%
	\caption{The attachment digraphs of the profiles in Example~\ref{ex:attachment-digraphs}. If a vertex has a unique outgoing arc, the arc is drawn in black. If the vertex has at least two outgoing arcs, the arcs are drawn in gray and curved.}
	\label{fig:attachment-digraph}
\end{figure}
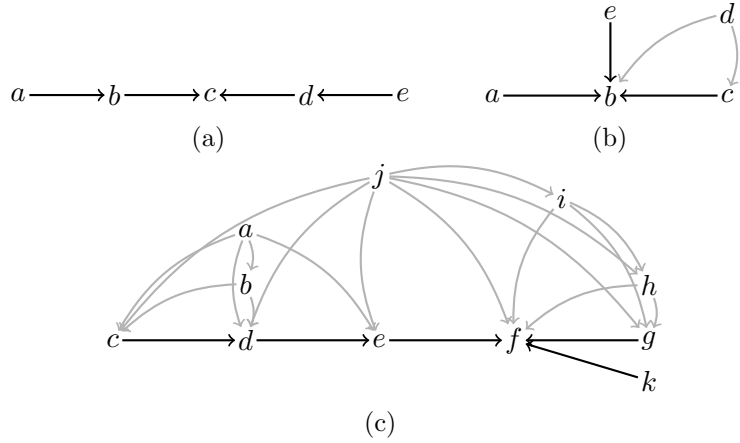
\begin{example}
	\label{ex:attachment-digraphs}
	The attachment digraphs of the following three profiles are shown in Figure~\ref{fig:attachment-digraph}.
	
	(a) Suppose $C = \{ a,b,c,d,e \}$, and let $P_1$ be the profile with voters $N = \{1,2\}$ 
            such that 
	    $$a \succ_1 b \succ_1 c \succ_1 d \succ_1 e\qquad\text{and}\qquad
	      e \succ_2 d \succ_2 c \succ_2 b \succ_2 a,$$ 
	    so that the two votes are the reverse of each other. Running Algorithm~\ref{alg:attachment-digraph}, we consider the sets $L_1 = \{a,e\}$ and $L_2 = \{b,d\}$.
	
	(b) Suppose $C = \{ a,b,c,d,e \}$, and let $P_2$ be the profile with voters $N = \{1,2\}$ 
            such that 
	    $$a \succ_1 b \succ_1 c \succ_1 d \succ_1 e\qquad\text{and}\qquad
              e \succ_2 b \succ_2 c \succ_2 d \succ_2 a.
	    $$ 
	    Running Algorithm~\ref{alg:attachment-digraph}, we consider the sets $L_1 = \{a,e\}$ and $L_2 = \{d\}$.
	
	(c) Suppose $C = \{a,b,c,d,e,f,g,h,i,j,k\}$, and let $P_3$ be the profile with voters 
            $N = \{1,2,3\}$ such that 
            \begin{align*}
            	&k \succ_1 f \succ_1 e \succ_1 d \succ_1 g \succ_1 h \succ_1 c 
            	\succ_1 i \succ_1 j \succ_1 b \succ_1 a, \\ 
            	&d \succ_2 c \succ_2 b \succ_2 e \succ_2 a \succ_2 f \succ_2 g 
            	\succ_2 h \succ_2 i \succ_2 j \succ_2 k, \\
            	&g \succ_3 f \succ_3 h \succ_3 i \succ_3 e \succ_3 d \succ_3 c 
            	\succ_3 b \succ_3 a \succ_3 j \succ_3 k. 
            \end{align*}
            Running Algorithm~\ref{alg:attachment-digraph}, we consider the sets 
            $L_1 = \{a,k\}$, $L_2 = \{b,j\}$, $L_3 = \{c, i\}$, $L_4 = \{d,h\}$, and $L_5 = \{e,g\}$.
\end{example}

Algorithm~\ref{alg:attachment-digraph} runs in time $O(|N| \cdot |C|^2)$. 
In the rest of this section, we will analyze the structure of the attachment digraph, 
and its relation to the set $\mathcal T(P)$ of trees on which $P$ is single-peaked. 
Throughout, we fix the profile $P$ and let $C_r$, $L_r$, and $B(a)$ refer to the sets 
constructed while running the $r$-th iteration of Algorithm~\ref{alg:attachment-digraph} on $P$.
We start with a few simple properties.

\begin{proposition}
	\label{prop:arcs-point-from-i-to-j}
	Let $a \in C$ be a candidate with $a \in L_r$. Then $B(a) \cap L_r = \emptyset$. 
	Hence, for every arc $(a,b) \in A$ with $a \in L_r$ and $b \in L_s$, we have that $s > r$.
\end{proposition}
\begin{proof}
	Note that the set $L_r$ is only well-defined when $|C_r|\ge 3$.
	Assume for a contradiction that $b \in B(a)$ and $b \in L_r$. Since $b \in L_r$, 
	there is some voter $i \in N$ such that $b = \text{bottom}(i, C_r)$.
	But then $b \not\succ_i a$ and $b\neq \text{second}(i, C_r)$ (as $|C_r|\ge 3$),
	so it cannot be the case that $b \in B(i, C_r, a)$, a contradiction.
	
	For the second statement, note that if $(a,b)\in A$ is an arc, then $b \in B(a)$. 
        Since $B(a) \subseteq C_r = C \setminus (L_1 \cup \ldots \cup L_{r - 1})$, 
        we must have $s \ge r$. By the previous paragraph, $s = r$ is impossible, and hence $s > r$.
\end{proof}

\begin{proposition}
	\label{prop:acyclic}
	Every attachment digraph $D = (C, A)$ is acyclic and has exactly one sink. 
\end{proposition}
\begin{proof}
	Suppose that the \emph{while}-loop of Algorithm~\ref{alg:attachment-digraph} 
        is executed $R-1$ times, and consider the sets $L_1,\dots,L_{R-1}$. 
	Set $L_R := C \setminus (L_1 \cup \ldots \cup L_{R-1})$. 
	Then $L_1, \dots, L_R$ is a partition of $C$. 
	
	For acyclicity, note that for each $a \in L_r$ with $1\le r< R$ 
        we have $B(i, C_r, a) \subseteq C_r$ and hence
        $B(a) \subseteq C_r$.
	Together with Proposition~\ref{prop:arcs-point-from-i-to-j}, this implies that
        if $a\in L_r$ then all outgoing arcs of $a$ point into $L_{r+1} \cup \ldots \cup L_R$.
	Hence, based on the partition $L_1,\dots,L_R$, the set $D$ can be topologically ordered
        and thus cannot contain a cycle.
	
	For the number of sinks, note that there is at least one sink in $D$ 
        because $D$ is acyclic. Since for every $a \in C \setminus L_R$ 
        we have $B(a) \neq \emptyset$, 
	at least one outgoing arc of $a$ is added to $D$. Thus, no vertex in $C\setminus L_R$ 
        is a sink. 
	The condition of the \emph{while}-loop implies that $|L_R| \le 2$.
	If $|L_R| = 1$, then there is exactly one sink, and we are done.
	If $|L_R| = 2$, then the final {\bf if} clause of the algorithm 
        adds an arc between the two vertices in $L_R$, which ensures that only one of them is a sink.
\end{proof}

If we wish to extract a tree $T \in \mathcal T(P)$ from the attachment digraph $D$, Trick's 
algorithm tells us that we must choose, for each non-sink vertex of $D$, exactly one outgoing 
arc, and add this arc as an edge. To formalize this process, we denote the sink vertex by $t$, 
and say that a function $f : C \setminus \{t\} \to C$ is an \emph{attachment function} for $D$ 
if $(a,f(a)) \in A$ is an arc 
of $D$ for every $a \in C \setminus \{t\}$. Thus, $f$ specifies one outgoing arc for each 
$a \in C \setminus \{t\}$. Given an attachment function $f$, we write $T(f)$ for the tree on $C$ 
with edge set
\[ 
\{ \{ a, f(a) \} : a \in C \setminus \{t\}  \}. 
\]
We now prove that every attachment function corresponds to a tree, 
and all trees in $\mathcal T(P)$ can be obtained in this way.

\begin{theorem}\label{thm:trick}
	Let $P$ be a profile that is single-peaked on some tree, and let $D$ be its attachment digraph. 
	Then $T \in \mathcal T(P)$ if and only if $T = T(f)$ for some attachment function~$f$.
	In other words, $P$ is single-peaked on a tree $T$ if and only if 
	the set of edges of $T$ is obtained by picking exactly one outgoing arc 
	for each non-sink vertex of $D$ and converting it into an undirected edge. 
\end{theorem}
\begin{proof}
	Suppose $T = T(f)$ for some attachment function $f$. Then $T$ is a possible output 
        of Algorithm~\ref{alg:trick}, for a suitable way of making the selections 
        from $B(a)$ for each vertex $a$ processed in the \emph{while}-loop. 
        Thus, by Theorem~\ref{thm:trick-is-correct}, the profile $P$ is single-peaked on $T$.
	
	We prove the converse by induction on $|C|$. If $|C| \le 2$, then $P$ is single-peaked 
        on the unique tree on $C$, which can be obtained as $T(f)$ for the unique attachment 
        function $f$. So suppose that $|C| \ge 3$, and that $T = (C, E)$ is a tree such that $P$ 
        is single-peaked on $T$. During the first iteration of 
        Algorithm~\ref{alg:attachment-digraph}, the algorithm determines the set $L_1$ of 
        candidates occurring in bottom position, and sets $C_2 = C \setminus L_1$. 
        By Proposition~\ref{prop:bottom-is-leaf}, each vertex in $L_1$ is a leaf of $T$. 
        Hence, the induced subgraph $T|_{C_2}$ is also a tree, and thus $P|_{C_2}$ 
        is single-peaked on $T|_{C_2}$. Also, by inspection of 
        Algorithm~\ref{alg:attachment-digraph}, the attachment digraph of $P|_{C_2}$ 
        is $D|_{C_2}$. By the inductive hypothesis, $T|_{C_2} = T(f')$ for some attachment 
        function $f'$ defined for $D|_{C_2}$. Thus, we can define an attachment function 
        $f$ so that for each $a \in C_2 \setminus \{t\}$ we set $f(a) = f'(a)$, 
        and for each $a \in L_1$ we set $f(a)$ to be the unique neighbor of $a$ in $T$. 
        By Corollary~\ref{cor:must-attach-to-B}, $T$ is obtained from $T|_{C_2}$ by attaching 
        each $a \in L_1$ to an element of $B(a)$, which implies that $f$ is a legal attachment 
        function. Thus, $T = T(f)$, which proves the claim.
\end{proof}

Using this characterization of the set $\mathcal T(P)$ and
noting that $T(f_1) \neq T(f_2)$ whenever $f_1 \neq f_2$,
we can conclude that the number of trees in $\mathcal T(P)$ 
is equal to the number of different attachment functions. 
This observation can be restated as follows.
 
\begin{corollary}\label{cor:counting}
	The number of trees in $\mathcal T(P)$ is equal to the product 
        of the out-degrees of the non-sink vertices of~$D$.
	Hence we can compute $|\mathcal T(P)|$ in polynomial time.
\end{corollary}

\begin{figure}
	\centering
		\begin{subfigure}[b]{0.4\linewidth}
		\centering
		\begin{tikzpicture}
		[	node distance=6mm and 11mm, inner sep=1.6pt]
		
		\node (b) {$b$};
		\node (a) [above left=of b] {$a$};
		\node (c) [below left=of b] {$c$};
		\node (d) [below right=of b] {$d$};
		\node (e) [above right=of b] {$e$};
		
		\path 
		(a) edge (b)
		(e) edge (b)
		(d) edge (b)
		(c) edge (b)
		;
		\end{tikzpicture}
	\end{subfigure}%
	\begin{subfigure}[b]{0.4\linewidth}
		\centering
		\begin{tikzpicture}
		[	node distance=12mm and 13mm, inner sep=2pt]
		
		\node (a) {$a$};
		\node (b) [right=of a] {$b$};
		\node (c) [right=of b] {$c$};
		\node (d) [right=of c] {$d$};
		\node (e) [above=of b] {$e$};
		
		\path 
		(a) edge (b)
		(e) edge (b)
		(d) edge (c)
		(c) edge (b)
		;
		\end{tikzpicture}
	\end{subfigure}
	\caption{The set $\mathcal{T}(P_2)$ of trees on which the profile $P_2$ from 
                 Example~\ref{ex:attachment-digraphs} is single-peaked consists of these two trees.
                 For the tree on the left, the attachment function has $f(d) = b$, 
                 while for the tree on the right, it has $f(d) = c$.}
	\label{fig:ex-all-trees}
\end{figure}
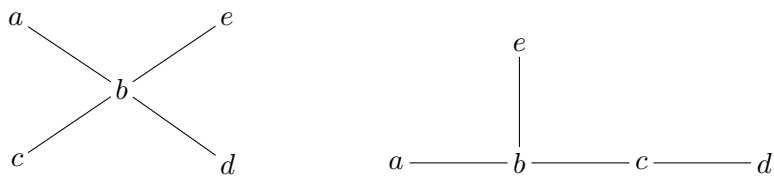

For the profiles in Example~\ref{ex:attachment-digraphs}, we see that $P_1$ is single-peaked on a 
unique tree (a path), that $P_2$ is single-peaked on exactly 2 trees (shown in 
Figure~\ref{fig:ex-all-trees}), and that $P_3$ is single-peaked on exactly 336 different trees.

It turns out that attachment digraphs have a lot of structure beyond the results of 
Proposition~\ref{prop:acyclic}. A key property, which will allow us to use essentially greedy algorithms, 
is what we call circumtransitivity.

\begin{definition}\label{def:circum}
	A directed acyclic graph $D=(C,A)$ is \emph{circumtransitive} if its vertices can be partitioned 
	into a set $\forced$ of \emph{forced} vertices and 
	a set $\free = C \setminus \forced$ of \emph{free} vertices 
	such that
	\begin{enumerate}
		\item every forced vertex $a \in \forced$ has out-degree at most 1, 
		      and if $(a,b) \in A$ then also $b \in \forced$, and
		\item every free vertex $a \in \free$ has out-degree at least 2, 
		      and whenever $a,b \in \free$ and $c \in C$ are such that $(a,b),(b,c)\in A$, 
	              then $(a,c)\in A$.
	\end{enumerate}
\end{definition}

\noindent
Intuitively, a circumtransitive digraph consists of an inner part (the \textit{forced part}), 
and an outer part, which is transitively attached to the inner part.

Recall that every directed acyclic graph $D$ has at least one sink. If $D$ is also 
circumtransitive, then its sink must be among the forced vertices.

\begin{theorem} \label{thm:circumtransitive}
	Every attachment digraph $(C, A)$ is circumtransitive.
\end{theorem}

\begin{proof}
	We will argue that the partition
	\[  \forced=\{a: d^+(a) \le 1\}, \qquad \free=\{a: d^+(a) \ge 2\}. \]
        satisfies the conditions in Definition~\ref{def:circum}.
	Suppose that the \emph{while}-loop of Algorithm~\ref{alg:attachment-digraph} 
        is executed $R-1$ times, partitioning the set $C$ as $L_1\cup\ldots L_{R-1}\cup L_R$,
	where $L_R:=C\setminus(L_1\cup\ldots\cup L_{R-1})$. 
	
	\smallskip
	\noindent\textit{Forced:} Let $a\in\forced$.
	If $d^+(a) = 0$, there is nothing to prove, so assume that 
	$d^+(a) = 1$, i.e., $(a,b)\in A$ for some $b\in C$. We will show that 
	$d^+(b) \in \{0,1\}$ and hence that $b \in \forced$. 
	
	If $b$ is a sink, we are done. 
	Otherwise, there exists an arc $(b,c)\in A$ for some $c\in C$.
	Suppose that $a \in L_r$ and $b \in L_s$ for some $1 \le r, s \le R$. 
        Note that $r < s\le R-1$ because neither $a$ nor $b$ are sinks, 
        and $(a,b)\in A$ (see Proposition~\ref{prop:arcs-point-from-i-to-j}).
	We will argue that $\text{top}(i, C_s)=b$ for some $i\in N$;
	this shows that $b$ has exactly one out-neighbor, as desired.

	Indeed, suppose this is not the case. Then $c\in B(b)$ implies 
	that $c\succ_i b$ for each $i\in N$. As $c\in C_r$, it cannot be the case
	that $\text{top}(i, C_r)=a$, $\text{second}(i, C_r)=b$ for some $i\in N$.
	Consequently, $B(i, C_r, a) = \{x\in C_r: x\succ_i a\}$ for each $i\in N$.
	As we have $B(a) = \{b\}$, it follows that $b\succ_i a$ for each $i\in N$.
	But then by transitivity $c\succ_i a$ for each $i\in N$, and hence $c\in B(a)$,
	a contradiction.
%
	
	\noindent\textit{Free:} 
	Consider vertices $a,b,c\in C$ with $a,b\in\free$ and $(a,b),(b,c)\in A$. 
	Since $a,b \in \free$, we have  $a,b \not\in L_R$ (indeed, recall that the out-degree
	of each vertex in $L_R$ is at most 1). 
	Thus, there exist $r, s$ with $1\le r<s< R$ such that $a \in L_r$ and $b \in L_s$.
	Note that if there was a voter $i \in N$ with $\text{top}(i, C_r) = a$, 
        then $|B(a)| = 1$, a contradiction with $d^+(a) > 1$. 
	Hence $\text{top}(i, C_r) \neq a$ for all $i \in N$.
	As $(a,b)\in A$, we have $b \in B(i, C_r, a)$ and therefore 
        $b \succ_i a$ for all $i\in N$. Similarly, since $(b,c)\in A$ and $d^+(b) > 1$, 
        we have $c \succ_i b$ for all $i \in N$. 
	Hence, by transitivity, $c \succ_i a$ for all $i\in N$. 
	Therefore $c\in \bigcap_{i\in N} B(i, C_r, a) = B(a)$, 
        and so $(a,c)\in A$, as desired.
\end{proof}

Suppose that $f$ is an attachment function for $D$.
Then for each forced vertex $a \in \forced \setminus \{t\}$, 
the value of $f(a)$ is uniquely determined, since $a$ has exactly one out-neighbor.
Note also that $D|_{\forced}$ is connected because we 
can reach the sink $t$ from every forced vertex. Hence, 
$\mathcal{G}(D|_{\forced})$ is a tree. 
It follows that for every $T \in \mathcal T(P)$, 
the tree $\mathcal{G}(D|_{\forced})$ is a subtree of $T$.

We will now study the free vertices $\free$ in more detail. The following 
proposition states that for every free vertex $a$, 
we can identify a pair of forced vertices that are adjacent in $D$ such that
$a$ can be attached to either of these vertices. 

\begin{proposition}\label{prop:twoforced}
	Suppose $|C| \ge 3$. For every free vertex $a \in \free$ of the attachment digraph $D=(C, A)$, there are two forced vertices $b, c \in \forced$ with 
	$(a,b), (a,c), (b,c) \in A$.
\end{proposition}
\begin{proof}
	Our proof proceeds in four steps. Let $a \in \free$ be a free vertex.
	
	\smallskip

	\noindent \emph{Step 1.} There is a forced vertex $b \in \forced$ with $(a,b) \in A$.
	
	\noindent \emph{Proof.} The directed acyclic graph $D$ has a unique sink $t$, so there exists a directed 
            path $a = c_1 \to c_2 \to \dots \to c_p = t$ from $a$ to $t$.   
	    Take such a path of minimum length. We will argue that $c_2$ 
            is a forced vertex. If $p=2$, we are done,
	    since in this case $(a, t) \in A$, and $t \in \forced$. 
            Now, suppose that $p \ge 3$. Suppose for the sake of contradiction 
            that $c_2$ is a free vertex. Then $c_1, c_2 \in \free$, and 
            $(c_1,c_2), (c_2,c_3) \in A$. Since $D$ is circumtransitive, 
            we have $(c_1,c_3) \in A$. But then $c_1 \to c_3 \to \dots \to c_p$ is a shorter 
            path from $a$ to $t$, a contradiction with our choice of the path.
	
	\smallskip

	\noindent\emph{Step 2.} There are at least two forced vertices $b, c \in \forced$ with $(a,b), (a,c) \in A$.
	
	\noindent \emph{Proof.} Let us say that a free vertex $a$ is {\em semi-forced}
            if there is a unique forced vertex $b \in \forced$ with $(a,b) \in A$; we need
	    to argue that the set of semi-forced vertices is empty. 
            Assume for the sake of contradiction that this is not the case,
            and consider the maximum value of $r$ such that $L_r$ contains a semi-forced vertex;
	    let $a$ be some semi-forced vertex in $L_r$. 
            As $a \in \free$, we have $d^+(a) \ge 2$, and so there exists a free vertex  
            $x \in \free$ with $(a,x) \in A$. Since $(a,x) \in A$, we have $x \in L_s$ 
            for some $s > r$. Now, $x$ is a free vertex, and $s>r$ implies that $x$
            is not semi-forced. Thus, there exist two distinct forced 
            vertices $y,z \in \forced$ with $(x,y), (x,z) \in A$. But then 
            by circumtransitivity we have 
            $(a,y)\in A$, $(a,z) \in A$, 
            which is a contradiction with our choice of $a$.

	\smallskip
	
	\noindent \emph{Step 3.} The set $\{b\in \forced: (a,b)\in A\}$ induces a subtree in $\mathcal G(D)$.
	
	\noindent \emph{Proof.} Consider a vertex $a \in L_r$, and suppose that $A$ contains arcs 
            $(a,b)$ and $(a,c)$ where $b,c\in\forced$.
	    Since $\mathcal{G}(D|_{\forced})$ is a tree, there is a unique path $Q$ 
            from $b$ to $c$ in $\mathcal{G}(D|_{\forced})$; 
            let $C_Q \subseteq \forced$ be the vertex set of this path.
	    Fix some tree $T\in\mathcal T(P)$. 
	    Then $\mathcal{G}(D|_{\forced})$ is a subgraph of $T$, and so $Q$ is a path in $T$.
	    Pick a voter $i \in N$. 
	    Since $b,c \in B(a)$, we have $|B(a)|>1$ and so $|B(i, C_r, a)| > 1$.
	    Hence, $a\neq\text{top}(i, C_r)$ and thus we have $b\succ_i a, c\succ_i a$.
	    Consider the top-initial segment of $i$'s vote given by 
            $W = \{ x\in C : x \succ_i a\}$.
	    By Proposition~\ref{prop:spt-equivalences}, since $P$ is single-peaked on $T$, 
            the set $W$ is connected in $T$.
	    Since $b, c \in W$, the path $Q$ must be contained in $T|_{W}$, 
            and hence $C_Q \subseteq W$.
	    Thus, $x\succ_i a$ for each $x\in C_Q$, and so $C_Q \subseteq B(i, C_r, a)$.
	    As this holds for every $i \in N$, we have $C_Q \subseteq B(a)$, 
            and so $C_Q \subseteq \{b\in \forced: (a,b)\in A\}$. 
            Hence, $\{b\in \forced: (a,b)\in A\}$ is connected in $\mathcal{G}(D|_{\forced})$.
	
	\smallskip

	\noindent\emph{Step 4.} There are two forced vertices $b, c \in \forced$ with $(a,b), (a,c), (b,c) \in A$.
	
	\noindent\emph{Proof.} The set $\{b\in \forced: (a,b)\in A\}$ is connected (by Step 3) and contains 
            at least two members (by Step 2). Hence, by definition of $\mathcal G$, 
            it contains some vertices $b$ and $c$ with $(b,c) \in A$.
\end{proof}

In the next section, we will use the properties of attachment digraphs established
in this section to develop algorithms that can check whether a given profile is single-peaked
on a tree that satisfies certain constraints. 

\section{Recognition Algorithms: Finding Nice Trees}\label{sec:recogn}
Suppose we are given a profile $P$ with $\mathcal T(P)\neq \emptyset$
and wish to find trees in $\mathcal T(P)$ that satisfy additional desiderata. 
In particular, we may want to find trees that can be used with the parameterized algorithms 
for the Chamberlin--Courant rule that we presented earlier.
We will now show how the attachment digraph can be used to achieve this.
We consider a variety of objectives, going beyond minimizing the number of internal vertices or the number of leaves.
These results may be useful for applications other than the computation of the Chamberlin--Courant rule.

We assume throughout this section that $|C| \ge 3$, 
since otherwise there is a unique tree $T$ on $C$, and the problem 
of selecting the best tree is trivial.

\subsection{Minimum Number of Internal Vertices}
In Section~\ref{sec:few-internal}, we saw an algorithm that could efficiently 
solve \textsc{Utilitarian CC} with the Borda scoring function
for profiles single-peaked on a tree $T$ with few internal vertices, 
where $T$ was taken as input to the algorithm. We now show how to find, given a profile $P$, 
a tree $T \in \mathcal T(P)$ that has the fewest internal vertices.
Algorithm~\ref{alg:min-internal} constructs an attachment function, and tries to make every 
vertex a leaf, if possible. In particular, every free vertex in the attachment digraph will 
become a leaf. We begin by showing that Algorithm~\ref{alg:min-internal} is 
well-defined, in the sense that whenever the algorithm chooses a vertex from a set, 
this set is non-empty, and whenever the attachment function is assigned a value, 
the respective arc is present in the attachment digraph.

\begin{algorithm}[th]
	\caption{Find $T\in\mathcal T(P)$ with fewest internal vertices}
	\begin{algorithmic}
		\State Let $D = (C,A)$ be the attachment digraph of $P$
		\State Let $\forced, \free$ be the sets of forced and free vertices in $D$
		\State Let $t$ be the sink vertex of $D$
		\State $f \gets \emptyset$, the attachment function under construction
		\For{each $a \in \forced \setminus \{t\}$}
			\State $f(a) \gets b$ where $b$ is the unique $b \in C$ with $(a,b) \in A$
		\EndFor
		\If{$|\forced| = 2$}
			\State pick some $c \in \forced$
			\For{each $a \in \free$}
				\State $f(a) \gets c$
			\EndFor
		\ElsIf{$|\forced| > 2$}
			\For{each $a \in \free$}
				\State find $c \in \forced$ such that $(a,c) \in A$ 
                                       and $c$ is internal in $\mathcal G(D|_{\forced})$
				\State $f(a) \gets c$
			\EndFor
		\EndIf
		\State \Return $T^* = T(f)$
	\end{algorithmic}
	\label{alg:min-internal}
\end{algorithm}

\begin{proposition}
	Algorithm~\ref{alg:min-internal} returns a tree $T^* \in \mathcal T(P)$.
\end{proposition}
\begin{proof}
	Our claim follows from Theorem~\ref{thm:trick} once we can show that the choices 
        of the algorithm are possible. Our running assumption that $|C| \ge 3$,
	combined with Proposition~\ref{prop:twoforced}, implies that $|\forced| \ge 2$. 
	
	Suppose that $|\forced| = 2$. By Proposition~\ref{prop:twoforced}, each $a\in\free$ is 
	adjacent to both vertices in $\forced$, and thus $(a,c) \in A$ irrespective 
        of which $c\in\forced$ is chosen by the algorithm. Thus the function $f$ is 
	an attachment function, i.e., the algorithm returns a tree in $\mathcal T(P)$. 

	On the other hand, suppose that $|\forced| > 2$. By Proposition~\ref{prop:twoforced}, 
	each free vertex $a \in \free$ has outgoing arcs to two forced vertices that 
	are adjacent in $\mathcal G(D|_{\forced})$. Since $|\forced| > 2$, 
	at most one of them can be a leaf in the tree $\mathcal G(D|_{\forced})$. 
	Hence, there is a $c \in \forced$ with $(a,c) \in A$ such that $c$ is internal 
	in $\mathcal G(D|_{\forced})$. Thus, the algorithm is well-defined in this case
	as well.
\end{proof}

Next, we show that Algorithm~\ref{alg:min-internal} returns an optimal tree.

\begin{proposition}
	Algorithm~\ref{alg:min-internal} runs in polynomial time and 
	returns a tree $T^* \in \mathcal T(P)$ with the minimum number 
	of internal vertices among trees in $\mathcal T(P)$. 
\end{proposition}
\begin{proof}
	The bound on the running time is immediate from the description of the algorithm.

	By Theorem~\ref{thm:trick} and the definition of $\forced$, 
	for every tree $T \in \mathcal T(P)$ we have 
	$\mathcal G(D|_{\forced}) \subseteq T$.
	Thus, if $a \in \forced$ is not a leaf in the tree 
	$\mathcal G(D|_{\forced})$, then $a$ cannot be a leaf in $T$.
	
	Suppose that $|\forced| = 2$. Since $|C| \ge 3$, we have $\free \neq \emptyset$. 
	Since the two members of $\forced$ are adjacent in any $T \in \mathcal T(P)$, 
	it cannot be the case that both of them are leaves in $T$. 
	Hence the number of leaves in $T \in \mathcal T(P)$ is at most $|\free| + 1$. 
	The tree $T^*$ has exactly $|\free| + 1$ leaves, and hence is optimal.

	On the other hand, suppose that $|\forced| > 2$. Note that every free vertex 
	$a \in \free$ is a leaf in $T^*$ because $f(a) \in \forced$ for all 
	$a \in C \setminus \{t\}$. Further, every leaf of $\mathcal G(D|_{\forced})$ 
	is also a leaf in $T^*$. By our initial observation, none of the remaining vertices 
	can be leaves in any $T \in \mathcal T(P)$, so $T^*$ has the maximum possible number 
	of leaves, and hence the minimum number of internal vertices.
\end{proof}

\subsection{Minimum Diameter}

It turns out that the tree found by Algorithm~\ref{alg:min-internal} is also optimal with respect 
to another metric: it minimizes the diameter.

\begin{proposition}
	Algorithm~\ref{alg:min-internal} returns a tree $T^* \in \mathcal T(P)$ with 
	the minimum diameter among trees in $\mathcal T(P)$. 
\end{proposition}
\begin{proof}
	Suppose that $|\forced| = 2$. Then $T^*$ is a star with center $c$; no tree on 
	three or more vertices has smaller diameter than a star.
	
	On the other hand, suppose that $|\forced| > 2$. In this case the diameter of $T^*$ 
	is equal to the diameter of $\mathcal G(D|_{\forced})$. 
	To see this, consider a longest path $(c_1,\dots,c_k)$ in $T^*$. 
	If $k = 2$, then $T^*$ is a star, which is a minimum-diameter tree 
	when there are $|C| \ge 3$ vertices. So suppose that $k \ge 3$.
	On a longest path, only $c_1$ and $c_k$ can be free vertices, 
	since all free vertices are leaves in $T^*$.
	Suppose $c_1 \in \free$. Then by construction of $T^*$ we have $c_2\in\forced$, 
	and $c_2$ is an internal vertex of $\mathcal G(D|_{\forced})$. 
	Hence, $c_2$ has at least two neighbors in 
	$\mathcal G(D)$ that are forced. 
	Thus, we can find a neighbor $c_1'$ of $c_2$ such that $c_1'$ is forced and $c_1' \neq c_3$. 
	Then we can replace $c_1$ by $c_1'$ in the longest path 
	(noting that $c_1'$ cannot appear elsewhere on the path 
	because  $\mathcal G(D|_{\forced})$ is a tree). 
	Similarly, if $c_k \in \free$, we can replace $c_k$ by a forced neighbor 
	of $c_{k-1}$. Having replaced all free vertices on the path 
	by forced vertices, we have obtained a longest path 
	in $T^*$ that is completely contained in  $\mathcal G(D|_{\forced})$. 
	Hence, the diameter of $T^*$ is equal to the diameter of $\mathcal G(D|_{\forced})$. 
	
	As $\mathcal G(D|_{\forced}) \subseteq T$ for every $T \in \mathcal T(P)$, the 
	diameter of any $T \in \mathcal T(P)$ must be at least the diameter of $\mathcal 
	G(D|_{\forced})$. Hence $T^*$ has the minimum diameter.
\end{proof}

\subsection{Minimum Number of Leaves}

In Section~\ref{sec:few-leaves}, we saw an algorithm for \textsc{Utilitarian CC} that is 
efficient when the input profile is single-peaked on a tree with few leaves. The algorithm 
assumed that the tree $T$ is part of the input. We will now describe a procedure that,  
given a profile $P$, finds a tree $T^* \in \mathcal T(P)$ with the minimum number of leaves.

\begin{algorithm}[th]
	\caption{Find $T\in\mathcal T(P)$ with fewest leaves}
	\begin{algorithmic}
		\State Let $D = (C,A)$ be the attachment digraph of $P$
		\State $f \gets \emptyset$, the attachment function under construction
		\State Let $t$ be the sink vertex of $D$, and let $s \in \forced$ be a forced vertex 
			with a unique outgoing arc $(s,t) \in A$ (such a vertex exists 
			by Proposition~\ref{prop:twoforced})
		\State $f(s) \gets t$
		\State Construct a bipartite graph $H$ with vertex set $L \cup R$ 
			where 
			$L = \{\ell_a : a \in C \setminus \{s,t\}\}$ and 
			$R = \{r_a : a \in C \}$, and edge set 
			$E_H = \{ \{ \ell_a, r_c\} : (a,c) \in A \}$
		\State Find a maximum matching $M \subseteq E_H$ in $H$
		\For{each $a \in C \setminus \{s, t\}$}
			\If{$\ell_a$ is matched in $M$, i.e. $\{\ell_a,r_c\} \in M$ for some $c \in C$}
				\State $f(a) \gets c$
			\Else
				\State take any $c \in C$ with $(a,c) \in A$
				\State $f(a) \gets c$
			\EndIf
		\EndFor
		\State \Return $T^* = T(f)$
	\end{algorithmic}
	\label{alg:min-leaves}
\end{algorithm}

Minimizing the number of leaves of a tree is equivalent to maximizing its number of internal 
vertices. Thus, to proceed, we first characterize the set of internal vertices of a tree $T(f)$.

\begin{proposition}
	\label{prop:f-internal}
	Let $f$ be an attachment function for the attachment digraph $D$. Then $a \in C \setminus \{t\}$ 
	is an internal vertex of the tree $T(f)$ if and only if $|f^{-1}(a)| \ge 1$, 
	i.e., $a$ is in the range of $f$. The sink vertex $t$ is an internal vertex of $T(f)$ 
	if and only if $|f^{-1}(t)| \ge 2$, i.e., there are two distinct vertices $a, b \in C$ 
	with $f(a) = f(b) = t$.
\end{proposition}
\begin{proof}
	A vertex is internal in a tree if and only if it has degree at least two. 
	From the definition of $T(f)$, for $a \in C \setminus \{t\}$, the degree of $a$ 
	is $1 + |f^{-1}(a)|$, and the degree of $t$ is $|f^{-1}(t)|$. 
	The claim follows immediately.
\end{proof}

Using this observation, we can prove that Algorithm~\ref{alg:min-leaves} returns an optimal tree. 
The algorithm is based on constructing a maximum matching.

\begin{proposition}
	Algorithm~\ref{alg:min-leaves} runs in polynomial time and returns 
	a tree $T^* \in \mathcal T(P)$ with the minimum number of leaves 
	among trees in $\mathcal T(P)$.
\end{proposition}
\begin{proof}
	The bound on the running time is immediate from the description of the algorithm.
	The algorithm constructs an attachment function, and hence by Theorem~\ref{thm:trick} 
	the output $T^*$ of the algorithm is a member of $\mathcal{T}(P)$.
	
	We will now argue that $T^*$ has the maximum number of internal vertices 
	among trees in $\mathcal T(P)$. By Proposition~\ref{prop:f-internal}, 
	it suffices to show that Algorithm~\ref{alg:min-leaves} 
	finds an attachment function $f$ that maximizes the number of vertices 
	$a$ with $|f^{-1}(a)| \ge 1$ if $a \neq t$ or $|f^{-1}(a)| \ge 2$ if $a = t$.

	First, we claim that under the attachment function $f$ constructed by 
	Algorithm~\ref{alg:min-leaves}, 
	a vertex $c \in C$ is an internal vertex of $T(f)$ if and only if 
	$r_c$ is matched in a maximum matching $M$. We start with the `if' direction.

	\begin{itemize}
		\item   Suppose $c = t$. If $r_t$ is matched in $M$ to $\ell_a$, 
			then both $a \in f^{-1}(t)$ and $s \in f^{-1}(t)$, 
			where $s$ is the vertex chosen at the very start of the algorithm. 
			By definition of the bipartite graph $H$, $a \neq s$, 
			and so $|f^{-1}(t)| \ge 2$, and hence $t$ is an internal vertex 
			by Proposition~\ref{prop:f-internal}.
		\item 	Suppose $c \neq t$. If $r_c$ is matched in $M$ to $\ell_a$, 
			then $a \in f^{-1}(c)$, and so $c$ is an internal vertex 
			by Proposition~\ref{prop:f-internal}. 
	\end{itemize}
	For the `only if' direction, suppose that $r_c$ is not matched in $M$. 
	Then in the \emph{for}-loop of Algorithm~\ref{alg:min-leaves}, we never 
	set $f(a) \gets c$ for any $a \in C \setminus \{t\}$, because otherwise 
	we could add the edge $\{ \ell_a, r_c \}$ to the matching $M$, contradicting 
	its maximality. Hence, if $c = t$ and $c$ is not matched, then $f^{-1}(t) = \{s\}$, 
	and so $t$ is not internal. If $c \neq t$ and $r_c$ is not matched, 
	then $f^{-1}(c) = \emptyset$, so $c$ is not internal.
	It follows that the number of internal vertices of $T(f)$ is $|M|$. 
	
	Now suppose that $T(f)$ is not optimal, and that $T' \in \mathcal T(P)$ is a tree with 
	$q > |M|$ internal vertices. By Theorem~\ref{thm:trick}, since $T' \in \mathcal T(P)$, 
	we have $T' = T(g)$ for some attachment function $g$. But then we can construct 
	a matching $M'$ in $H$ of size $|M'| = q$, as follows:
	\begin{itemize}
		\item 	If $t$ is an internal vertex in $T'$, then by Proposition~\ref{prop:f-internal}, 
			we have $|g^{-1}(t)| \ge 2$. Select some $a \in g^{-1}(t)$ with $a \neq s$, 
			and add $\{ \ell_a, r_t \}$ to $M'$. 
		\item 	For each $c \in C \setminus \{t\}$ that is an internal vertex of $T'$, 	
			select some $a \in g^{-1}(c)$ (which exists by Proposition~\ref{prop:f-internal}),
			and add $\{\ell_a, r_c\}$ to $M'$.
	\end{itemize}
	Clearly, we have added $q$ edges to $M'$. As $g$ is a function, $M'$ is a matching. 
	Since $|M'| > |M|$, we have a contradiction to the maximality of $M$.
\end{proof}

\subsection{Minimum Max-Degree}

\begin{figure}
	\centering
	\begin{tikzpicture}
	[n/.style={circle, fill=black!80, draw=none, inner sep=1.8pt},
	a/.style={circle, draw, inner sep=2.8pt},
	capa/.style={fill=white, font=\footnotesize, inner sep=0.3pt}]
	\node (slabel) at (-6, 0.35) {$x$};
	\node (tlabel) at (6, 0.35) {$z$};
	
	\node (c1label) at (-2, 2.35) {$\ell_{c_1}$};
	\node (c2label) at (-2, 1.35) {$\ell_{c_2}$};
	\node (c3label) at (-2, 0.35) {$\ell_{c_3}$};
	\node (c4label) at (-2, -0.65) {$\ell_{c_4}$};
	
	\node (c1-label) at (2, 2.35) {$r_{c_1}$};
	\node (c2-label) at (2, 1.35) {$r_{c_2}$};
	\node (c3-label) at (2, 0.35) {$r_{c_3}$};
	\node (c4-label) at (2, -0.65) {$r_{c_4}$};
	\node (c5-label) at (2, -1.65) {$r_t$};
	
	\node[n] (source) at (-6,0) {};
	\node[n] (sink) at (6,0) {};
	
	\node[n] (c1) at (-2,2) {};
	\node[n] (c2) at (-2,1) {};
	\node[n] (c3) at (-2,0) {};
	\node[n] (c4) at (-2,-1) {};
	
	\node[n] (c1-) at (2,2) {};
	\node[n] (c2-) at (2,1) {};
	\node[n] (c3-) at (2,0) {};
	\node[n] (c4-) at (2,-1) {};
	\node[n] (c5-) at (2,-2) {};
	
	\draw[-latex] (source) 
	edge node[capa] {$1$} (c1)
	edge node[capa] {$1$} (c2)
	edge node[capa] {$1$} (c3)
	edge node[capa] {$1$} (c4)
	
	(c1) edge node[capa] {$1$} (c2-)
	edge node[capa] {$1$} (c3-)
	(c2) edge node[capa] {$1$} (c4-)
	(c3) edge node[capa] {$1$} (c2-)
	(c4) edge node[capa] {$1$} (c3-)
	edge node[capa] {$1$} (c5-)
	
	;
	
	\draw[latex-] (sink) 
	edge node[capa] {$k-1$} (c1-)
	edge node[capa] {$k-1$} (c2-)
	edge node[capa] {$k-1$} (c3-)
	edge node[capa] {$k-1$} (c4-)
	edge node[capa] {$k$} (c5-);
	\end{tikzpicture}
	\caption{Flow network $H$ constructed by Algorithm~\ref{alg:degree}.}
	\label{fig:degree-flow}
\end{figure}
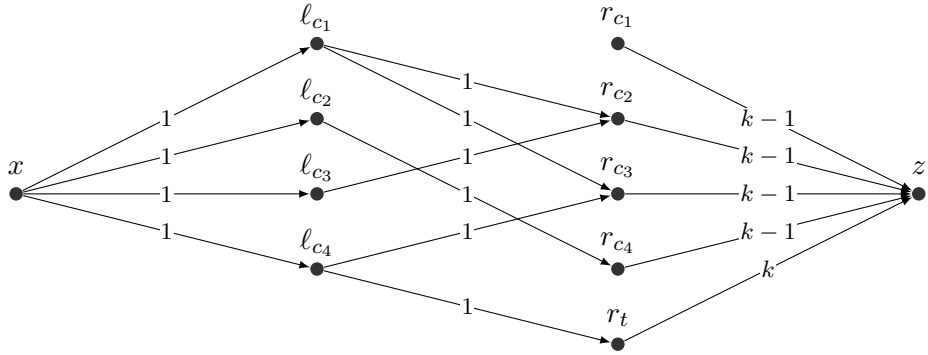

Another measure of tree complexity is its maximum degree. 
To minimize this quantity, we can use the following algorithm, 
which accepts as input a profile $P$ and a positive integer $k$,
and decides whether $P$ is single-peaked on some tree with maximum degree at most $k$. 
It is based on constructing a 
maximum flow network. An example network is shown in Figure~\ref{fig:degree-flow}.

\begin{algorithm}[th]
	\caption{Decide whether there is $T\in\mathcal T(P)$ with maximum degree at most $k$}
	\begin{algorithmic}
		\State Let $D = (C,A)$ be the attachment digraph of $P$
		\State Let $t$ be the sink vertex of $D$
		\State Let $L = \{\ell_a : a \in C \setminus \{t\} \}$ and $R = \{r_a : a \in C\}$ 
		and construct a flow network $H$ on vertex set $\{x,z\} \cup L \cup R$ with arc set 
		\[ E_H = \{ (x, \ell_a) : a \in C \setminus \{t\} \} \cup  
			 \{ (\ell_a, r_b) : a, b \in C, (a,b) \in A \} \cup 
			 \{ (r_a, z) : a \in C \}, 			 
		\] and capacities 
		$\text{cap}(x, \ell_a) = 1$ for all $a \in C \setminus \{t\}$, 
		$\text{cap}(\ell_a, r_b) = 1$ for all $(a,b) \in A$, 
		$\text{cap}(r_a, z) = k - 1$ for all $a \in C \setminus \{t\}$, and 
		$\text{cap}(r_t, z) = k$
		\State Find a maximum flow in $H$
		\State $f \gets \emptyset$, the attachment function under construction
		\If{the size of the maximum flow is $|C|-1$}
			\State For each $(a,b) \in A$ such that the arc $(\ell_a, r_b)$ is saturated, 
				set $f(a) \gets b$
			\State \Return $T^* = T(f)$
		\Else
			\State \Return there is no $T^* \in \mathcal T(P)$ with maximum degree at most $k$
		\EndIf
	\end{algorithmic}
	\label{alg:degree}
\end{algorithm}

\begin{proposition}
	Algorithm~\ref{alg:degree} runs in polynomial time and returns a tree $T^* \in \mathcal T(P)$ 
	with maximum degree at most $k$ if one exists.
\end{proposition}
\begin{proof}
	The bound on the running time is immediate from the description of the algorithm.

	Let $f$ be some attachment function. By definition of $T(f)$, for each $a \in C \setminus \{t\}$,
	the degree of $a$ in $T(f)$ is $1 + |f^{-1}(a)|$, because there is one edge in $T(f)$ 
	corresponding to an outgoing arc of $a$ in $D$, and $|f^{-1}(a)|$ edges in $T(f)$ 
	corresponding to incoming arcs of $a$ in $D$. Also, the degree of the sink vertex $t$ in $T(f)$ 
	is $|f^{-1}(t)|$. Thus, our task reduces to deciding whether there exists an attachment function 
	$f$ with 
	\begin{equation}
	\text{$1 + |f^{-1}(a)| \le k$ (i.e., $|f^{-1}(a)| \le k-1$) for each $a \in C\setminus \{t\}$ 
	and $|f^{-1}(t)| \le k$.}
	\label{eq:low-degree-f}
	\end{equation}
	
	Such attachment functions are in one-to-one correspondence with (integral) flows of size 
	$|C|-1$ in the flow network constructed by Algorithm~\ref{alg:degree}. Indeed, suppose $f$ is an 
	attachment function satisfying \eqref{low-degree-f}. Send one unit of flow from the 
	super-source $x$ along each of its $|C|-1$ outgoing links. 
	For each $a \in C \setminus \{t\}$, send the one unit of flow that arrives to $\ell_a$ 
	towards $r_{f(a)}$. Finally, for each $b \in C$, send the flow that arrives into $r_b$ towards 
	the super-sink $z$. This flow satisfies the capacity constraints because $f$ satisfies 
	\eqref{low-degree-f}. Conversely, any integral flow of size $|C|-1$ can be used
	to define an attachment function that satisfies \eqref{low-degree-f}: 
	for each $a\in C\setminus\{t\}$ there must be one unit of flow
	leaving $\ell_a$; we set $f(a)=b$, where $r_b$ is the destination of this flow.
	The resulting $f$ satisfies \eqref{low-degree-f} due to the capacity constraints of the links 
	between nodes in $R$ and the super-sink $z$.
\end{proof}

\subsection{Minimum Pathwidth}

Here, we show how to find a tree $T \in \mathcal T(P)$ of minimum pathwidth.
Our algorithm is based on an algorithm by 
\citet{scheffler1990linear}, 
which computes a minimum-width path decomposition of a given tree in linear time.

We need a preliminary result showing that a tree always admits a minimum-width path decomposition 
with a certain property: one endpoint of each edge appears in a bag of the path decomposition that has some 
`slack', in the sense that the bag does not have maximum cardinality.

\begin{lemma}
	\label{lem:transform-path-decomp}
	For every tree $T = (C, E)$, there exists a path decomposition $S_1, \dots, S_r$ of $T$ 
	of minimum width $w$ such that, for every edge $e \in E$, there is an endpoint $a \in e$ 
	for which there exists a bag $S_i$ with $a \in S_i$ such that $|S_i| \le w$ 
	(note that $\max_i |S_i| = w + 1$).
\end{lemma}
\begin{proof}
	
	We show how to transform an arbitrary path decomposition of $T$ into a path decomposition 
	of the same width having the desired property.
	
	Suppose $S_1, \dots, S_r$ is a path decomposition of $T$ with width $w$. For each edge 
	$\{a,b\}\in E$, we do the following: Because $\{a,b\}$ is an edge, there exists a bag containing 
	both $a$ and $b$ Consider the smallest value of $i \in \{1, \dots, r\}$ 
	such that $a,b \in S_i$.
	
	If $i = 1$, we set $S^*=S_i\setminus\{b\}$, and append the new bag $S^*$ to the 
	left of the sequence $S_1, \dots, S_r$. Then $S^*, S_1, \dots S_r$ is still a path decomposition
	of $T$, in this path decomposition $a$ appears in $S^*$, and $|S^*| < |S_1| \le w+1$, 
	so $|S^*|\le w$.
	
	If $i > 1$, then one of $a$ or $b$ does not appear in $S_{i-1}$. Assume without loss
	of generality that $b \not \in S_{i-1}$. Again, set $S^*=S_i\setminus\{b\}$, 
	and note that $|S^*|\le w$. Place the new bag $S^*$ in between $S_{i-1}$ and $S_i$. 
	Then $S_1, \dots, S_{i-1}, S^*, S_{i}, \dots, S_r$ is still a path decomposition
	of $T$, in this path decomposition $a$ appears in $S^*$, and $|S^*|\le w$.
\end{proof}

Clearly, the transformation described in the proof of Lemma~\ref{lem:transform-path-decomp} can be 
performed in polynomial time. Since one can find some path decomposition of a tree 
in polynomial time \cite{scheffler1990linear}, one can 
find a path decomposition with the property stated in Lemma~\ref{lem:transform-path-decomp} 
in polynomial time as well.

\begin{algorithm}[th]
	\caption{Find a tree $T\in\mathcal T(P)$ of minimum pathwidth}
	\begin{algorithmic}
		\State Let $D = (C,A)$ be the attachment digraph of $P$
		\State Let $S_1, \dots, S_r$ be a path decomposition of $\mathcal G(D|_{\forced})$ 
		       of minimum width $w$ that satisfies the condition of Lemma~\ref{lem:transform-path-decomp}
		\State $f \gets \emptyset$, the attachment function under construction
		\For{each $a \in C \setminus \{t\}$}
			\If{$a \in \forced$}
				\State $f(a) \gets b$, for the unique $b \in C$ with $(a,b) \in A$
			\ElsIf{$a \in \free$}
				\State Let $c_1, c_2 \in \forced$ be two forced vertices 
				       such that $(c_1,c_2), (a,c_1),(a,c_2) \in A$
				\State \hskip\algorithmicindent (these exist by 
					Proposition~\ref{prop:twoforced})
				\State Since $\{c_1, c_2\}$ is an edge of $\mathcal G(D|_{\forced})$, 
					by the condition of Lemma~\ref{lem:transform-path-decomp}, 
				\State \hskip\algorithmicindent there is a bag $S_j$ with $c_i \in S_j$ 
					and $|S_j| \le w$, for some $i \in \{1,2\}$
				\State $f(a) \gets c_i$
				\State Make a new bag $S_a = S_j \cup \{ a \}$ and place it to 
				       the right of $S_j$
				\State \hskip\algorithmicindent in the sequence of the path decomposition
			\EndIf
		\EndFor
		\State \Return $T^* = T(f)$
	\end{algorithmic}
	\label{alg:pathwidth}
\end{algorithm}

\begin{proposition}
	Algorithm~\ref{alg:pathwidth} returns a tree $T^* \in \mathcal T(P)$ with minimum pathwidth 
        among trees in $\mathcal T(P)$ in polynomial time.
\end{proposition}
\begin{proof}
	Note first that the forced part $\mathcal G(D|_{\forced})$ is a tree and hence its 
	minimum-width path decomposition can be computed efficiently.

	Next we claim that the path decomposition constructed by Algorithm~\ref{alg:pathwidth} 
        is in fact a path decomposition of the output tree $T(f)$.
	Indeed, each free vertex $a \in \free$ becomes a leaf in $T(f)$, and only occurs in a single bag 
	$S_a$ in the constructed path decomposition. Since $a$ is a leaf, there is only one edge of $T$ 
	that contains it (namely, $\{ a, f(a) \}$), and we have $a, f(a) \in S_a$. 
	Also, since $a$ only occurs in a single bag, the set of bags containing $a$ 
	is trivially an interval of the path decomposition sequence. 
	
	Next, observe that the path decomposition of $T(f)$ has the same width $w$ as the pathwidth 
	of the forced part $\mathcal G(D|_{\forced})$, because all new bags have cardinality 
	at most $w + 1$.
	Now, because $\mathcal G(D|_{\forced})$ is a subgraph of every $T \in \mathcal T(P)$, 
	no tree in $\mathcal T(P)$ can have a smaller pathwidth than $\mathcal G(D|_{\forced})$. 
	Since Algorithm~\ref{alg:pathwidth} identifies a tree $T \in \mathcal T(P)$ 
	with the same pathwidth as $\mathcal G(D|_{\forced})$, it outputs an optimal solution.
\end{proof}

\subsection{Other Graph Types}

To conclude this section, we explain how to recognize 
whether $\mathcal T(P)$ contains trees of certain types.

\paragraph{Paths}
The literature 
contains several algorithms for recognizing profiles that are single-peaked on a path. The 
algorithms by \citet{doignon1994polynomial} and \citet{escoffier2008single} can be 
implemented to run in time $O(mn)$. One could also use some of the algorithms presented 
above. Algorithm~\ref{alg:min-leaves} finds a tree $T \in \mathcal T(P)$ with a minimum 
number of leaves. Clearly, if $\mathcal T(P)$ contains a path, then this will be discovered 
by the algorithm. Alternatively, Algorithm~\ref{alg:degree} can be used to look for a tree $T \in 
\mathcal T(P)$ with maximum degree $k = 2$; it will succeed if and only if $P$ is 
single-peaked on a path. However, 
both Algorithm~\ref{alg:min-leaves} and Algorithm~\ref{alg:degree} 
depend on pre-computing the attachment digraph, which takes time $O(m^2n)$. Thus,  
an attachment digraph-based approach would necessarily be slower than 
the linear-time algorithms from previous work.

\paragraph{Stars}
In Proposition~\ref{prop:spt-star}, we observed that a profile is single-peaked on a star 
graph if and only if there is a candidate $c \in C$ such that every voter ranks $c$ in either 
first or second position. This condition can easily be verified in $O(n)$ time, without 
needing to compute the attachment digraph. Note that Algorithm~\ref{alg:min-internal} 
(minimizing the number of internal vertices) will output a star whenever $\mathcal T(P)$ 
contains a star graph.

\paragraph{Caterpillars}
Caterpillar graphs are exactly the trees of pathwidth 1 \cite{proskurowski1999classes}, 
and so Algorithm~\ref{alg:pathwidth} can check whether a profile is single-peaked on a caterpillar.
In fact, one can use an even simpler algorithm: it suffices
to compute $\mathcal G(D|_{\forced})$ and check that it is a caterpillar. 
Indeed, if not, then no tree in $\mathcal T(P)$ can be a caterpillar. 
On the other hand, if $\mathcal G(D|_{\forced})$ is a caterpillar
then Algorithm~\ref{alg:min-internal} 
finds a caterpillar graph in $\mathcal T(P)$.
To see this, recall that this algorithm 
attaches every free vertex as a leaf to an internal vertex 
of $\mathcal G(D|_{\forced})$.

\paragraph{Subdivision of a Star}
A tree is a subdivision of a star if at most one vertex has degree 3 or higher. We can find a 
subdivision of a star in $\mathcal T(P)$, should one exist, by adapting Algorithm~\ref{alg:degree}: 
we guess the center of the subdivision of the star, and then assign suitable upper bounds on 
the vertex degrees by appropriately setting the capacity constrains in the flow network.

\section{Hardness of Recognizing Single-Peakedness on a Specific Tree}
The algorithms presented in Section~\ref{sec:recogn} enable us to answer a wide range of questions about 
the set $\mathcal T(P)$. However the following NP-hardness result shows that not 
every such question can be answered efficiently unless P = NP.

Two graphs $G_1 = (V_1, E_1)$ and $G_2 = (V_2, E_2)$ are said to be \emph{isomorphic} if there is a 
bijection $\phi : V_1 \to V_2$ such that for all $u,v \in V_1$, it holds that $\{ u, v\} \in 
E_1$ if and only if $\{ \phi(u), \phi(v) \} \in E_2$; we write $G_1\cong G_2$ 
whenever this is the case.
We consider the following computational 
problem.

\begin{center}
\begin{tabular*}{0.95\linewidth}{ll}
	\toprule
	\multicolumn{2}{l}{\textsc{Single-Peaked Tree Labeling}} \\
	\midrule
	\textit{Instance:} & Profile $P$ over $C$, a tree $T_0$ on $|C|$ vertices \\
	\textit{Question:} & Is there a tree $T = (C,E)$ such that $T\cong T_0$ and $P$ is single-peaked on $T$? \\
	\bottomrule
\end{tabular*}
\end{center}

In this problem, we are given a `template' unlabeled tree $T_0$, and need to decide whether 
we can label the vertices in this template by candidates so as to make the input profile 
single-peaked on the resulting labeled tree. For example, if $T_0$ is a path, then the 
problem is to decide whether the profile $P$ is single-peaked on a path, and in this case the 
problem is easy to solve. However, the template $T_0$ occurs in the input to the 
decision problem, and it is not clear how to proceed if we would like to check 
whether $T_0$ `fits' into the attachment digraph. 
In fact, as we now show, this problem is NP-complete.

\begin{theorem}\label{thm:giventree}
	The problem \textsc{Single-Peaked Tree Labeling} is \textup{NP}-complete 
	even if $T_0$ is restricted to diameter at most four.
\end{theorem}

\begin{proof}
	The problem is in NP since, having guessed a tree $T$ and an isomorphism $\phi$, 
  	we can easily check that $\phi$ is an isomorphism between $T$ and $T_0$ 
	and that the input profile is single-peaked on $T$.
	
	For the hardness proof, we reduce \textsc{Exact Cover by 3-Sets (X3C)} to our problem.
	An instance of X3C is given by a ground set
	$X$ and a collection $\mathcal Y$ of size-3 subsets of $X$. 
        It is a `yes'-instance if there is a subcollection $\mathcal Y'\subseteq \mathcal Y$ 
        of size $|X|/3$
	such that each element of $X$ appears in exactly one set in $\mathcal Y'$.
	This problem is NP-hard \cite{gj}.

	Suppose we are given an X3C-instance with ground set $X = \{x_1,\dots,x_p\}$, 
	where $p = 3p'$ for some positive integer $p'$, 
	and a collection $\mathcal Y = \{Y_1,\dots,Y_q\}$ of 3-element subsets of $X$. 
	We then construct an instance of \textsc{Single-Peaked Tree Labeling} as follows.
	First, we construct a tree $T_0$ 
	with vertex set $C_0 = \{ a, b_1, \dots, b_{q - p'}, c_1, \dots, c_{p'}\} \cup \{d_{i,j} : 1 \le i \le p', 1 \le j \le 3 \}$,
	and edge set 
	$E_0 = \{ \{ a, b_i \} : 1 \le i \le q - p' \} \cup \{ \{ a, c_i \} : 1 \le i \le p' \} \cup \{ \{c_i, d_{i,j} \} : 1 \le i \le p', 1 \le j \le 3 \}$.
	The resulting tree is drawn below. It has $3p'+p'+1+ (q-p') = p+q+1$ vertices and 
	diameter 4.
	
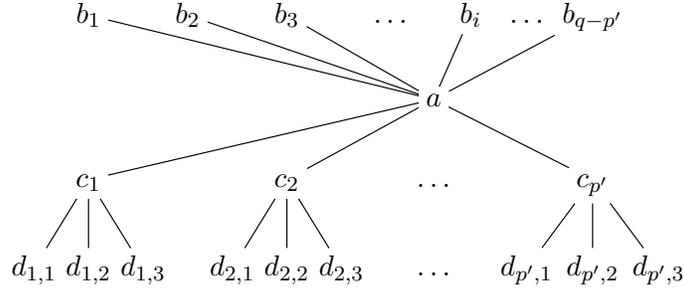
\[\xymatrix@C-2.2pc@R-0.5pc{
		& b_1 \ar@{-}[drrrrrrr] && b_2 \ar@{-}[drrrrr] && b_3 \ar@{-}[drrr] && \dots && b_i \ar@{-}[dl] && \dots & b_{q-p'} \ar@{-}[dllll] \\
		&&&&&&&& a \ar@{-}[dlllllll] \ar@{-}[dlll] \ar@{-}[drrrr] \\
		& c_1 \ar@{-}[dl] \ar@{-}[d] \ar@{-}[dr] 
		&&&& c_2 \ar@{-}[dl] \ar@{-}[d] \ar@{-}[dr]
		&&& \dots 
		&&&& c_{p'} \ar@{-}[dl] \ar@{-}[d] \ar@{-}[dr] \\
		d_{1,1} & d_{1,2} & d_{1,3}
		&&
		d_{2,1} & d_{2,2} & d_{2,3}
		&& \dots &&&
		d_{p', 1} & d_{p',2} & d_{p',3}
	}\]
	Next, we construct a profile $P$ with $|N| = p+q$ voters on the
	candidate set $C = \{ z, x_1, \dots, x_p, y_1, \dots, y_q \}$.
	$P$ will contain one vote for each object in $X$ and one vote for each set in $\mathcal Y$. 
	In the following, all indifferences can be broken arbitrarily.
For each object $x_i$, we add a voter $v_{x_i}$:
\[ z \succ \{ y_j : Y_j \ni x_i \} \succ x_i \succ \{ y_j : Y_j \not\ni x_i \} \succ X \setminus \{x_i\}. \]
Intuitively, the presence of this voter will force $x_i$ to be attached to $z$ or to 
a candidate that corresponds to a set containing $x_i$.
For each set $Y_j$, we add a voter $v_{Y_j}$:
\[ z \succ y_j \succ \{ y_\ell : 1\le \ell \le q, \ell\neq j \} \succ X. \]
The presence of this voter will force an edge from $z$ to  $y_j$.

This completes the description of the reduction. We now prove that it is correct.

Suppose the given X3C-instance is a `yes'-instance, and let $\mathcal Y'$
be a cover consisting of $p'$ sets. Renumbering the elements and sets if necessary, 
we can assume that $\mathcal Y'=\{Y_1, \dots, Y_{p'}\}$ and 
$Y_j=\{x_{3j-2}, x_{3j-1}, x_{3j}\}$ for each $j=1, \dots, p'$.
Then we build a labeling isomorphism $\phi : C_0 \to C$ as follows: 
We set $\phi(a) = z$.
For each $j=1, \dots, p'$, we set $\phi(c_j)=y_j$, 
and $\phi(d_{j, k})=x_{3(j-1)+k}$ for $k=1, 2, 3$.
Also, for each $j=1, \dots, q-p'$, we set $\phi(b_j)=y_{p'+j}$.
Note that $\phi$ is a bijection because $\mathcal Y'$ is an exact cover.
The resulting labeled tree $T$ is shown below.
It is easy to check that the profile $P$ is single-peaked on $T$.

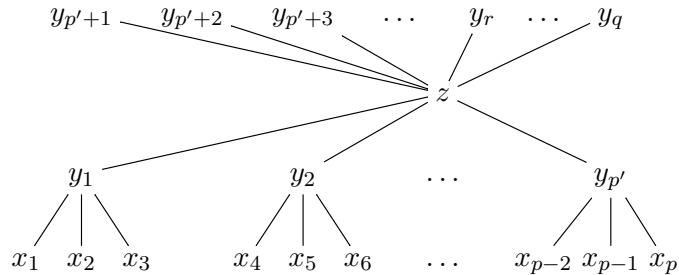
\[\xymatrix@C-2.2pc@R-0.5pc{
		& y_{p'+1} \ar@{-}[drrrrrrr] && y_{p'+2} \ar@{-}[drrrrr] && y_{p'+3} \ar@{-}[drrr] && 
                   \dots && y_{r} \ar@{-}[dl] && \dots & y_{q} \ar@{-}[dllll] \\
		&&&&&&&& z \ar@{-}[dlllllll] \ar@{-}[dlll] \ar@{-}[drrrr] \\
		&y_{1} \ar@{-}[dl] \ar@{-}[d] \ar@{-}[dr] 
		&&&& y_{2} \ar@{-}[dl] \ar@{-}[d] \ar@{-}[dr]
		&&& \dots 
		&&&& y_{p'} \ar@{-}[dl] \ar@{-}[d] \ar@{-}[dr] \\
		x_{1} & x_{2} & x_{3}
		&&
		x_{4} & x_{5} & x_{6}
		&& \dots &&&
		x_{p-2} & x_{p-1} & x_{p}
	}\]
	
	Conversely, suppose that there is a tree $T$ isomorphic to $T_0$ 
        such that $P$ is single-peaked on $T$.
	Let $\phi : C_0 \to C$ is a witnessing isomorphism.
	Note that the vertex $z$ of $T$ must have degree at least $q$, 
        because for each $j \in \{1, \dots, q\}$, voter $v_{Y_j}$ 
        can only be single-peaked on $T$ if $z$ and $y_j$ are adjacent in $T$.
	There is only one such vertex in $T_0$, namely $a$, and hence $\phi(a) = z$. 
	The vertex $z$ of $T$ has exactly $q$ neighbors, which then must all be labeled by some $y_j$. 
	Exactly $p'$ of the $q$ neighbors of $z$ have degree 4.
	Let $\mathcal Y' = \{ Y_j \in \mathcal Y : y_j = \phi(c_i) \text{ for some } 1 \le i \le p' \}$ 
        be the collection of the $p'$ sets corresponding to candidates 
	that occupy the vertices of degree 4.
	We claim that $\mathcal Y'$ is a cover.
	Let $x_i \in X$. The vertex labeled with $x_i$ must be a leaf of $T$ 
        because all internal vertices of $T$ have already been labeled otherwise. 
        Then, because $v_{x_i}$ is single-peaked on $T$, the set 
        $\{z,x_i\} \cup \{ y_j : Y_j \ni x_i \}$ must be connected in $T$, 
        so the neighbor of $x_i$ must be a member of that set. 
        But $x_i$ cannot be a neighbor of $z$, so $x_i$ is a neighbor of some $y_j$ where $x_i \in Y_j$. 
        This implies that $y_j$ is the label of a degree-4 vertex. Hence $Y_j \in \mathcal Y'$, 
        and so $x_i$ is covered by $\mathcal Y'$.
\end{proof}
By creating multiple copies of the center vertex and adding some peripheral vertices, 
we can adjust this reduction to show that {\sc Single Peaked Tree Labeling} remains hard even if 
each vertex of $T_0$ has degree at most three (we omit the somewhat tedious proof).

In the appendix, by modifying the reduction in the proof of Theorem~\ref{thm:giventree}, 
we show that it is also NP-complete to decide 
whether a given preference profile is single-peaked on a \emph{regular} tree, i.e., a tree 
where all internal vertices have the same degree (Theorem~\ref{thm:hard-reg}). 
This hardness result stands in contrast to the many easiness results of Section~\ref{sec:recogn}.

\section{Conclusions}\label{sec:concl}
Without any restrictions on the structure of voters' preferences, 
winner determination under the Chamberlin--Courant rule is NP-hard.
Positive results have been obtained when preferences are assumed to be single-peaked, 
and we studied whether these results can be extended to preferences that are single-peaked on a tree.
For the egalitarian variant of the rule, we showed that this is indeed the case:
a winning committee can be computed in polynomial time for any tree and any scoring function.
For the utilitarian setting, we show that winner determination is hard for general preferences 
single-peaked on a tree, but we find positive results when imposing additional restrictions. 
One algorithm we present runs in polynomial time when preferences are single-peaked on a tree 
which has a constant number of leaves, and another runs efficiently on a tree with a small number 
of internal vertices. Interestingly, the two algorithms are, in some sense, incomparable.
Specifically, the former algorithm works for all scoring functions, 
while for the latter algorithm this is not the case (though it does work for the most common scoring function, 
i.e., the Borda scoring function). On the other hand, the latter algorithm establishes 
that computing a winning committee is in FPT with respect 
to the combined parameter `committee size, number of internal vertices', 
while the former algorithm does not establish fixed-parameter tractability with respect
to the combined parameter `committee size, number of leaves'.
An open question is whether there exist FPT algorithms or W[1]-hardness results for the parameters `number of internal vertices' and the combined parameter `committee size, number of leaves'.

It would be interesting to see whether our easiness results for preferences 
that are single-peaked on a tree
extend to the egalitarian version of the Monroe rule \cite{monroe}, where each committee member
has to represent approximately the same number of voters. 
\citet{betzler2013computation} show that this rule becomes 
easy for preferences single-peaked on a path, but their argument for that rule 
is much more intricate than for egalitarian Chamberlin--Courant.

To make our parameterized winner determination algorithms applicable, 
we have investigated the problem of deciding whether the input profile 
is single-peaked on a `nice' tree, for several notions of `niceness'.
To this end, we have proposed a new data structure, namely, the attachment 
digraph, which compactly encodes the set $\mathcal T(P)$ of all trees $T$ 
such that the profile $P$ is single-peaked on $T$,
and showed how to use it to identify trees with desirable properties.
In particular, we showed how to find a tree in $\mathcal T(P)$ with the minimum number of leaves 
or the minimum number of internal vertices, to be used with our winner
determination algorithms. To demonstrate the power of our framework, 
we also designed efficient algorithms
for several other notions of `niceness', such as small diameter, 
small maximum degree and small pathwidth. However, there are also notions
of `niceness' that defy this approach: we show that it is NP-hard to decide
whether an input profile is single-peaked on a regular tree.
Another interesting measure of `niceness' is vertex deletion distance to a path, 
i.e., the number of vertices that need to be deleted from a tree so that the remaining 
graph is a path. In particular, this parameter is relevant for the elicitation
results of \citet{dey2016elicitation}. Finding a tree in $\mathcal T(P)$ 
that minimizes this parameter is equivalent to finding a tree with the maximum diameter, 
which is closely related to the problem of finding a maximum-length path 
in the attachment digraph. However, we were not able to design a polynomial-time 
algorithm for this problem (or to show that it is computationally hard).
Similarly, the complexity of finding a tree in $\mathcal T(P)$ that has
the minimum path cover number (another parameter considered by \citet{dey2016elicitation})
is an open problem. 

It would also be interesting to explore the parameterized complexity of the 
problems related to the identification of `nice' trees. 
For instance, {\sc Single-Peaked Tree Labeling} is trivially fixed-parameter tractable
with respect to the number of candidates, as we could explore all possible labelings;
however, its parameterized complexity with respect to the number of voters is an
interesting open problem. In a similar vein, we can ask if we can find a tree
in $\mathcal T(P)$ with approximately minimal vertex deletion distance to a path 
or approximately minimal path cover number. Indeed, constant-factor approximation
algorithms for these problems can still be used in conjunction with the elicitation
algorithms of \citet{dey2016elicitation} in order to improve over elicitation algorithms
for unstructured preferences.

\section*{Acknowledgments}
This research was supported by 
National Research Foundation (Singapore) under grant RF2009-08
(Edith Elkind, Lan Yu), by the European Research Council (ERC) under grant 639945 (ACCORD) (Edith Elkind, Dominik Peters), by EPSRC (Dominik Peters), and by
EAPSI and National Science Foundation (USA) under grant OISE-1209805 (Hau Chan).

A preliminary version of the results in this paper appeared in two conference papers 
\cite{yu2013multiwinner,pe16}. In this paper, we unify the exposition of these papers, 
provide more intuition, and give pseudocode and detailed analysis of  several algorithms 
that were only sketched in the conference versions.  
\citet{yu2013multiwinner} leave as an open question whether the Chamberlin--Courant rule 
on profiles single-peaked on a tree remains NP-hard for the Borda scoring function. 
We answer this question in Theorem~\ref{thm:hard-borda}.

We would like to thank the anonymous JAIR reviewers for their careful reading of the paper
and many useful suggestions, which ranged from catching typos to simplifying some of the proofs.

\appendix

\section{Hardness of Utilitarian Chamberlin--Courant for Low-Degree Trees}

Here we modify the reduction in the proof of Theorem~\ref{thm:hard-borda} 
to establish that {\sc Utilitarian CC} remains hard on trees of maximum degree~3.

\begin{theorem}\label{thm:hard-deg3}
	Given a profile $P$ that is single-peaked on a tree with maximum degree 3, 
        a target committee size $k$, and a target score $B$, it is \textup{NP}-complete 
        to decide whether there exists a committee of size $k$ with score at least $B$ 
        under the utilitarian Chamberlin--Courant rule with the Borda scoring function.
\end{theorem}
\begin{proof}
	We will provide a reduction from the classic \textsc{Vertex Cover} problem.
	Given an instance $(G, t)$ of \textsc{Vertex Cover} such that $G=(V, E)$,
	$V=\{u_1, \dots, u_p\}$ and $E=\{e_1, \dots, e_q\}$,
        we construct an instance 
	of \textsc{Utilitarian CC} as follows. 

	Let $M=5p^2q$; intuitively, $M$ is a large number.
	We introduce three candidates $a_i$, $y_i$ and $z_i$ 
        for each vertex $u_i \in V$, and $M$ dummy candidates. 
        Formally, we set $A=\{a_1, \dots, a_p\}$, $Y=\{y_1, \dots, y_p\}$, $Z=\{z_1, \dots, z_p\}$, 
        $D = \{ d_1,\dots,d_M \}$, and define the candidate set to be
	$C= A \cup Y \cup Z\cup D$.
	We set the target committee size to be $k = p + t$.
	
	We now introduce the voters, who will come in three types: $N = N_1 \cup N_2 \cup N_3$.
	
	\[\def\arraystretch{1.3}
	\begin{array}{cccccccccccccccccccccc}
	\toprule
	\multicolumn{3}{c}{N_1} & & \multicolumn{3}{c}{N_2} & & \multicolumn{3}{c}{N_3} \\
	\cmidrule{1-3} \cmidrule{5-7} \cmidrule{9-11} 
	5pq & \cdots & 5pq      && 1 & \cdots & 1 && M & \cdots & M \\
	\cmidrule{1-3} \cmidrule{5-7} \cmidrule{9-11} 
	y_1 & 	& y_p      && A              && A               &&  z_1   &&   z_p  \\
	z_1 &   & z_p      && y_{j_{1,1}}    && y_{j_{q,1}}     &&  y_1   &&   y_p   \\
	A   &   & A        && y_{j_{1,2}}    && y_{j_{q,2}}     &&  A     && A \\
	D   &   & D        && D              && D               &&  D     && D \\
	\smash\vdots       && \smash\vdots   && \smash\vdots    &&  \smash\vdots     && \smash\vdots 
        && \smash\vdots  \\   
	\bottomrule
	\end{array}
	\]
	
	\begin{itemize}
		\item The set $N_1$ consists of $5pq$ identical voters for each $u_i\in V$: 
                      they rank $y_i$ first, $z_i$ second, and $a_i$ third, followed by 
                      other candidates:
		      \[ 
                      y_i \succ z_i \succ a_i \succ a_{i+1} \succ \dots \succ a_p 
                          \succ a_{i-1} \succ \dots \succ a_1 
                          \succ d_1 \succ \dots \succ d_M \succ \cdots 
                      \]
		\item The set $N_2$ consists of a single voter for each edge $e_j\in E$: 
                      this voter ranks candidates in $A$ first (as $a_1 \succ\dots\succ a_p$), 
                      followed by the two candidates 
                      from $Y$ that correspond to the endpoints of $e_j$ (in an arbitrary order),
                      followed by the dummy candidates $d_1,\dots,d_M$, 
                      followed by all other candidates as specified below.
                      The purpose of these voters is to ensure that every edge is covered 
                      by one of the vertices that correspond to a committee member, 
                      and to incur a heavy penalty of $M$ if the edge is uncovered.
		\item The set $N_3$ is a set of $M$ identical voters for each $u_i\in V$ 
                      who all rank $z_i$ first and $y_i$ second:
		      \[ 
                      z_i \succ y_i \succ a_i \succ a_{i+1} \succ \dots \succ a_p \succ a_{i-1} 
                          \succ \dots \succ a_1 \succ d_1 \succ \dots \succ d_M \succ \cdots 
                      \]
	\end{itemize}
	
	\noindent
	We complete the voters' preferences so that the resulting profile is 
        single-peaked on the following tree:
	
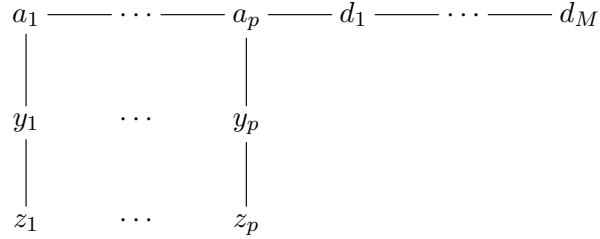
\[
	\xymatrix{
		a_1 \ar@{-}[d] \ar@{-}[r] & \cdots \ar@{-}[r] & a_p \ar@{-}[d] 
                    \ar@{-}[r] & d_1 \ar@{-}[r] & \cdots \ar@{-}[r] & d_M \\
		y_1 \ar@{-}[d] & \cdots & y_p \ar@{-}[d]  \\
		z_1 & \cdots & z_p
	}
	\]
	This tree is obtained by starting with a path through $A$ and $D$, 
        and then attaching $y_i$ as a leaf onto $a_i$ and $z_i$ as a leaf onto $y_i$
        for every $i = 1,\dots,p$. 
        Note that the resulting tree has maximum degree 3. 
        It remains to specify how to complete each vote in our profile to
        ensure that the resulting profile is single-peaked on this tree. 
        Inspecting the tree, we see that it suffices to ensure that 
        for each $i=1, \dots, p$ it holds that in all votes where the positions 
        of $y_i$ and $z_i$ are not given explicitly, candidate $y_i$ is ranked above~$z_i$.

	We will again reason about costs rather than scores.
	We set the upper bound on cost to be $B=(5pq)(p-t)+q(p+1)$ (note that by construction, $M > B$). 
	This completes the description of our instance of the \textsc{Utilitarian CC} 
	problem with the Borda scoring function $\vecs = (0, -1, -2, \dots)$.
	Intuitively, the `correct committee' we have in mind consists of all $z_i$ candidates 
    (of which there are $p$) and a selection of $y_i$ candidates that corresponds 
    to a vertex cover (of which there should be $t$), should a vertex cover of size $t$ exist. 
		Now let us prove that the reduction is correct.

	Suppose we have started with a `yes'-instance of \textsc{Vertex Cover}, 
        and let $S$ be a collection of $t$ vertices that forms a vertex cover of $G$. 
	Consider the committee $W=Z\cup\{y_i : u_i\in S\}$. Note that $|W| = p + t = k$.
	The voters in $N_3$ and $5pqt$ voters in $N_1$ have their most-preferred 
        candidate in $W$, so they contribute $0$ to the cost of $W$.
	For the remaining $(5pq)(p-t)$ voters in $N_1$, their contribution to the cost of $W$ is $1$, 
        since $z_i\in W$ for all $i$. 
	Further, each voter in $N_2$ contributes at most $p+1$ to the cost.
	Indeed, the candidates that correspond to the endpoints of the 
        respective edge are ranked in positions $p+1$ and $p+2$
        in this voter's ranking, and since $S$ is a vertex cover for $G$,
	one of these candidates is in $S$. We conclude that 
        $\cost^+_\mu(P, W)\le (5pq)(p-t)+q(p+1) = B$.
	
	Conversely, suppose there exists a committee $W$ of size $k = p+t$
	with $\cost^+_\mu(P, W)\le B$.
	Note first that $W$ has to contain all candidates in $Z$: otherwise, 
        there are $M$ voters in $N_3$
	with cost at least $1$, and then the utilitarian Chamberlin--Courant cost of $W$
	is at least $M > B$, a contradiction. Thus $Z\subseteq W$. 
	We will now argue that $W\setminus Z$ is a subset of $Y$, 
	and that $S' = \{u_i : y_i\in W\setminus Z\}$ is a vertex cover of $G$.
	Suppose that $W\setminus Z$ contains too few candidates from $Y$, 
        i.e., at most $t-1$ candidates from $Y$. Then $N_1$ contains at least $(5pq)(p-(t-1))$
	voters who contribute at least $1$ to the cost of $W$, so 
	$\cost^+_\mu(P, W)\ge (5pq)(p-t+1) > (5pq)(p-t) + q(p+1) = B$, a contradiction.
	Thus, we have $W\setminus Z\subseteq Y$. 
	Now, suppose that $S'$ is not a vertex cover for~$G$.
	Let $e_j\in E$ be an edge that is not covered by $S'$, 
	and consider the voter in $N_2$ corresponding to $e_j$. Clearly, 
	none of the candidates ranked in positions $1, \dots, p+2+M$ by this voter
	appear in $W$. Thus, this voter contributes 
	more than $M$ to the cost of $W$, so the total cost of $W$ is more than $M > B$, 
        a contradiction.
	Thus, a `yes'-instance of  \textsc{Utilitarian CC} corresponds
	to a `yes'-instance of \textsc{Vertex Cover}.
\end{proof}

\section{Hardness of Utilitarian Chamberlin--Courant for Stars}
For the Borda scoring function, we have seen in Theorem~\ref{thm:hard-borda} that 
\textsc{Utilitarian CC} is NP-complete for trees of diameter 4, but in 
Section~\ref{sec:few-internal} we have argued that this problem is easy for stars, i.e.,  
trees of diameter 2. The algorithm that worked for stars 
uses specific properties of the Borda scoring function. In this section, 
we show that for some positional scoring functions 
\textsc{Utilitarian CC} remains hard even on stars.

Recall from Proposition~\ref{prop:spt-star} that a profile $P$ is single-peaked on a star if and only if there is a candidate $c$ such that for every $i \in N$, either $\text{top}(i) = c$ or $\text{second}(i) = c$.

\begin{theorem}\label{thm:hard-star}
	\textsc{Utilitarian CC} is \textup{NP}-hard even for profiles
	that are single-peaked on a star. The hardness result holds
	for any family of positional scoring functions
	whose scoring vectors $\vecs$ satisfy 
	$s_1=0$, $s_2=\dots=s_\ell<0$, $s_{\ell+1}<s_\ell$ for some $\ell\ge 5$.   
\end{theorem}
\begin{proof}
	We will reduce from the restricted version of \textsc{Exact Cover by 3-Sets (X3C)}
	to our problem. Recall that an instance of X3C is given by a ground set
	$X$ and a collection $\mathcal Y$ of size-3 subsets of $X$; 
        it is a `yes'-instance if there is a subcollection $\mathcal Y'\subseteq \mathcal Y$ 
        of size $|X|/3$
	such that each element of $X$ appears in exactly one set in $\mathcal Y'$.
	This problem is NP-hard. Moreover it remains NP-hard even if each element
        of $X$ appears in at most three sets in $\mathcal Y$ \cite{Gonzalez85}.
	
	Fix a family of positional scoring functions $\mu$ that satisfy 
	the condition in the statement of the theorem for some $\ell\ge 5$.
	For brevity, we write $s=s_2$ and $S=s_{\ell+1}$. 
	
	Given an instance $(X, \mathcal{Y})$ of \textsc{X3C} such that 
	$X=\{x_1, \dots, x_p\}$, $p=3p'$, $\mathcal Y=\{Y_1, \dots, Y_q\}$, 
	and $|\{Y_j\in\mathcal{Y} : x_i\in Y_j\}|\le 3$ for each $x_i\in X$, we construct an instance 
	of our problem as follows. 
	We set $Y=\{y_1, \dots, y_q\}$, create $p$ sets of dummy candidates $D_1, \dots, D_p$
	of size $\ell$ each, and let
	$C=\{a, z\}\cup X \cup Y\cup D_1\cup\ldots \cup D_p$. 
	
	We now introduce the voters, who will come in three types $N = N_1 \cup N_2 \cup N_3$.
	
	For each $Y_j\in \mathcal{Y}$ we construct $p+1$ voters who rank $y_j$ first, $a$ second
	and $z$ third, followed by all other candidates in an arbitrary order;
	let $N_1$ denote the set of voters constructed in this way.
	Further, for each $x_i\in X$, we construct a voter who ranks $x_i$ first, followed
	by $a$, followed by the candidates $y_j$ such that $x_i\in Y_j$ (in an arbitrary order), 
	followed by candidates in $D_i$, followed by all other candidates 
	in an arbitrary order. Denote the resulting set of voters by $N_2$.
	Finally, let $N_3$ be a set of $(p+1)(q+1)$ voters who all rank $z$
	first and $a$ second, followed by all other candidates in an arbitrary order.
	Set $B=s((p+1)(q-p')+p)$ and $k=p'+1$. 
	This completes the description of our instance
	of the \textsc{Utilitarian CC} problem. Observe that every voter in $N$
	ranks $a$ second, so by Proposition~\ref{prop:spt-star} the constructed profile
	in single-peaked on a star. 
	
		\[\def\arraystretch{1.3}
	\begin{array}{ccccccccccccccccccccc}
		\toprule
		\multicolumn{3}{c}{N_1} & & \multicolumn{3}{c}{N_2} & & \multicolumn{1}{c}{N_3} \\
		\cmidrule{1-3} \cmidrule{5-7} \cmidrule{9-9} 
		p+1 & \cdots & p+1      && 1 & \cdots & 1 && (p+1)(q+1) \\
		\cmidrule{1-3} \cmidrule{5-7} \cmidrule{9-9} 
		y_1 &              & y_q          && x_1              && x_p                  &&  z  \\
		a &              & a          && a    && a        &&  a  \\
		z   &              & z            && y_{j_{1,1}}    && y_{j_{p,1}}        &&  \smash\vdots\\
		\smash\vdots &     & \smash\vdots &&  y_{j_{1,2}}    && y_{j_{p,2}}             \\
		&     & &&  y_{j_{1,3}}    && y_{j_{p,3}}             \\
		&     &  && d_1            && d_1               \\
		&         &                       && \smash\vdots   && \smash\vdots        \\[-2pt]
		&         &                       && d_M            && d_M                \\[-2pt]
		&         &                       && \smash\vdots   && \smash\vdots       \\
		\bottomrule
	\end{array}
	\]
	
	Suppose we have started with a `yes'-instance of \textsc{X3C}, and let $\mathcal{Y}'$
	be a collection of $p'$ subsets in $\mathcal Y$ that cover $X$. 
	Set $W=\{z\}\cup\{y_j : Y_j\in\mathcal{Y}'\}$.
	Clearly, all voters in $N_3$ and $(p+1)p'$ voters in $N_1$ are perfectly represented
	by $W$ and so contribute 0 to the score of $W$. 
	The remaining voters in $N_1$ 
	contribute $s = s_3$ to the score of $W$, since $z\in W$. 
	Further, each voter in $N_2$
	contributes at least $s = s_5$ to the score of $W$. 
	Indeed, consider a voter who ranks some $x_i\in X$ first.
	Then he ranks the candidates $y_j$ such that $x_i\in Y_j$ in positions $5$ 
	or higher. Since $\mathcal{Y}'$ is a cover of $X$, at least one of these candidates
	appears in $W$. We conclude that $\score^+_\mu(P, W)\ge s((p+1)(q-p')+p) = B$.
	
	Conversely, suppose there exists a committee $W$ of size $p'+1$
	such that $\score^+_\mu(P, W)\ge B$.
	Note first that $W$ has to contain $z$: otherwise, each voter in $N_3$
	would contribute at most $s$ to the score, and the total score of $W$ 
	would  be at most $s|N_3| < B$. 
	We will now argue that $W\setminus\{z\}$
	is a subset of $Y$ and that $\mathcal{Y}'' = \{Y_j : y_j\in W\setminus \{z\}\}$ 
	is an exact cover of $X$.
	Suppose first that $W\setminus\{z\}$ contains at most $p'-1$
	candidates from $Y$. Then $N_1$ contains at least $(p+1)(q-p'+1)$
	voters who contribute at most $s$ to the score of $W$, so 
	$\score^+_\mu(P, W)\le s(p+1)(q-p'+1) < B$, a contradiction.
	Thus, we have $W\setminus\{z\}\subseteq Y$. Now, suppose that
	$\mathcal{Y}''$ is not an exact cover of $X$, 
	let $x_i$ be an element of $X$ that is not covered by $\mathcal{Y}''$, 
	and consider the voter in $N_2$ that ranks $x_i$ first. Clearly, 
	none of the candidates ranked in positions $1, \dots, \ell$ by this voter
	appear in $W$. Thus, this voter 
	contributes at most $S$ to the score of $W$. All other voters in $N_2$ 
	contribute at most $s$ to the score of $W$. Further, there are $(p+1)(q-p')$ voters
	in $N_1$ who are not perfectly represented by $W$. We conclude that the total
	score of $W$ is at most 
	$s((p+1)(q-p')+(p-1))+S < B$, a contradiction.
	Thus, a `yes'-instance of  \textsc{Utilitarian CC} corresponds
	to a `yes'-instance of \textsc{X3C}.
\end{proof}

\section{Hardness of Recognizing Preferences Single-Peaked on Regular Trees}
Recall that a tree is $k$-regular if every non-leaf vertex has degree $k$.

\begin{theorem}\label{thm:hard-reg}
	Given a profile $P$, it is \textup{NP}-complete to decide whether $P$ is single-peaked on a regular tree, i.e., whether there exists 
	a positive integer $k$ such that $P$ is single-peaked on a 
	$k$-regular tree. The problem is also hard for each fixed $k\ge 4$.
\end{theorem}

\begin{proof}
	The problem is in NP since for a given $k$-regular tree $T$ we can easily check whether 
	$P$ is single-peaked on $T$.
	
	We start by giving a hardness proof for fixed $k=4$, and later explain how to modify the reduction 
	for other fixed $k$ and for non-fixed $k$.
	
	Again, we reduce from X3C. Suppose that we are given an X3C-instance with 
	ground set $\{x_1,\dots,x_p\}$, $p = 3p'$, and a collection of subsets 
	$\mathcal Y = \{Y_1,\dots,Y_q\}$.
	We construct a profile over the following candidates. 
	Let 
	\[
	X = \{ x_1, \dots, x_p\}, \quad 
	S = \{s_0, s_1, \dots, s_q, s_{q+1}\}, \quad 
	L = \{\ell_1,\dots, \ell_q\}, \quad 
	Y = \{ y_1, \dots, y_q \}. 
	\]
Then our candidate set is 
	$X\cup S\cup L \cup Y$; that is, there is one candidate per element $x_i$, 
	three candidates $s_j,\ell_j,y_j$ for each set $Y_j$, 
	and two further candidates $s_0$ and $s_{q+1}$. 
	The candidate $y_j$ represents the set $Y_j$. 
	We now introduce the voters; the reader may find it helpful 
	to look at the tree in Figure~\ref{fig:tree-hard-reg} to understand the 
	intuition behind the construction.
	
	First, we force $(s_0, s_1, \dots, s_q, s_{q+1})$ to form a path. 
	For this, we need to force that $(s_j, s_{j+1})$ is an edge, for each $j = 0,\dots,q$.
	To this end, for each $j=0, \dots, q$ we create a voter who ranks the candidates as
	\[ s_j \succ s_{j+1} \succ s_{j+2} \succ\dots\succ s_{q+1}\succ s_{j-1}\succ \dots\succ 
	s_1\succ s_0\succ Y\succ L\succ X. \]
	We also force $s_0$ and $s_{q+1}$ to be leaves. This requires introducing two more voters:
	\[ s_1\succ\dots\succ s_{q+1}\succ Y\succ L\succ X\succ s_0, \]
	\[ s_0\succ\dots\succ s_{q}  \succ Y\succ L\succ X\succ s_{q+1}. \]
	Further, we force each $\ell_j \in L$ and each $x_i \in X$ to be a leaf.
	That is, for each $j=1, \dots, q$ we introduce a voter who ranks the candidates as
	\[ S\succ Y\succ L \setminus\{\ell_j\}\succ X\succ \ell_j,  \]
	and for each $i=1, \dots, p$ we introduce a voter who ranks the candidates as
	\[ S\succ Y\succ L\succ X \setminus\{x_i\}\succ x_i.  \]
	Here (and below) we write $S$ as shorthand for $s_0 \succ s_1 \succ \dots \succ s_{q+1}$.
	
	Now, for each $j=1, \dots, q$ the vertices $\ell_j$ and $y_j$ need to have an edge to $s_j$.
	Thus, for each $j=1, \dots, q$ we introduce two voters who rank the candidates as
	\[ \ell_j\succ s_j\succ s_{j+1}\succ \dots\succ s_{q+1}\succ s_{j-1}\succ \dots\succ s_0\succ 
		Y\succ L \setminus \{\ell_j\}\succ X, \]
	\[ y_j\succ s_j\succ s_{j+1}\succ \dots\succ s_{q+1}\succ s_{j-1}\succ \dots\succ s_0\succ 
		Y  \setminus \{y_j\}\succ L\succ X. \]
	Finally, for each $i=1, \dots, p$, we introduce a voter whose role is to ensure that
	$x_i$ is attached to a vertex $y_j$ such that $x_i\in Y_j$ 
	(or to an element of $S$, but this will never happen):
	\begin{equation}
	S\succ \{ y_j : x_i \in Y_j \}\succ x_i\succ 
        \{y_j: x_i\not\in Y_j\}\succ L\succ X \setminus \{x_i\}. \label{eqn:4-reg-attach}
	\end{equation}
	
	This concludes the description of the reduction. We will now prove that it is correct.
	
	Suppose that a collection $\mathcal Y' = \{Y_{j_1},\dots,Y_{j_{p'}}\}$ 
	forms an exact cover of $X$. 
	Then the constructed profile is single-peaked on the 4-regular tree 
	on the candidate set with the following set of edges $E$ (see Figure~\ref{fig:tree-hard-reg}): 
	$\{s_j, s_{j+1}\}$ for all $j = 0, \dots, q$; 
	$\{\ell_j, s_j\}$ for all $j=1, \dots, q$; 
	$\{y_j, s_j\}$ for all $j=1, \dots, q$; 
	$\{x_i, y_k\}$ for all $x_i\in X$, where $Y_k$ is the set in $\mathcal Y'$ 
	that contains $x_i$. 
	By considering top-initial segments of the votes given above, 
	we see that this choice makes all votes single-peaked on this tree.
		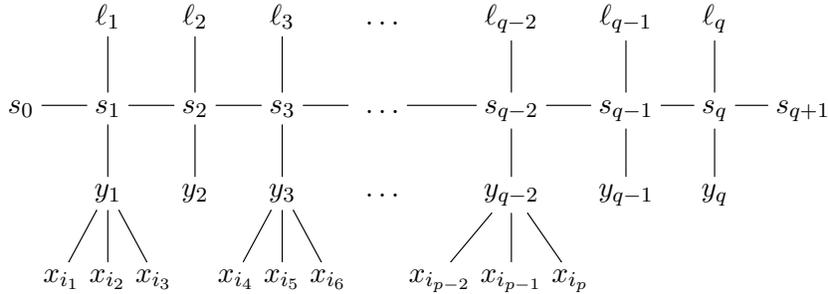
\begin{figure}[ht]
	\centering
	\[\xymatrix@C-2.2pc@R-0.5pc{
		&& \ell_1 \ar@{-}[d] && \ell_2 \ar@{-}[d] && \ell_3\ar@{-}[d] &&& \dots && 
		\ell_{q-2}\ar@{-}[d] &\quad& \ell_{q-1}\ar@{-}[d] &\quad& \ell_q\ar@{-}[d]\\
		s_0\ar@{-}[rr]
		&&s_1\ar@{-}[rr] \ar@{-}[d]
		&&s_2\ar@{-}[rr] \ar@{-}[d]
		&&s_3\ar@{-}[rr] \ar@{-}[d]
		&&& \dots \ar@{-}[rr]
		&& s_{q-2} \ar@{-}[rr] \ar@{-}[d]
		&& s_{q-1} \ar@{-}[rr] \ar@{-}[d]
		&& s_q  \ar@{-}[rr] \ar@{-}[d]
		&\quad& s_{q+1}\\
		&&y_1 \ar@{-}[dl] \ar@{-}[d] \ar@{-}[dr] 
		&& y_2
		&& y_3 \ar@{-}[dl] \ar@{-}[d] \ar@{-}[dr]
		&&& \dots 
		&& y_{q-2} \ar@{-}[dl] \ar@{-}[d] \ar@{-}[dr] 
		&& y_{q-1}
		&& y_q \\
		&x_{i_1} & x_{i_2} & x_{i_3}
		&&
		x_{i_4} & x_{i_5} & x_{i_6}
		&&  
		&	x_{i_{p-2}} & x_{i_{p-1}} & x_{i_{p}}
	}\]
\caption{Tree $T$ constructed in the proof of Theorem~\ref{thm:hard-reg}.}
\label{fig:tree-hard-reg}
\end{figure}

	Conversely, suppose there is a 4-regular tree $T = (C,E)$ such that all votes 
	are single-peaked on $T$. We show that in this case our instance of X3C admits
	an exact cover.
	When introducing the voters, 
	we argued that for any such tree it must be the case that
	$\{s_j, s_{j+1}\} \in E$, for all $j = 0, \dots, q$ and 
	$\{\ell_j, s_j\}, \{y_j, s_j\}\in E$ for all $j = 1, \dots, q$. 
	Further, we know that each of the vertices in $L\cup X \cup \{s_0, s_{q+1}\}$ is a leaf of $T$.
	However, we do not yet know how the vertices in $X$ are attached to the rest of the tree.
	As no vertex can be attached to a leaf, each vertex $x_i\in X$	
	is attached to some vertex in $Y\cup S\setminus\{s_0, s_{q+1}\}$.
	Now, each vertex $s_j$ in $S\setminus\{s_0, s_{q+1}\}$ already has four neighbors in $T$ 
	(namely, $s_{j-1}, s_{j+1}, \ell_j, y_j$); 
	thus, since $T$ is 4-regular, $x_i$ cannot be a neighbor of $s_j$. 
	Thus $x_i$'s neighbor is some $y_j\in Y$. 
	Now, consider the voter whose preferences are given by (\ref{eqn:4-reg-attach}); 
	for this voter's preferences to be single-peaked on $T$, 
	$x_i$ must be attached to a vertex $y_j$ that corresponds to a set $Y_j \ni x_i$. 
	On the other hand, by 4-regularity, each $y_j$ is connected to either 0 or 3 
	vertices in $X$. Hence this tree encodes an exact cover of $X$.
	
	For other fixed values of $k\ge 5$, we can perform essentially the same reduction 
	from the problem 
	\textsc{exact cover by $(k-1)$-sets}, which is also NP-hard;%
\footnote{For example, one can reduce from X3C to X4C by taking an X3C instance with $3p$ elements, adding $p$ dummy elements $d_1, \dots ,d_p$, and replacing each 3-set $Y$ by the $p$ 4-sets $Y \cup \{d_1\}, \dots Y \cup \{d_p\}$.}	
	the only modification 
	we need to make is to use $k-3$ copies of the set $L$.
	
	If the value of $k$ is not fixed, we can do the following (see picture): 
	Prepend $s_{-1}$ to the path $S$ (where now $s_{-1}$ is forced to be a leaf, 
	but $s_0$ is not), and introduce new leaves $a$ and $b$ that must be attached to $s_0$. 
	Modify the votes given by (\ref{eqn:4-reg-attach}) in such a way that $x_i$ can be attached to 
	$s_1,\dots,s_q$ or an appropriate $y_j$, but not to $s_0$. Then any regular tree 
	on which the profile is single-peaked will have to be 4-regular, 
	due to $s_0$ having degree 4, and the argument above goes through.
	\[\xymatrix@C-1.5pc@R-0.5pc{
		&& a \ar@{-}[d]
		&& \ell_1 \ar@{-}[d] &&\dots\\
		s_{-1} \ar@{-}[rr]
		&&s_0\ar@{-}[rr] \ar@{-}[d]
		&&s_1\ar@{-}[rr] \ar@{-}[d]
		&&\dots\\
		&&b
		&&y_1 \ar@{-}[dl] \ar@{-}[d] \ar@{-}[dr] 
		&& \dots\\
		&&&x_{i_1} & x_{i_2} & x_{i_3}
		&&
	}\]
\end{proof}

\vskip 0.2in
\bibliography{spt}

\begin{thebibliography}{}

\bibitem[\protect\BCAY{Arrow, Sen,\ \BBA\ Suzumura}{Arrow
  et~al.}{2002}]{arrow-handbook}
Arrow, K., Sen, A., \BBA\ Suzumura, K.\BEDS. \BBOP2002\BBCP.
\newblock {\Bem Handbook of Social Choice and Welfare, Volume 1}.
\newblock Elsevier.

\bibitem[\protect\BCAY{Bartholdi, Tovey,\ \BBA\ Trick}{Bartholdi
  et~al.}{1989}]{dodgson-hard1}
Bartholdi, J., Tovey, C., \BBA\ Trick, M. \BBOP1989\BBCP.
\newblock \BBOQ Voting schemes for which it can be difficult to tell who won
  the election\BBCQ\
\newblock {\Bem Social Choice and Welfare}, {\Bem 6\/}(2), 157--165.

\bibitem[\protect\BCAY{{{Bartholdi}}\ \BBA\ Trick}{{{Bartholdi}}\ \BBA\
  Trick}{1986}]{bar-tri:j:sp}
{{Bartholdi}}, III, J.\BBACOMMA\  \BBA\ Trick, M.~A. \BBOP1986\BBCP.
\newblock \BBOQ Stable matching with preferences derived from a psychological
  model\BBCQ\
\newblock {\Bem Operation Research Letters}, {\Bem 5\/}(4), 165--169.

\bibitem[\protect\BCAY{Betzler, Slinko,\ \BBA\ Uhlmann}{Betzler
  et~al.}{2013}]{betzler2013computation}
Betzler, N., Slinko, A., \BBA\ Uhlmann, J. \BBOP2013\BBCP.
\newblock \BBOQ On the computation of fully proportional representation\BBCQ\
\newblock {\Bem Journal of Artificial Intelligence Research}, {\Bem 47\/}(1),
  475--519.

\bibitem[\protect\BCAY{Black}{Black}{1958}]{black-book}
Black, D. \BBOP1958\BBCP.
\newblock {\Bem The Theory of Committees and Elections}.
\newblock Cambridge University Press.

\bibitem[\protect\BCAY{Bodlaender}{Bodlaender}{1994}]{bodlaender1994tourist}
Bodlaender, H.~L. \BBOP1994\BBCP.
\newblock \BBOQ A tourist guide through treewidth\BBCQ\
\newblock {\Bem Acta Cybernetica}, {\Bem 11\/}(1--2), 1--21.

\bibitem[\protect\BCAY{Brandt, Brill, Hemaspaandra,\ \BBA\ Hemaspaandra}{Brandt
  et~al.}{2015}]{brandt2015bypassing}
Brandt, F., Brill, M., Hemaspaandra, E., \BBA\ Hemaspaandra, L.~A.
  \BBOP2015\BBCP.
\newblock \BBOQ Bypassing combinatorial protections: Polynomial-time algorithms
  for single-peaked electorates\BBCQ\
\newblock {\Bem Journal of Artificial Intelligence Research}, {\Bem 53},
  439--496.

\bibitem[\protect\BCAY{Chamberlin\ \BBA\ Courant}{Chamberlin\ \BBA\
  Courant}{1983}]{chamberlin}
Chamberlin, J.\BBACOMMA\  \BBA\ Courant, P. \BBOP1983\BBCP.
\newblock \BBOQ Representative deliberations and representative decisions:
  {P}roportional representation and the {B}orda rule\BBCQ\
\newblock {\Bem American Political Science Review}, {\Bem 77\/}(3), 718--733.

\bibitem[\protect\BCAY{Clearwater, Puppe,\ \BBA\ Slinko}{Clearwater
  et~al.}{2015}]{clearwater2015generalizing}
Clearwater, A., Puppe, C., \BBA\ Slinko, A. \BBOP2015\BBCP.
\newblock \BBOQ Generalizing the single-crossing property on lines and trees to
  intermediate preferences on median graphs\BBCQ\
\newblock In {\Bem Proceedings of the 24th International Joint Conference on
  Artificial Intelligence (IJCAI)}, \BPGS\ 32--38.

\bibitem[\protect\BCAY{Conitzer}{Conitzer}{2009}]{conitzer2009eliciting}
Conitzer, V. \BBOP2009\BBCP.
\newblock \BBOQ Eliciting single-peaked preferences using comparison
  queries\BBCQ\
\newblock {\Bem Journal of Artificial Intelligence Research}, {\Bem 35},
  161--191.

\bibitem[\protect\BCAY{Conitzer, Derryberry,\ \BBA\ Sandholm}{Conitzer
  et~al.}{2004}]{conitzer2004combinatorial}
Conitzer, V., Derryberry, J., \BBA\ Sandholm, T. \BBOP2004\BBCP.
\newblock \BBOQ Combinatorial auctions with structured item graphs\BBCQ\
\newblock In {\Bem Proceedings of the 19th National Conference on Artificial
  Intelligence (AAAI)}, \lowercase{\BVOL}~4, \BPGS\ 212--218.

\bibitem[\protect\BCAY{Constantinescu\ \BBA\ Elkind}{Constantinescu\ \BBA\
  Elkind}{2021}]{andrei}
Constantinescu, A.~C.\BBACOMMA\  \BBA\ Elkind, E. \BBOP2021\BBCP.
\newblock \BBOQ Proportional representation under single-crossing preferences
  revisited\BBCQ\
\newblock In {\Bem Proceedings of the 35th AAAI Conference on Artificial
  Intelligence (AAAI)}, \BPGS\ 5286--5293.

\bibitem[\protect\BCAY{Cornaz, Galand,\ \BBA\ Spanjaard}{Cornaz
  et~al.}{2012}]{cor-gal-spa:c:spwidth}
Cornaz, D., Galand, L., \BBA\ Spanjaard, O. \BBOP2012\BBCP.
\newblock \BBOQ Bounded single-peaked width and proportional
  representation\BBCQ\
\newblock In {\Bem Proceedings of the 20th European Conference on Artificial
  Intelligence (ECAI)}, \BPGS\ 270--275.

\bibitem[\protect\BCAY{Danilov}{Danilov}{1994}]{danilov1994structure}
Danilov, V.~I. \BBOP1994\BBCP.
\newblock \BBOQ The structure of non-manipulable social choice rules on a
  tree\BBCQ\
\newblock {\Bem Mathematical Social Sciences}, {\Bem 27\/}(2), 123--131.

\bibitem[\protect\BCAY{Demange}{Demange}{1982}]{demange1982single}
Demange, G. \BBOP1982\BBCP.
\newblock \BBOQ Single-peaked orders on a tree\BBCQ\
\newblock {\Bem Mathematical Social Sciences}, {\Bem 3\/}(4), 389--396.

\bibitem[\protect\BCAY{Dey\ \BBA\ Misra}{Dey\ \BBA\
  Misra}{2016}]{dey2016elicitation}
Dey, P.\BBACOMMA\  \BBA\ Misra, N. \BBOP2016\BBCP.
\newblock \BBOQ Elicitation for preferences single peaked on trees\BBCQ\
\newblock In {\Bem Proceedings of the 25th International Joint Conference on
  Artificial Intelligence (IJCAI)}, \BPGS\ 215--221.

\bibitem[\protect\BCAY{Doignon\ \BBA\ Falmagne}{Doignon\ \BBA\
  Falmagne}{1994}]{doignon1994polynomial}
Doignon, J.-P.\BBACOMMA\  \BBA\ Falmagne, J.-C. \BBOP1994\BBCP.
\newblock \BBOQ A polynomial time algorithm for unidimensional unfolding
  representations\BBCQ\
\newblock {\Bem Journal of Algorithms}, {\Bem 16\/}(2), 218--233.

\bibitem[\protect\BCAY{Elkind, Lackner,\ \BBA\ Peters}{Elkind
  et~al.}{2017}]{structure-survey}
Elkind, E., Lackner, M., \BBA\ Peters, D. \BBOP2017\BBCP.
\newblock \BBOQ Structured preferences\BBCQ\
\newblock In Endriss, U.\BED, {\Bem Trends in Computational Social Choice}. AI
  Access Foundation.

\bibitem[\protect\BCAY{Elkind, Lackner,\ \BBA\ Peters}{Elkind
  et~al.}{2016}]{elkind2016preferencerestrictions}
Elkind, E., Lackner, M., \BBA\ Peters, D. \BBOP2016\BBCP.
\newblock \BBOQ Preference restrictions in computational social choice: Recent
  progress\BBCQ\
\newblock In {\Bem Proceedings of the 25th International Joint Conference on
  Artificial Intelligence (IJCAI)}, \BPGS\ 4062--4065.

\bibitem[\protect\BCAY{Escoffier, Spanjaard,\ \BBA\ Tydrichov{\'{a}}}{Escoffier
  et~al.}{2020}]{escoffier2020}
Escoffier, B., Spanjaard, O., \BBA\ Tydrichov{\'{a}}, M. \BBOP2020\BBCP.
\newblock \BBOQ Recognizing single-peaked preferences on an arbitrary graph:
  Complexity and algorithms\BBCQ\
\newblock In {\Bem Proceedings of the 13th International Symposium on
  Algorithmic Game Theory (SAGT)}.

\bibitem[\protect\BCAY{Escoffier, Lang,\ \BBA\ {\"O}zt{\"u}rk}{Escoffier
  et~al.}{2008}]{escoffier2008single}
Escoffier, B., Lang, J., \BBA\ {\"O}zt{\"u}rk, M. \BBOP2008\BBCP.
\newblock \BBOQ Single-peaked consistency and its complexity\BBCQ\
\newblock In {\Bem Proceedings of the 18th European Conference on Artificial
  Intelligence (ECAI)}, \BPGS\ 366--370.

\bibitem[\protect\BCAY{Faliszewski, Skowron, Slinko,\ \BBA\ Talmon}{Faliszewski
  et~al.}{2017}]{mw-survey}
Faliszewski, P., Skowron, P., Slinko, A., \BBA\ Talmon, N. \BBOP2017\BBCP.
\newblock \BBOQ Multiwinner voting: {A} new challenge for social choice
  theory\BBCQ\
\newblock In Endriss, U.\BED, {\Bem Trends in Computational Social Choice}. AI
  Access Foundation.

\bibitem[\protect\BCAY{Fitzsimmons\ \BBA\ Lackner}{Fitzsimmons\ \BBA\
  Lackner}{2020}]{fitzsimmons2020incomplete}
Fitzsimmons, Z.\BBACOMMA\  \BBA\ Lackner, M. \BBOP2020\BBCP.
\newblock \BBOQ Incomplete preferences in single-peaked electorates\BBCQ\
\newblock {\Bem Journal of Artificial Intelligence Research}, {\Bem 67},
  797--833.

\bibitem[\protect\BCAY{Garey\ \BBA\ Johnson}{Garey\ \BBA\ Johnson}{1979}]{gj}
Garey, M.~R.\BBACOMMA\  \BBA\ Johnson, D.~S. \BBOP1979\BBCP.
\newblock {\Bem Computers and Intractability: {A} Guide to the Theory of
  {NP}-Completeness}.
\newblock {W. H. Freeman and Company}.

\bibitem[\protect\BCAY{Godziszewski, Batko, Skowron,\ \BBA\
  Faliszewski}{Godziszewski et~al.}{2021}]{godziszewski2021analysis}
Godziszewski, M., Batko, P., Skowron, P., \BBA\ Faliszewski, P. \BBOP2021\BBCP.
\newblock \BBOQ An analysis of approval-based committee rules for
  {2D}-{Euclidean} elections\BBCQ\
\newblock In {\Bem Proceedings of the 35th AAAI Conference on Artificial
  Intelligence (AAAI)}, \BPGS\ 5448--5455.

\bibitem[\protect\BCAY{Gonzalez}{Gonzalez}{1985}]{Gonzalez85}
Gonzalez, T.~F. \BBOP1985\BBCP.
\newblock \BBOQ Clustering to minimize the maximum intercluster distance\BBCQ\
\newblock {\Bem Theoretical Computer Science}, {\Bem 38}, 293--306.

\bibitem[\protect\BCAY{Gottlob\ \BBA\ Greco}{Gottlob\ \BBA\
  Greco}{2013}]{gottlob2013decomposing}
Gottlob, G.\BBACOMMA\  \BBA\ Greco, G. \BBOP2013\BBCP.
\newblock \BBOQ Decomposing combinatorial auctions and set packing
  problems\BBCQ\
\newblock {\Bem Journal of the ACM}, {\Bem 60\/}(4), 1--39.

\bibitem[\protect\BCAY{Guo\ \BBA\ Niedermeier}{Guo\ \BBA\
  Niedermeier}{2006}]{guo2006exact}
Guo, J.\BBACOMMA\  \BBA\ Niedermeier, R. \BBOP2006\BBCP.
\newblock \BBOQ Exact algorithms and applications for tree-like weighted set
  cover\BBCQ\
\newblock {\Bem Journal of Discrete Algorithms}, {\Bem 4\/}(4), 608--622.

\bibitem[\protect\BCAY{Hemaspaandra, Hemaspaandra,\ \BBA\ Rothe}{Hemaspaandra
  et~al.}{1997}]{dodgson-hard2}
Hemaspaandra, E., Hemaspaandra, L.~A., \BBA\ Rothe, J. \BBOP1997\BBCP.
\newblock \BBOQ Exact analysis of {D}odgson elections: {L}ewis {C}arroll's 1876
  voting system is complete for parallel access to {NP}\BBCQ\
\newblock {\Bem Journal of ACM}, {\Bem 44\/}(6), 806--825.

\bibitem[\protect\BCAY{Hemaspaandra, Spakowski,\ \BBA\ Vogel}{Hemaspaandra
  et~al.}{2005}]{kemeny-hard}
Hemaspaandra, E., Spakowski, H., \BBA\ Vogel, J. \BBOP2005\BBCP.
\newblock \BBOQ The complexity of {Kemeny} elections\BBCQ\
\newblock {\Bem Theoretical Computer Science}, {\Bem 349\/}(3), 382--391.

\bibitem[\protect\BCAY{Karp}{Karp}{1972}]{karp1972reducibility}
Karp, R.~M. \BBOP1972\BBCP.
\newblock \BBOQ Reducibility among combinatorial problems\BBCQ\
\newblock In {\Bem Complexity of Computer Computations}, \BPGS\ 85--103.
  Springer.

\bibitem[\protect\BCAY{Kung}{Kung}{2015}]{kung2015sorting}
Kung, F.-C. \BBOP2015\BBCP.
\newblock \BBOQ Sorting out single-crossing preferences on networks\BBCQ\
\newblock {\Bem Social Choice and Welfare}, {\Bem 44\/}(3), 663--672.

\bibitem[\protect\BCAY{Lu\ \BBA\ Boutilier}{Lu\ \BBA\
  Boutilier}{2011}]{mw-hard2}
Lu, T.\BBACOMMA\  \BBA\ Boutilier, C. \BBOP2011\BBCP.
\newblock \BBOQ Budgeted social choice: {F}rom consensus to personalized
  decision making\BBCQ\
\newblock In {\Bem Proceedings of the 22nd International Joint Conference on
  Artificial Intelligence (IJCAI)}, \BPGS\ 280--286.

\bibitem[\protect\BCAY{Mirrlees}{Mirrlees}{1971}]{mir:j:sc}
Mirrlees, J. \BBOP1971\BBCP.
\newblock \BBOQ An exploration in the theory of optimal income taxation\BBCQ\
\newblock {\Bem Review of Economic Studies}, {\Bem 38}, 175--208.

\bibitem[\protect\BCAY{Monroe}{Monroe}{1995}]{monroe}
Monroe, B. \BBOP1995\BBCP.
\newblock \BBOQ Fully proportional representation\BBCQ\
\newblock {\Bem American Political Science Review}, {\Bem 89\/}(4), 925--940.

\bibitem[\protect\BCAY{Moulin}{Moulin}{1980}]{moulin1980strategy}
Moulin, H. \BBOP1980\BBCP.
\newblock \BBOQ On strategy-proofness and single peakedness\BBCQ\
\newblock {\Bem Public Choice}, {\Bem 35\/}(4), 437--455.

\bibitem[\protect\BCAY{Munagala, Shen,\ \BBA\ Wang}{Munagala
  et~al.}{2021}]{munagala2021optimal}
Munagala, K., Shen, Z., \BBA\ Wang, K. \BBOP2021\BBCP.
\newblock \BBOQ Optimal algorithms for multiwinner elections and the
  {Chamberlin}--{Courant} rule\BBCQ\
\newblock In {\Bem Proceedings of the 22nd ACM Conference on Economics and
  Computation (ACM EC)}, \BPGS\ 697--717.

\bibitem[\protect\BCAY{Peters\ \BBA\ Elkind}{Peters\ \BBA\ Elkind}{2016}]{pe16}
Peters, D.\BBACOMMA\  \BBA\ Elkind, E. \BBOP2016\BBCP.
\newblock \BBOQ Preferences single-peaked on nice trees\BBCQ\
\newblock In {\Bem Proceedings of the 30th AAAI Conference on Artificial
  Intelligence (AAAI)}, \BPGS\ 594--600.

\bibitem[\protect\BCAY{Peters\ \BBA\ Lackner}{Peters\ \BBA\
  Lackner}{2020}]{spoc}
Peters, D.\BBACOMMA\  \BBA\ Lackner, M. \BBOP2020\BBCP.
\newblock \BBOQ Preferences single-peaked on a circle\BBCQ\
\newblock {\Bem Journal of Artificial Intelligence Research}, {\Bem 68},
  463--502.

\bibitem[\protect\BCAY{Peters, Roy,\ \BBA\ Sadhukhan}{Peters
  et~al.}{2019}]{peters2019unanimous}
Peters, H., Roy, S., \BBA\ Sadhukhan, S. \BBOP2019\BBCP.
\newblock \BBOQ Unanimous and strategy-proof probabilistic rules for
  single-peaked preference profiles on graphs\BBCQ\
\newblock Working Paper.

\bibitem[\protect\BCAY{Procaccia, Rosenschein,\ \BBA\ Zohar}{Procaccia
  et~al.}{2008}]{mw-hard1}
Procaccia, A.~D., Rosenschein, J., \BBA\ Zohar, A. \BBOP2008\BBCP.
\newblock \BBOQ On the complexity of achieving proportional
  representation\BBCQ\
\newblock {\Bem Social Choice and Welfare}, {\Bem 30\/}(3), 353--362.

\bibitem[\protect\BCAY{Proskurowski\ \BBA\ Telle}{Proskurowski\ \BBA\
  Telle}{1999}]{proskurowski1999classes}
Proskurowski, A.\BBACOMMA\  \BBA\ Telle, J.~A. \BBOP1999\BBCP.
\newblock \BBOQ Classes of graphs with restricted interval models\BBCQ\
\newblock {\Bem Discrete Mathematics \& Theoretical Computer Science}, {\Bem
  3\/}(4), 167--176.

\bibitem[\protect\BCAY{Puppe\ \BBA\ Slinko}{Puppe\ \BBA\
  Slinko}{2019}]{puppe2019condorcet}
Puppe, C.\BBACOMMA\  \BBA\ Slinko, A. \BBOP2019\BBCP.
\newblock \BBOQ Condorcet domains, median graphs and the single-crossing
  property\BBCQ\
\newblock {\Bem Economic Theory}, {\Bem 67\/}(1), 285--318.

\bibitem[\protect\BCAY{Roberts}{Roberts}{1977}]{rob:j:tax}
Roberts, K. \BBOP1977\BBCP.
\newblock \BBOQ Voting over income tax schedules\BBCQ\
\newblock {\Bem Journal of Public Economics}, {\Bem 8\/}(3), 329--340.

\bibitem[\protect\BCAY{Rothe, Spakowski,\ \BBA\ Vogel}{Rothe
  et~al.}{2003}]{young-hard}
Rothe, J., Spakowski, H., \BBA\ Vogel, J. \BBOP2003\BBCP.
\newblock \BBOQ Exact complexity of the winner problem for {Young}
  elections\BBCQ\
\newblock {\Bem Theory of Computing Systems}, {\Bem 36\/}(4), 375--386.

\bibitem[\protect\BCAY{Scheffler}{Scheffler}{1990}]{scheffler1990linear}
Scheffler, P. \BBOP1990\BBCP.
\newblock \BBOQ A linear algorithm for the pathwidth of trees\BBCQ\
\newblock In {\Bem Topics in combinatorics and graph theory}, \BPGS\ 613--620.
  Springer.

\bibitem[\protect\BCAY{Schummer\ \BBA\ Vohra}{Schummer\ \BBA\
  Vohra}{2002}]{schummer2002strategy}
Schummer, J.\BBACOMMA\  \BBA\ Vohra, R.~V. \BBOP2002\BBCP.
\newblock \BBOQ Strategy-proof location on a network\BBCQ\
\newblock {\Bem Journal of Economic Theory}, {\Bem 104\/}(2), 405--428.

\bibitem[\protect\BCAY{Sheng~Bao\ \BBA\ Zhang}{Sheng~Bao\ \BBA\
  Zhang}{2012}]{sheng2012review}
Sheng~Bao, F.\BBACOMMA\  \BBA\ Zhang, Y. \BBOP2012\BBCP.
\newblock \BBOQ A review of tree convex sets test\BBCQ\
\newblock {\Bem Computational Intelligence}, {\Bem 28\/}(3), 358--372.

\bibitem[\protect\BCAY{Skowron, Faliszewski,\ \BBA\ Slinko}{Skowron
  et~al.}{2015a}]{mw-approx}
Skowron, P., Faliszewski, P., \BBA\ Slinko, A. \BBOP2015a\BBCP.
\newblock \BBOQ Achieving fully proportional representation: Approximability
  results\BBCQ\
\newblock {\Bem Artificial Intelligence}, {\Bem 222}, 67--103.

\bibitem[\protect\BCAY{Skowron, Yu, Faliszewski,\ \BBA\ Elkind}{Skowron
  et~al.}{2015b}]{skowron2015complexity}
Skowron, P., Yu, L., Faliszewski, P., \BBA\ Elkind, E. \BBOP2015b\BBCP.
\newblock \BBOQ The complexity of fully proportional representation for
  single-crossing electorates\BBCQ\
\newblock {\Bem Theoretical Computer Science}, {\Bem 569}, 43--57.

\bibitem[\protect\BCAY{Sliwinski\ \BBA\ Elkind}{Sliwinski\ \BBA\
  Elkind}{2019}]{sliwinski2019preferences}
Sliwinski, J.\BBACOMMA\  \BBA\ Elkind, E. \BBOP2019\BBCP.
\newblock \BBOQ Preferences single-peaked on a tree: Sampling and tree
  recognition\BBCQ\
\newblock In {\Bem Proceedings of the 28th International Joint Conference on
  Artificial Intelligence (IJCAI)}, \BPGS\ 580--586.

\bibitem[\protect\BCAY{Tarjan\ \BBA\ Yannakakis}{Tarjan\ \BBA\
  Yannakakis}{1984}]{tarjan1984simple}
Tarjan, R.~E.\BBACOMMA\  \BBA\ Yannakakis, M. \BBOP1984\BBCP.
\newblock \BBOQ Simple linear-time algorithms to test chordality of graphs,
  test acyclicity of hypergraphs, and selectively reduce acyclic
  hypergraphs\BBCQ\
\newblock {\Bem SIAM Journal on Computing}, {\Bem 13\/}(3), 566--579.

\bibitem[\protect\BCAY{Trick}{Trick}{1989a}]{trick1988induced}
Trick, M.~A. \BBOP1989a\BBCP.
\newblock \BBOQ Induced subtrees of a tree and the set packing problem\BBCQ\
\newblock \BTR, IMA Preprint Series \#377, University of Minnesota,
  https://conservancy.umn.edu/bitstream/handle/11299/4749/377.pdf.

\bibitem[\protect\BCAY{Trick}{Trick}{1989b}]{trick1989recognizing}
Trick, M.~A. \BBOP1989b\BBCP.
\newblock \BBOQ Recognizing single-peaked preferences on a tree\BBCQ\
\newblock {\Bem Mathematical Social Sciences}, {\Bem 17\/}(3), 329--334.

\bibitem[\protect\BCAY{Yu, Chan,\ \BBA\ Elkind}{Yu
  et~al.}{2013}]{yu2013multiwinner}
Yu, L., Chan, H., \BBA\ Elkind, E. \BBOP2013\BBCP.
\newblock \BBOQ Multiwinner elections under preferences that are single-peaked
  on a tree\BBCQ\
\newblock In {\Bem Proceedings of the 23rd International Joint Conference on
  Artificial Intelligence (IJCAI)}, \BPGS\ 425--431.

\end{thebibliography}
\bibliographystyle{theapa}

\end{document}